\documentclass[11pt,reqno]{amsart}
\usepackage{amsmath}
\usepackage{amssymb}
\usepackage{verbatim}
\usepackage{subfigure}
\usepackage{epsfig}
\usepackage{wavelet}

\newtheorem{lemma}{Lemma}
\newtheorem{prop}[lemma]{Proposition}
\newtheorem{cor}[lemma]{Corollary}
\newtheorem{theorem}[lemma]{Theorem}

\newtheorem{example}{Example}

\newtheorem{algorithm}{Algorithm}

\numberwithin{equation}{section}
\numberwithin{figure}{section}

\newcommand{\tpctf}{\operatorname{TP-\mathbb{C}TF}}
\newcommand{\pc}{\mathsf{c}}

\begin{document}

\title{Compactly Supported Tensor Product Complex Tight Framelets with Directionality}

\author{Bin Han}

\address{Department of Mathematical and Statistical Sciences,
University of Alberta, Edmonton,\quad Alberta, Canada T6G 2G1.
\quad {\tt bhan@ualberta.ca}\quad {\tt zzhao7@ualberta.ca}\quad
{\tt http://www.ualberta.ca/$\sim$bhan}
}

\author{Qun Mo}

\address{Department of Mathematics, Zhejiang University,
Hangzhou 310027, P. R. China. \quad {\tt moqun@zju.edu.cn}
}

\author{Zhenpeng Zhao}

\thanks{Research of B. Han and Z. Zhao supported in part by NSERC Canada under Grant RGP 228051. Research of Q. Mo supported in part by the NSF of China under Grants 10971189 and 11271010, and by the fundamental research funds for the Central Universities.}

\makeatletter \@addtoreset{equation}{section} \makeatother
\begin{abstract}
Although tensor product real-valued wavelets have been successfully applied to many high-dimensional problems, they can only capture well edge singularities along the coordinate axis directions. As an alternative and improvement of tensor product real-valued wavelets and dual tree complex wavelet transform, recently tensor product complex tight framelets with increasing directionality have been introduced in \cite{Han:MMNP:2013} and applied to image denoising in \cite{HanZhao}. Despite several desirable properties, the directional tensor product complex tight framelets constructed in \cite{Han:MMNP:2013,HanZhao} are bandlimited and do not have compact support in the space/time domain. Since compactly supported wavelets and framelets are of great interest and importance in both theory and application, it remains as an unsolved problem whether there exist compactly supported tensor product complex tight framelets with directionality. In this paper, we shall satisfactorily answer this question by proving a theoretical result on directionality of tight framelets and by introducing an algorithm to construct compactly supported complex tight framelets with directionality.
Our examples show that compactly supported complex tight framelets with directionality can be easily derived from any given eligible low-pass filters and refinable functions.
Several examples of compactly supported tensor product complex tight framelets with directionality have been presented.
\end{abstract}


 \maketitle

\pagenumbering{arabic}

\section{Introduction and Motivations}

Having better directionality and employing tensor product of a correlated pair of one-dimensional orthogonal wavelet filter banks, dual tree complex wavelet transform in \cite{K2001,SBK} has shown superior performance in applications over the commonly adopted tensor product real-valued wavelets. As alternatives and improvements to dual tree complex wavelet transform,
tensor product complex tight framelets with directionality have been introduced in \cite{Han:MMNP:2013}. It has been demonstrated in \cite{HanZhao} that tensor product complex tight framelets with improved directionality significantly perform better, in terms of PSNR (peak signal-to-noise ratio), than dual tree complex wavelet transform in the model problem of image denoising. However, the tensor product complex tight framelets constructed in \cite{Han:MMNP:2013,HanZhao} are only bandlimited, that is, they have compact support in the frequency domain but they are not compactly supported in the space/time domain. Since compactly supported wavelets and framelets are of importance in both theory and application due to their good space-frequency localization and computational efficiency desired in many applications, it is a natural and important problem for us to investigate compactly supported tensor product complex tight framelets with directionality.
In this paper we shall satisfactorily resolve this problem by studying and constructing compactly supported tensor product complex tight framelets in $\dLp{2}$ with directionality.

To explain our motivations, let us first introduce some notation and definitions.
For a function $f: \dR \rightarrow \C$ and a $d\times d$ real-valued matrix $U$, we shall adopt the following notation:
\[
f_{U; k}(x):=\lb U; k\rb f(x):=|\det(U)|^{1/2} f(Ux-k), \qquad x,k\in \dR.
\]
For $\phi, \psi^1, \ldots, \psi^s\in \dLp{2}$ with $s\in \N$, we define an affine system generated by $\phi, \psi^1, \ldots, \psi^s$ as follows:
\[
\AS_0(\phi; \psi^1, \ldots, \psi^s):=\{\phi(\cdot-k) \setsp k\in \dZ\} \cup \{ \psi^\ell_{2^jI_d; k} \setsp j\in \N \cup\{0\}, k\in \dZ, \ell=1, \ldots, s\},
\]
where $I_d$ denotes the $d\times d$ identity matrix. Recall that $\{\phi; \psi^1, \ldots, \psi^s\}$ is a ($d$-dimensional dyadic) tight framelet in $\dLp{2}$ if $\AS_0(\phi; \psi^1, \ldots, \psi^s)$ is a (normalized) tight frame for $\dLp{2}$, that is,
\be \label{def:tf}
\|f\|_{\dLp{2}}^2=\sum_{k\in \dZ} |\la f, \phi(\cdot-k)\ra|^2+\sum_{j=0}^\infty \sum_{\ell=1}^s \sum_{k\in \dZ} |\la f, \psi^\ell_{2^j I_d; k}\ra|^2, \qquad \forall\; f\in \dLp{2}.
\ee
In particular, if $\AS_0(\phi; \psi^1, \ldots, \psi^s)$ is an orthonormal basis for $\dLp{2}$, then we call $\{\phi; \psi^1, \ldots, \psi^s\}$ an orthogonal wavelet in $\dLp{2}$.
If $\{\phi; \psi^1, \ldots, \psi^s\}$ is a tight framelet in $\dLp{2}$, then $\{\psi^1, \ldots, \psi^s\}$ must be a homogeneous tight framelet in $\dLp{2}$ (see \cite{Daub:book,Han:acha:2012}), in other words,
\[
\|f\|_{\dLp{2}}^2=\sum_{j\in \Z} \sum_{\ell=1}^s \sum_{k\in \dZ} |\la f, \psi^\ell_{2^j I_d; k}\ra|^2, \qquad \forall\; f\in \dLp{2}.
\]
Consequently, it follows directly from \eqref{def:tf} and the above identity that
every function $f\in \dLp{2}$ has the following representations:
\be \label{f:tf}
f=\sum_{k\in \dZ} \la f, \phi(\cdot-k)\ra \phi(\cdot-k)+\sum_{j=0}^\infty \sum_{\ell=1}^s \sum_{k\in \dZ} \la f, \psi^\ell_{2^jI_d; k}\ra \psi^\ell_{2^jI_d; k}=
\sum_{j\in \Z} \sum_{\ell=1}^s \sum_{k\in \dZ} \la f, \psi^\ell_{2^jI_d; k}\ra \psi^\ell_{2^jI_d; k}
\ee
with the series converging unconditionally in $\dLp{2}$.

Due to many desirable properties such as sparsity and good space-frequency localization, wavelet representations in \eqref{f:tf} have been used in many applications (\cite{CRSS,Daub:book,K2001,SBK,Shen}). 
To have a fast algorithm to compute the wavelet coefficients in \eqref{f:tf}, wavelets and framelets are often derived from refinable functions and filter banks.
By $\dlp{2}$ we denote the space of all complex-valued sequences $u=\{u(k)\}_{k\in \dZ}: \dZ\rightarrow \C$ such that $\|u\|_{\dlp{2}}:=(\sum_{k\in \dZ} |u(k)|^2)^{1/2}<\infty$. The Fourier series (or symbol) of a sequence $u\in \dlp{2}$ is defined to be $\wh{u}(\xi):=\sum_{k\in \dZ} u(k) e^{-ik\cdot \xi}, \xi\in \dR$, which is a $2\pi\dZ$-periodic measurable function in $\dTLp{2}$ such that $\|\wh{u}\|_{\dTLp{2}}^2:=\frac{1}{(2\pi)^d} \int_{[-\pi,\pi)^d} |\wh{u}(\xi)|^2 d\xi=\|u\|_{\dlp{2}}^2=\sum_{k\in \dZ} |u(k)|^2<\infty$.

By $\dlp{0}$ we denote the set of all finitely supported sequences on $\dZ$.
If $a\in \dlp{0}$ and $\wh{a}(0)=1$, then $\prod_{j=1}^\infty \wh{a}(2^{-j}\xi)$ is convergent for every $\xi\in \dR$ and it is well known (\cite{Daub:book}) that there exists a compactly supported distribution $\phi$ on $\dR$ such that $\wh{\phi}(\xi)=\prod_{j=1}^\infty \wh{a}(2^{-j}\xi), \xi\in \dR$, where the Fourier transform is defined to be $\wh{f}(\xi):=\int_{\dR} f(x) e^{-ix\cdot \xi} dx$ for $f\in \dLp{1}$.
For $b_1, \ldots, b_s\in \dlp{2}$, we define $\psi^1, \ldots, \psi^s$ by
\[
\wh{\psi^\ell}(\xi):=\wh{b_\ell}(\xi/2)\wh{\phi}(\xi/2), \qquad \xi\in \dR,\; \ell=1,\ldots, s.
\]
Then $\{\phi; \psi^1, \ldots, \psi^s\}$ is a tight framelet in $\dLp{2}$ if and only if $\{a; b_1, \ldots, b_s\}$ is a tight framelet filter bank satisfying
\begin{equation} \label{tffb}
|\wh{a}(\xi)|^2+\sum_{\ell=1}^s |\wh{b_\ell}(\xi)|^2=1
\quad \mbox{and} \quad
\wh{a}(\xi)\ol{\wh{a}(\xi+\pi \omega)}+\sum_{\ell=1}^s \wh{b_\ell}(\xi) \ol{\wh{b_\ell}(\xi+\pi \omega)}=0, \qquad \forall\; \omega\in \Omega\bs \{0\}
\end{equation}
for almost every $\xi\in \dR$, where $\Omega:=[0,1]^d \cap \dZ$. See \cite{CRSS,CHS,Daub:book,DGM,DHRS,Han:frame,Han:acha:2012,Han:acha:2013,
HGS,HanMo:SIMAA,MoZhuang:laa,RonShen:twf,S:acha,Shen} and many references therein on tight framelets in $\dLp{2}$ and their applications.

High-dimensional wavelets and framelets are often obtained from one-dimensional wavelets and framelets through tensor product. The main advantages of tensor product wavelets and framelets lie in that they have a simple fast numerical algorithm and the construction of one-dimensional wavelets and framelets is often relatively easy.
To our best knowledge, almost all successful wavelet-based methods in applications have used tensor product real-valued wavelets and framelets, partially due to their simplicity and fast implementation. To illustrate the tensor product method, for simplicity, let us only discuss the particular case of dimension two here. For two one-dimensional functions $f, g: \R \rightarrow \C$, their tensor product $f\otimes g$ in dimension two is defined to be
$(f\otimes g)(x,y):=f(x) g(y)$, $x,y\in \R$.
Similarly, for two sequences $u,v : \Z \rightarrow \C$, their two-dimensional tensor product filter $u\otimes v$ is defined to be
$(u\otimes v)(j,k):=u(j) v(k)$, $j,k\in \Z$.
Let $\{\phi; \psi^1, \ldots, \psi^s\}$ be a tight framelet in $\Lp{2}$ with an underlying tight framelet filter bank $\{a; b_1, \ldots, b_s\}$ such that $\wh{\phi}(2\xi)=\wh{a}(\xi)\wh{\phi}(\xi)$ and $\wh{\psi^\ell}(2\xi)=\wh{b_\ell}(\xi)\wh{\phi}(\xi)$, $\ell=1, \ldots, s$.
Using tensor product, we obtain a tight framelet $\{\phi; \psi^1, \ldots, \psi^s\} \otimes \{\phi; \psi^1, \ldots, \psi^s\}$ in $L_2(\R^2)$ with an underlying tensor product tight framelet filter bank
$\{a; b_1, \ldots, b_s\}\otimes \{a; b_1, \ldots, b_s\}$ for dimension two.
More precisely, define
\[
\Psi:=\{ \phi\otimes \psi^1, \ldots, \phi\otimes \psi^s\}\cup \{\psi^1\otimes \phi, \ldots, \psi^s\otimes \phi\}\cup\{\psi^\ell \otimes \psi^m \setsp \ell, m=1, \ldots, s\},
\]
then we have a two-dimensional tight frame $\AS_0(\phi\otimes \phi; \Psi)$ for $L_2(\R^2)$ satisfying
\[
\|f\|_{L_2(\R^2)}^2=\sum_{k\in \Z^2} |\la f, (\phi\otimes\phi)(\cdot-k)\ra|^2+
\sum_{j=0}^\infty \sum_{\psi\in \Psi} \sum_{k\in \Z^2} |\la f, \psi_{2^jI_2; k}\ra|^2, \qquad \forall\; f\in L_2(\R^2).
\]
Moreover, $\phi\otimes \phi$ satisfies the tensor product refinement equation $\wh{\phi\otimes\phi}(2\xi)=\wh{a\otimes a}(\xi)\wh{\phi\otimes \phi}(\xi)$, a.e. $\xi\in \R^2$ and for each $\psi\in \Psi$, $\wh{\psi}(2\xi)=\wh{b_\psi}(\xi)\wh{\phi\otimes \phi}(\xi)$, where
\[
\{ b_\psi \setsp \psi\in \Psi\}:=\{a\otimes b_1, \ldots, a\otimes b_s\} \cup \{b_1\otimes a, \ldots, b_s\otimes a\}\cup\{b_\ell\otimes b_m \setsp \ell, m=1, \ldots, s\}.
\]
Note that $\{a\otimes a; b_\psi, \psi\in \Psi\}$ is a two-dimensional tensor product tight framelet filter bank.

Though tensor product real-valued wavelets and framelets have been widely used in many applications, they have some shortcomings, for example, lack of directionality in high dimensions.
For two-dimensional data such as images, edge singularities are ubiquitous and play a more fundamental role in image processing than point singularities. As a consequence, tensor product real-valued wavelets are only suboptimal since they can only efficiently capture edge singularities along the coordinate axis directions. For the convenience of the reader, let us explain this point in more detail. When $\phi$ and $\psi^1, \ldots, \psi^s$ are real-valued functions in $\Lp{2}$ such that $\wh{\phi}(0)=1$ and $\wh{\psi^1}(0)=\cdots=\wh{\psi^s}(0)=0$, in general $\wh{\phi}$ concentrates essentially near the origin while $\wh{\psi^1}, \ldots, \wh{\psi^s}$ concentrate largely outside a neighborhood of the origin. Since every $\psi^\ell, \ell=1, \ldots, s$ is real-valued, it is trivial to notice that $\ol{\wh{\psi^\ell}(\xi)}=\wh{\psi^\ell}(-\xi)$ and consequently, the magnitude of the frequency spectrum of $\wh{\psi^\ell}$ is symmetric about the origin.
For dimension two, it is not difficult to see that all $\phi\otimes \psi^\ell$ have horizontal direction while all $\psi^\ell \otimes \phi$ have vertical direction for $\ell=1, \ldots, s$. However, all $\psi^\ell\otimes \psi^{m}$ do not exhibit any directionality for $\ell,m =1, \ldots, s$. The same phenomenon can be said for the associated tight framelet filter bank: all $a\otimes b_\ell$ exhibit horizontal direction, all $b_\ell\otimes a$ exhibit vertical direction, but $b_\ell\otimes b_m$ do not exhibit any directionality for $\ell, m=1, \ldots, s$.
To see this point better, let us look at the particular example of the Haar orthogonal wavelet $\{\phi; \psi\}$ with $\phi=\chi_{[0,1]}$ and $\psi=\chi_{[0, \frac{1}{2}]}-\chi_{[\frac{1}{2},1]}$. Then
$\phi\otimes \phi=\chi_{[0,1]^2}$ and
\[
\phi\otimes \psi=\chi_{[0,1]\times [0, \frac{1}{2}]}-\chi_{[0,1]\times [\frac{1}{2},1]}, \quad
\psi\otimes \phi=\chi_{[0, \frac{1}{2}]\times [0,1]}-\chi_{[\frac{1}{2},1]\times [0,1]}, \quad
\psi\otimes \psi=\chi_{[0,\frac{1}{2}]^2\cup[\frac{1}{2},1]^2}-
\chi_{[0,\frac{1}{2}]\times [\frac{1}{2},1]\cup [\frac{1}{2},1]\times [0, \frac{1}{2}]}.
\]
We can clearly observe that $\phi\otimes \psi$ has horizontal
direction, $\psi\otimes \phi$ has vertical direction, but $\psi\otimes \psi$ does not exhibit any directionality. Note that the above Haar orthogonal wavelet $\{\phi; \psi\}$ has the underlying orthogonal wavelet filter bank $\{a; b\}$ with $a=\{\frac{1}{2}, \frac{1}{2}\}_{[0,1]}$ and $b=\{\frac{1}{2}, -\frac{1}{2}\}_{[0,1]}$. Then
\[
a\otimes a=\left[\begin{matrix} \tfrac{1}{4} &\tfrac{1}{4} \smallskip \\
\tfrac{1}{4} &\tfrac{1}{4}\end{matrix}\right]_{[0,1]^2}, \quad
a\otimes b=\left[\begin{matrix} -\tfrac{1}{4} &-\tfrac{1}{4}\smallskip \\
\tfrac{1}{4} &\tfrac{1}{4}\end{matrix}\right]_{[0,1]^2}, \quad
b\otimes a=\left[\begin{matrix} \tfrac{1}{4} &-\tfrac{1}{4} \smallskip \\
\tfrac{1}{4} &-\tfrac{1}{4}\end{matrix}\right]_{[0,1]^2}, \quad
b\otimes b=\left[\begin{matrix} -\tfrac{1}{4} &\tfrac{1}{4}\smallskip \\
\tfrac{1}{4} &-\tfrac{1}{4}\end{matrix}\right]_{[0,1]^2}.
\]
From above, we observe that $a\otimes b$ has horizontal direction, $b\otimes a$ has vertical direction, but $b\otimes b$ does not exhibit any directionality.

As one of the most popular and successful approaches to enhance the performance of tensor product real-valued wavelets, the dual tree complex wavelet transform proposed in \cite{K2001,SBK} uses tensor product of a correlated pair of finitely supported orthogonal wavelet filter banks and offers $6$ directions with impressive performance in many applications. However, horizontal and vertical directions are very common in many two-dimensional data such as images.
It is also difficult to generalize the approach of dual tree complex wavelet transform to have more than $6$ directions by using dyadic orthogonal wavelet filter banks.
To further improve the performance of and to provide alternatives to dual tree complex wavelet transform, recently \cite{Han:MMNP:2013} introduced tensor product complex tight framelets with increasing directionality.
Tensor product complex tight framelets not only offer alternatives to dual tree complex wavelet transform but also have improved directionality. In \cite{Han:MMNP:2013,HanZhao}, a family of tensor product complex tight framelets has been constructed in the frequency domain and their
performance for image denoising has been reported in \cite{HanZhao}.
With more directions and using the tensor product structure, the bandlimited tensor product complex tight framelets constructed in \cite{Han:MMNP:2013,HanZhao} indeed significantly perform better than dual tree complex wavelet transform in the area of image denoising. See \cite{HanZhao,K2001,SBK} and many references therein on dual tree complex wavelet transform, and see \cite{Han:acha:2012,Han:MMNP:2013,HanZhao} for more details on directional complex tight framelets.

This paper is largely motivated by the approach
introduced in \cite{Han:MMNP:2013} using tensor product complex tight framelet filter banks. Let us recall here the tensor product tight framelet filter banks constructed in the frequency domain in \cite{Han:MMNP:2013,HanZhao}.
Let $\pP_{m}(x):=(1-x)^m \sum_{j=0}^{m-1} \binom{m+j-1}{j} x^j$ with $m\in \N$.
Then $\pP_{m}$ satisfies
the identity $\pP_{m}(x)+\pP_{m}(1-x)=1$ (see \cite{Daub:book}).
For $c_L<c_R$ and two positive numbers $\gep_L, \gep_R$ satisfying $\gep_L+\gep_R\le c_R-c_L$, we define a bump function $\chi_{[c_L, c_R]; \gep_L, \gep_R}$ on $\R$ by
\be \label{bump:func}
\chi_{[c_L, c_R]; \gep_L, \gep_R}(\xi):=
\begin{cases} 0, \quad &\xi\le c_L-\gep_L \; \mbox{or}\; \xi \ge c_R+\gep_R, \\
\sin\big(\tfrac{\pi}{2}\pP_{m}(\tfrac{c_L+\gep_L-\xi}{2\gep_L})\big), \quad &c_L-\gep_L<\xi<c_L+\gep_L,\\
1, \quad &c_L+\gep_L\le \xi\le c_R-\gep_R,\\
\sin\big(\tfrac{\pi}{2}\pP_{m}(\tfrac{\xi-c_R+\gep_R}{2\gep_R})\big),
\quad &c_R-\gep_R<\xi<c_R+\gep_R.
\end{cases}
\ee
For simplicity of discussion and presentation, here we only recall a special type of tensor product complex tight framelet filter banks constructed in \cite{Han:MMNP:2013,HanZhao}.
We define a real-valued symmetric low-pass filter $a\in \lp{2}$ and two complex-valued high-pass filters $b^{p}, b^{n}\in \lp{2}$ by
\be \label{tpctf:ab}
\wh{a}:=\chi_{[-c, c]; \gep, \gep} \qquad \mbox{and}\qquad
\wh{b^{p}}:=\chi_{[c, \pi]; \gep, \gep}, \qquad
\wh{b^{n}}:=\chi_{[-\pi,-c]; \gep,\gep},
\ee
where $c$ and $\gep$ are positive numbers satisfying
$0<\gep \le \min(\tfrac{c}{2}, \tfrac{\pi}{2}-c)$.
Then it is easy to directly check that $\{a; b^{p}, b^{n}\}$ is a one-dimensional tight framelet filter bank such that $a$ is real-valued and symmetric about the origin with $\wh{a}(0)=1$.
Define functions $\phi, \psi^{p}, \psi^{n}$ on $\R$ by
\[
\wh{\phi}(\xi):=\prod_{j=1}^\infty \wh{a}(2^{-j}\xi), \qquad
\wh{\psi^{p}}(\xi):=\wh{b^{p}}(\xi/2)\wh{\phi}(\xi/2), \quad
\wh{\psi^{n}}(\xi):=\wh{b^{n}}(\xi/2)\wh{\phi}(\xi/2), \qquad\xi\in \R.
\]
Then $\{\phi; \psi^{p}, \psi^{n}\}$ is a tight framelet in $\Lp{2}$. Note that all the functions $\phi$, $\psi^{p}$, $\psi^{n}$ are bandlimited, that is, their Fourier transforms have compact support.
Moreover, $\phi$ is real-valued and symmetric about the origin.
Note that $\wh{b^n}(\xi)=\ol{\wh{b^p}(-\xi)}$ and therefore, we have $b^n=\ol{b^p}$ and $\psi^n=\ol{\psi^p}$.
More importantly, both functions $\psi^{p}, \psi^{n}$ are complex-valued and enjoy the following frequency separation property:
\be \label{psi:direction}
\wh{\psi^{p}}(\xi)=0, \qquad \forall\; \xi\in (-\infty,0]\qquad  \mbox{and}\qquad \wh{\psi^{n}}(\xi)=0, \qquad \forall\; \xi\in [0, \infty).
\ee
In other words, the frequency spectrum of $\psi^{p}$ vanishes on the negative interval $(-\infty, 0]$ and concentrates only inside the positive interval $[0,\infty)$, while the frequency spectrum of $\psi^{n}$ vanishes on the positive interval $[0,\infty)$ and concentrates only inside the negative interval $(-\infty,0]$.
The property in \eqref{psi:direction} is the key ingredient to produce directionality for tensor product complex tight framelets in high dimensions. To see this point, let us look at the two-dimensional tensor product tight framelets $\{\phi; \psi^{p}, \psi^{n}\} \otimes
\{\phi; \psi^{p}, \psi^{n}\}$, that is,
\be \label{2D:tpctf}
\{\phi\otimes \phi\}\cup \{ \phi\otimes \psi^p, \phi\otimes \psi^n, \psi^p\otimes \phi, \psi^n \otimes \phi, \psi^p\otimes \psi^p, \psi^p\otimes \psi^n, \psi^n\otimes \psi^p, \psi^n\otimes \psi^n\}.
\ee
By \eqref{psi:direction}, for $f, g\in \{\psi^{p}, \psi^{n}\}$, we see that $\wh{f\otimes g}=\wh{f}\otimes \wh{g}$ concentrates on a small rectangle away from the origin. As a consequence, both the real and imaginary parts of $f\otimes g$ exhibit good directions.
We now provide the detail here. For a complex-valued function $f: \dR \rightarrow \C$, we define
\[
f^{[r]}(x):=\re(f(x)) \qquad \mbox{and}\qquad f^{[i]}(x):=\im(f(x)), \qquad x\in \dR.
\]
That is, $f=f^{[r]}+if^{[i]}$ with both $f^{[r]}$ and $f^{[i]}$ being real-valued functions on $\dR$.
Similarly, for $u: \dZ\rightarrow \C$, we can write $u=u^{[r]}+iu^{[i]}$ with both sequences $u^{[r]}$ and $u^{[i]}$ having real coefficients.
Define
\[
\psi^{p, [r]}:=\re(\psi^{p}), \quad \psi^{p, [i]}:=\im(\psi^{p}),\quad
\psi^{n, [r]}:=\re(\psi^{n}), \quad \psi^{n, [i]}:=\im(\psi^{n})
\]
and similarly
\[
b^{p, [r]}:=\re(b^{p}), \quad b^{p, [i]}:=\im(b^{p}),\quad
b^{n, [r]}:=\re(b^{n}), \quad b^{n, [i]}:=\im(b^{n}).
\]
Then all the above functions and filters are real-valued. It is trivial to check that
\be \label{tp:real}
\{\phi; \psi^{p,[r]}, \psi^{n,[r]}, \psi^{p,[i]}, \psi^{n,[i]}\}
\ee
is a real-valued tight framelet in $\Lp{2}$ with the underlying real-valued tight framelet filter bank $\{a; b^{p, [r]}$, $b^{n,[r]}$, $b^{p, [i]}, b^{n, [i]}\}$.
However, we do not apply tensor product to this real-valued one-dimensional tight framelet since it shares the same shortcoming as tensor product real-valued wavelets or framelets.
Instead, we take tensor product of the one-dimensional complex tight framelet first for dimension two as in \eqref{2D:tpctf}, then we separate their real and imaginary parts to derive a real-valued tight framelet in $L_2(\R^2)$.
%
%
%
%
If in addition $b^n=\ol{b^p}$ and consequently, $\psi^n=\ol{\psi^p}$ since $\phi$ is real-valued, then $\{\phi; \sqrt{2} \psi^{p,[r]}, \sqrt{2}\psi^{p, [i]}\}$ is a real-valued tight framelet in $\Lp{2}$ with the underlying tight framelet filter bank $\{a; \sqrt{2} b^{p,[r]}, \sqrt{2}b^{p, [i]}\}$. Moreover,
%
\be \label{tp:ctf}
\begin{split}
&\sqrt{2} \bigl\{ \tfrac{\sqrt{2}}{2} \phi\otimes \phi; \phi \otimes \psi^{p,[r]},
\phi\otimes \psi^{p, [i]}, \psi^{p,[r]}\otimes \phi, \psi^{p,[i]}\otimes \phi, \psi^{p,[r]}\otimes \psi^{p,[r]}-\psi^{p,[i]}\otimes \psi^{p,[i]}, \\
&\qquad \psi^{p,[r]}\otimes \psi^{p,[r]}+\psi^{p,[i]}\otimes \psi^{p,[i]},
\psi^{p,[r]}\otimes \psi^{p,[i]}-\psi^{p,[i]}\otimes \psi^{p,[r]}, \psi^{p,[r]}\otimes \psi^{p,[i]}+\psi^{p,[i]}\otimes \psi^{p,[r]}\bigl\}
\end{split}
\ee
is a two-dimensional real-valued tight framelet in $L_2(\R^2)$ with the following underlying two-dimensional real-valued tight framelet filter bank
\be \label{tp:ctf:fb}
\begin{split}
&\sqrt{2} \bigl\{\tfrac{\sqrt{2}}{2} a\otimes a; a\otimes b^{p,[r]},
a\otimes b^{p,[i]}, b^{p,[r]}\otimes a, b^{p,[i]}\otimes a, b^{p,[r]}\otimes b^{p,[r]}-b^{p,[i]}\otimes b^{p,[i]},\\
 &\qquad b^{p,[r]}\otimes b^{p,[r]}+b^{p,[i]}\otimes b^{p,[i]},
b^{p,[r]}\otimes b^{p,[i]}-b^{p,[i]}\otimes b^{p,[r]},
b^{p,[r]}\otimes b^{p,[i]}-b^{p,[i]}\otimes b^{p,[r]}
\bigl\}.
\end{split}
\ee
Now one can check that the derived two-dimensional real-valued tight framelet exhibits four directions:
\begin{enumerate}
\item $\phi\otimes \psi^{p,[r]}$ and
$\phi\otimes \psi^{p,[i]}$ have horizontal direction along  $0^\circ$;

\item $\psi^{p,[r]}\otimes \phi$ and $\psi^{p,[i]}\otimes \phi$ have vertical direction along $90^\circ$;

\item $\psi^{p,[r]}\otimes \psi^{p,[r]}-\psi^{p,[i]}\otimes \psi^{p,[i]}$ and $\psi^{p,[r]}\otimes \psi^{p,[r]}+\psi^{p,[i]}\otimes \psi^{p,[i]}$ have direction along $45^\circ$;

\item $\psi^{p,[r]}\otimes \psi^{p,[i]}-\psi^{p,[i]}\otimes \psi^{p,[r]}$ and $\psi^{p,[r]}\otimes \psi^{p,[i]}+\psi^{p,[i]}\otimes \psi^{p,[r]}$
have direction along  $-45^\circ$.
\end{enumerate}

As discussed in \cite{Han:MMNP:2013,HanZhao}, more directions can be achieved by using more high-pass filters. For simplicity, we only discuss the particular case $\{\phi; \psi^p, \psi^n\}$ in this paper, which plays a critical role for obtaining general finitely supported tensor product complex tight framelets with increasing directionality.

Although the derived two-dimensional real-valued tight framelet in \eqref{tp:ctf} and its underlying real-valued tight framelet filter bank in \eqref{tp:ctf:fb} no longer have the tensor product structure,
it is not difficult to see that they can be obtained through a simple transform using a constant unitary matrix from $\{\phi; \sqrt{2} \psi^{p,[r]}, \sqrt{2} \psi^{p,[i]}\}
\otimes \{\phi; \sqrt{2} \psi^{p,[r]}, \sqrt{2} \psi^{p,[i]}\}$ and its underlying real-valued tight framelet filter bank $\{a; \sqrt{2} b^{p,[r]}, \sqrt{2} b^{p,[i]}\}
\otimes \{a; \sqrt{2} b^{p,[r]}$, $\sqrt{2} b^{p,[i]}\}$.
Therefore, similar to dual tree complex wavelet transform in \cite{K2001,SBK}, the algorithm using the tensor product complex tight framelets in \eqref{tp:ctf} with their filter banks in \eqref{tp:ctf:fb} can be implemented using the tensor product discrete framelet transform employing the tight framelet filter bank $\{a; \sqrt{2} b^{p,[r]}, \sqrt{2} b^{p,[i]}\}$, followed by simple linear combinations of the wavelet/framelet coefficients.

However, the filters $a, b^p, b^n$ constructed in \eqref{tpctf:ab} (see \cite{Han:MMNP:2013,HanZhao} for more detail) have infinite support in the time domain. Since compactly supported wavelets and framelets have great interest and importance in both theory and application, this naturally leads us to ask the following question:

\begin{enumerate}
\item[Q1:]
Is it possible to construct compactly supported one-dimensional complex tight framelets $\{\phi; \psi^p, \psi^n\}$ with finitely supported tight framelet filter banks $\{a; b^p, b^n\}$ such that $\wh{\psi^p}$ almost vanishes on the negative interval $(-\infty, 0]$ and $\wh{\psi^n}$ almost vanishes on the positive interval $[0, \infty)$?
\end{enumerate}

By $\wh{\psi^p}(2\xi)=\wh{b^p}(\xi)\wh{\phi}(\xi)$ and $\wh{\psi^n}(2\xi)=\wh{b^n}(\xi)\wh{\phi}(\xi)$, since generally $\wh{\phi}\approx \chi_{[-\pi, \pi]}$, to satisfy the condition in \eqref{psi:direction},
it is very natural to require that
$\wh{b^p}$ should be relatively small on the negative interval $[-\pi,0)$ so that $\wh{b^p}$ concentrates largely on the positive interval $[0,\pi)$, while $\wh{b^n}$ should be relatively small on the positive interval $[0, \pi)$ so that $\wh{b^n}$ concentrates largely on the negative interval $[-\pi,0)$. In other words, to achieve directionality for tensor product tight framelets, the two high-pass filters $b^p$ and $b^n$ must have good frequency separation property. A natural quantity to measure the quality of frequency separation (and therefore, the directionality of tensor product tight framelets) is
\be \label{B:bpn}
B_{b^p, b^n}(\xi):=|\wh{b^p}(\xi+\pi)|^2+|\wh{b^n}(\xi)|^2, \qquad \xi\in [0,\pi].
\ee
That is, the smaller the quantity $B_{b^p, b^n}$ on the interval $[0,\pi]$, the better the frequency separation of the two high-pass filters $b^p$ and $b^n$ in the frequency domain and consequently, the better the directionality of their associated high-dimensional tensor product tight framelets. More precisely, if we can construct a tight framelet filter bank $\{a; b^p, b^n\}$ such that the quantity $B_{b^p, b^n}(\xi)$ is relatively small for all $\xi\in [0,\pi]$, then the high-pass filters $b^p$ and $b^n$ have good frequency separation and thus, the resulting tensor product tight framelet filter bank $\{a; b^p, b^n\}\otimes \{a; b^p, b^n\}$ and its associated real-valued tight framelet by separating real and imaginary parts in $\{\phi; \psi^p, \psi^n\}\otimes \{\phi; \psi^p, \psi^n\}$ will have four directions:  $0^\circ$ (horizontal), $\pm 45^\circ$, and $90^\circ$ (vertical) in dimension two.

In addition to Q1, we are interested in the following two problems:

\begin{enumerate}
\item[Q2:] For filters $a, b^p, b^n\in \lp{2}$ such that $\{a; b^p, b^n\}$ is a tight framelet filter bank, can we achieve $B_{b^p, b^n}(\xi)\approx 0$ for all $\xi\in [0, \pi]$? More precisely, given a low-pass filter $a\in \lp{2}$, we want to find a sharp theoretical lower bound which is a function $A: [0,\pi]\rightarrow [0, \infty)$ depending only on the given filter $a$ such that
    (i) $|\wh{b^p}(\xi+\pi)|^2+|\wh{b^n}(\xi)|^2\ge A(\xi)$ a.e. $\xi\in [0,\pi]$ for any tight framelet filter bank $\{a; b^p, b^n\}$. (ii) There exists a tight framelet filter bank $\{a; \mathring{b}^p, \mathring{b^n}\}$ derived from the filter $a$ such that
    $|\wh{\mathring{b}^p}(\xi+\pi)|^2+|\wh{\mathring{b}^n}(\xi)|^2=A(\xi)$ a.e. $\xi\in [0,\pi]$.

\item[Q3:] From every given real-valued low-pass filter $a\in \lp{0}$ such that $1-|\wh{a}(\xi)|^2-|\wh{a}(\xi+\pi)|^2\ge 0$ (which is a necessary condition for building a tight framelet filter bank derived from $a$), can we construct a finitely supported tight framelet filter bank $\{a; b^p, b^n\}$ such that its associated tensor product complex tight framelet exhibit almost best possible directionality? More precisely, is it possible to construct finitely supported high-pass filters $b^p, b^n$ such that $\{a; b^p, b^n\}$ is a tight framelet filter bank and $|\wh{b^p}(\xi+\pi)|^2+|\wh{b^n}(\xi)|^2\approx A(\xi)$ on $[0,\pi]$? Here $A$ is the sharp theoretical lower bound for frequency separation in Q2.
\end{enumerate}

We shall satisfactorily and positively answer all the above questions in this paper. We shall provide a sharp theoretical lower bound for frequency separation using the natural quantity $|\wh{b^p}(\xi+\pi)|^2+|\wh{b^n}(\xi)|^2$. More precisely, we shall prove in Section~2 the following result which completely answers Q2:

\begin{theorem}\label{thm:lowerbound}
Let $a, b^p, b^n\in \lp{2}$ such that $\{a; b^p, b^n\}$ is a tight framelet filter bank. Then
\be \label{lowerbound}
|\wh{b^p}(\xi+\pi)|^2+|\wh{b^n}(\xi)|^2 \ge A(\xi), \qquad a.e.\; \xi\in [0,\pi],
\ee
where the frequency separation function $A$ associated with the filter $a$ is defined to be
\be \label{A}
A(\xi):=\frac{2-|\wh{a}(\xi)|^2-|\wh{a}(\xi+\pi)|^2-
\sqrt{4\big(1-|\wh{a}(\xi)|^2-|\wh{a}(\xi+\pi)|^2\big)
+(|\wh{a}(\xi)|^2-|\wh{a}(\xi+\pi)|^2)^2}}{2}.
\ee
Moreover, the inequality in \eqref{lowerbound} is sharp in the sense that
there exist $\mathring{b}^p, \mathring{b}^n \in \lp{2}$ such that $\{a; \mathring{b}^p, \mathring{b}^n\}$ is a tight framelet filter bank satisfying $|\wh{\mathring{b}^p}(\xi+\pi)|^2+|\wh{\mathring{b}^n}(\xi)|^2 =A(\xi)$ a.e. $\xi\in [0,\pi]$.
If in addition the filter $a$ is real-valued, that is, $\wh{a}(\xi)=\ol{\wh{a}(-\xi)}$ a.e. $\xi\in \R$, then the tight framelet filter bank $\{a; \mathring{b}^p, \mathring{b}^n\}$ can satisfy the additional property: $\wh{\mathring{b}^n}(\xi)=\ol{\wh{\mathring{b}^p}(-\xi)}$ a.e. $\xi\in \R$, that is, $\mathring{b}^n=\ol{\mathring{b}^p}$.
\end{theorem}

Interestingly, as demonstrated by the following result, the frequency separation function $A$ in \eqref{A} is often very small for most known low-pass filters in the literature.

\begin{theorem}\label{thm:A:Bspline}
Let $A$ be the frequency separation function defined in \eqref{A} associated with a filter $a\in \lp{2}$ satisfying $|\wh{a}(\xi)|^2+|\wh{a}(\xi+\pi)|^2\le 1$ for almost every $\xi\in \R$. Then
\be \label{est:A}
0\le A(\xi)\le \min(|\wh{a}(\xi)|^2, |\wh{a}(\xi+\pi)|^2), \qquad a.e.\;\xi\in \R.
\ee
In particular,
\begin{enumerate}
\item[(i)] $A(\xi)=0$ a.e. $\xi\in [0,\pi]$ if and only if $\wh{a}(\xi)\wh{a}(\xi+\pi)=0$ a.e. $\xi\in \R$.
\item[(ii)] $A(\xi)= \min(|\wh{a}(\xi)|^2, |\wh{a}(\xi+\pi)|^2)$ a.e. $\xi\in [0,\pi]$ if and only if $|\wh{a}(\xi)|^2+|\wh{a}(\xi+\pi)|^2=1$ for almost every $\xi\in \R$ satisfying $\min(|\wh{a}(\xi)|^2, |\wh{a}(\xi+\pi)|^2)\ne 0$. In particular, if $|\wh{a}(\xi)|^2+|\wh{a}(\xi+\pi)|^2=1$ a.e. $\xi\in \R$ (that is, $a$ is an orthogonal filter), then $A(\xi)= \min(|\wh{a}(\xi)|^2, |\wh{a}(\xi+\pi)|^2)$ a.e. $\xi\in [0,\pi]$.
\item[(iii)] If $a$ is the B-spline filter $a^B_m$ of order $m$ given by $\wh{a^B_m}(\xi):=\cos^{2m}(\xi/2)$ with $m\in \N$, then
\be \label{A:aBm}
4^{-m} \sin^m(\xi) \le A(\xi)\le 4^{1-m} \sin^m(\xi), \qquad \forall\; \xi \in [0,\pi].
\ee
\end{enumerate}
\end{theorem}

To answer Q1 and Q3 and to construct tight framelet filter banks with directionality, in Section~3 we shall investigate the structure of all finitely supported tight framelet filter banks $\{a; b^p, b^n\}$ derived from a given filter $a$.
More precisely, from any given finitely supported filter $a\in \lp{0}$, in Theorem~\ref{thm:tffb} and Algorithm~\ref{alg:tffb} we shall construct all possible finitely supported tight framelet filter banks $\{a; b^p, b^n\}$ derived from a given low-pass filter $a$.
For prescribed filter lengths of $b^p$ and $b^n$, such a result enables us to find the best possible complex tight framelet filter bank $\{a; b^p, b^n\}$ having the best possible frequency separation, that is, having the smallest possible $\int_0^\pi \big[|\wh{b^p}(\xi+\pi)|^2+|\wh{b^n}(\xi)|^2\big] d\xi$.
Finally, in Section~4 we shall provide an algorithm for constructing
finitely supported complex tight framelet filter banks $\{a; b^p, b^n\}$ having the smallest possible $\int_0^\pi \big[ |\wh{b^p}(\xi+\pi)|^2+|\wh{b^n}(\xi)|^2\big] d\xi$ among all high-pass filters $b^p$ and $b^n$ with prescribed filter supports.
Several examples will be presented to illustrate the results and algorithms in this paper.

This paper mainly concentrates on the construction of a particular family of finitely supported tensor product complex tight framelet filter banks with directionality (more precisely, in the terminology of \cite{HanZhao}, $\tpctf_3$ having four directions in dimension two). We shall leave the construction of general finitely supported tensor product complex tight framelet filter banks with increasing directionality (that is, tensor product complex tight framelets $\tpctf_n$ with $n\ge 4$) and
their possible applications as a future work.

\section{A Sharp Lower Bound for Directionality of Tight Framelet Filter Banks}

In this section, we shall prove the sharp theoretical lower bound stated in Theorem~\ref{thm:lowerbound} for the best possible frequency separation of a tight framelet filter bank $\{a; b^p, b^n\}$ derived from a given low-pass filter $a$. Then we shall prove Theorem~\ref{thm:A:Bspline} showing that the frequency separation function $A$ in \eqref{A} is often small for many known low-pass filters. As a contrast to the result in Theorem~\ref{thm:lowerbound} for complex-valued tight framelet filter banks, at the end of this section we provide a result showing that all real-valued tight framelet filter banks cannot have good frequency separation.

\begin{proof}[Proof of Theorem~\ref{thm:lowerbound}]
Since $\{a; b^p, b^n\}$ is a tight framelet filter bank, it follows from the definition in \eqref{tffb} with $d=1$ that
\be \label{def:tffb}
\left[ \begin{matrix} \wh{b^p}(\xi) &\wh{b^n}(\xi)\\
\wh{b^p}(\xi+\pi) &\wh{b^n}(\xi+\pi)\end{matrix} \right]
\left[ \begin{matrix} \ol{\wh{b^p}(\xi)} &\ol{\wh{b^p}(\xi+\pi)}\\
\ol{\wh{b^n}(\xi)} &\ol{\wh{b^n}(\xi+\pi)} \end{matrix} \right]=
\left[ \begin{matrix} 1-|\wh{a}(\xi)|^2 &-\wh{a}(\xi)\ol{\wh{a}(\xi+\pi)}\\
-\wh{a}(\xi+\pi)\ol{\wh{a}(\xi)} &1-|\wh{a}(\xi+\pi)|^2\end{matrix} \right].
\ee
Since the determinant of the matrix on the right-hand side of \eqref{def:tffb} is $1-|\wh{a}(\xi)|^2-|\wh{a}(\xi+\pi)|^2$,
it follows directly from \eqref{def:tffb} that we must have $1-|\wh{a}(\xi)|^2-|\wh{a}(\xi+\pi)|^2\ge 0$ a.e. $\xi\in \R$.

We also notice from \eqref{def:tffb} that $\{a; b^p, b^n\}$ is a tight framelet filter bank if and only if for almost every $\xi\in [0,\pi]$, the following three equations hold:
\begin{align}
&|\wh{a}(\xi)|^2+|\wh{b^p}(\xi)|^2+|\wh{b^n}(\xi)|^2=1, \label{eq1}\\
&|\wh{a}(\xi+\pi)|^2+|\wh{b^p}(\xi+\pi)|^2+|\wh{b^n}(\xi+\pi)|^2=1, \label{eq2}\\
&\wh{a}(\xi)\ol{\wh{a}(\xi+\pi)}+\wh{b^p}(\xi)\ol{\wh{b^p}(\xi+\pi)}+\wh{b^n}(\xi)
\ol{\wh{b^n}(\xi+\pi)}=0. \label{eq3}
\end{align}
In the rest of the proof, we always assume $\xi\in [0,\pi]$. Note that \eqref{eq1} and \eqref{eq2} imply
\be \label{eq1:1}
|\wh{b^p}(\xi)|=\sqrt{1-|\wh{a}(\xi)|^2-|\wh{b^n}(\xi)|^2},\qquad
|\wh{b^n}(\xi+\pi)|= \sqrt{1-|\wh{a}(\xi+\pi)|^2-|\wh{b^p}(\xi+\pi)|^2}.
\ee
Using \eqref{eq1:1}, we deduce from \eqref{eq3} that
\begin{align*}
|\wh{a}(\xi)\wh{a}(\xi+\pi)|^2 &\le \Big(
|\wh{b^p}(\xi){\wh{b^p}(\xi+\pi)}|+|\wh{b^n}(\xi){\wh{b^n}(\xi+\pi)}|\Big)^2\\
&=\Big(
|\wh{b^p}(\xi+\pi)|\sqrt{1-|\wh{a}(\xi)|^2-|\wh{b^n}(\xi)|^2}+|\wh{b^n}(\xi)|
\sqrt{1-|\wh{a}(\xi+\pi)|^2-|\wh{b^p}(\xi+\pi)|^2}\Big)^2\\
&\le \Big(|\wh{b^p}(\xi+\pi)|^2+|\wh{b^n}(\xi)|^2\Big) \Big( 2-|\wh{a}(\xi)|^2-|\wh{a}(\xi+\pi)|^2-( |\wh{b^p}(\xi+\pi)|^2+|\wh{b^n}(\xi)|^2)\Big),
\end{align*}
where we used Cauchy-Schwarz inequality in the last inequality.
Define $B(\xi):=|\wh{b^p}(\xi+\pi)|^2+|\wh{b^n}(\xi)|^2$. Then the above inequality can be rewritten as
\be \label{f:func}
f(B(\xi)) \ge 0 \quad \mbox{with}\quad
f(x):=-x^2+\big(2-|\wh{a}(\xi)|^2-|\wh{a}(\xi+\pi)|^2\big)x-|\wh{a}(\xi)\wh{a}(\xi+\pi)|^2.
\ee
Since $f$ is a polynomial of degree two, by calculation, we see that $f$ has two real roots:
\[
A(\xi) \quad \mbox{and}\quad 2-|\wh{a}(\xi)|^2-|\wh{a}(\xi+\pi)|^2-A(\xi),
\]
where $A$ is defined in \eqref{A}.
Rewrite $A(\xi)$ in \eqref{A} as
\be \label{A:1}
A(\xi)=\frac{2-|\wh{a}(\xi)|^2-|\wh{a}(\xi+\pi)|^2-\sqrt{C(\xi)}}{2}
\ee
with
\be \label{C}
C(\xi):=4\big(1-|\wh{a}(\xi)|^2
-|\wh{a}(\xi+\pi)|^2\big)+(|\wh{a}(\xi)|^2-|\wh{a}(\xi+\pi)|^2)^2.
\ee
Note that we can also rewrite the function $C(\xi)$ as follows:
\be \label{C:2}
C(\xi)=(2-|\wh{a}(\xi)|^2-|\wh{a}(\xi+\pi)|^2)^2-4|\wh{a}(\xi)\wh{a}(\xi+\pi)|^2
\le (2-|\wh{a}(\xi)|^2-|\wh{a}(\xi+\pi)|^2)^2.
\ee
From the expression of $A$ in \eqref{A:1} and the above inequality, we see that $A(\xi)\ge 0$ and
\be \label{A>0}
0\le A(\xi)\le 2-|\wh{a}(\xi)|^2-|\wh{a}(\xi+\pi)|^2-A(\xi).
\ee
In particular, we see that $f(x)>0$ if and only if $A(\xi)<x< 2-|\wh{a}(\xi)|^2-|\wh{a}(\xi+\pi)|^2-A(\xi)$. Therefore, since $f(x)<0$ for all $x<A(\xi)$, we conclude from $f(B(\xi))\ge 0$ that $B(\xi)\ge A(\xi)$. Thus, we proved inequality \eqref{lowerbound}.

We now show that the inequality in \eqref{lowerbound} is sharp by
explicitly constructing a tight framelet filter bank $\{a; \mathring{b}^p, \mathring{b}^n\}$ satisfying $|\wh{\mathring{b}^p}(\xi+\pi)|^2
+|\wh{\mathring{b}^n}(\xi)|^2=A(\xi)$ for all $\xi\in [0,\pi]$.
In the following, we shall construct such $2\pi$-periodic measurable functions $\wh{\mathring{b}^p}$ and $\wh{\mathring{b}^n}$ by defining $\wh{\mathring{b}^p}(\xi),\wh{\mathring{b}^p}(\xi+\pi), \wh{\mathring{b}^n}(\xi), \wh{\mathring{b}^n}(\xi+\pi)$ on the interval $\xi\in [0,\pi]$.

For $\xi\in [0,\pi]$, we define
\be \label{mbp}
\wh{\mathring{b}^p}(\xi+\pi)=\begin{cases}
\frac{1}{2}, &\text{if $C(\xi)=0$,}\\
\sqrt{ \frac{1}{2}A(\xi)\left(1-\frac{|\wh{a}(\xi)|^2
-|\wh{a}(\xi+\pi)|^2}{\sqrt{C(\xi)}}\right)}, &\text{otherwise}
\end{cases}
\ee
and
\be \label{mbn}
\wh{\mathring{b}^n}(\xi)=\begin{cases}
\frac{1}{2}, &\text{if $C(\xi)=0$,}\\
\sqrt{ \frac{1}{2}A(\xi)\left(1+\frac{|\wh{a}(\xi)|^2-|\wh{a}(\xi+\pi)|^2}
{\sqrt{C(\xi)}}\right)}, &\text{otherwise.}
\end{cases}
\ee
We first show that both $\wh{\mathring{b}^p}(\xi+\pi)$ and $\wh{\mathring{b}^n}(\xi)$ are well defined nonnegative functions for $\xi\in [0,\pi]$. By the definition of $C(\xi)$ in \eqref{C}, it is straightforward to see that
$\sqrt{C(\xi)}\ge \Big| |\wh{a}(\xi)|^2-|\wh{a}(\xi+\pi)|^2\Big|$ for $\xi\in [0,\pi]$.
Consequently, we have
\[
\left| \frac{ |\wh{a}(\xi)|^2- |\wh{a}(\xi+\pi)|^2}{\sqrt{C(\xi)}}\right| \le 1.
\]
Since $A(\xi)\ge 0$, we now see that both $\wh{\mathring{b}^p}(\xi+\pi)$ in \eqref{mbp} and $\wh{\mathring{b}^n}(\xi)$ in \eqref{mbn} are well defined nonnegative functions for $\xi\in [0,\pi]$.
Let $\beta(\xi)$ denote the phase of $\wh{a}(\xi)\ol{\wh{a}(\xi+\pi)}$, that is, $\beta$ is a real-valued measurable function on $[0,\pi]$ such that
\be \label{beta}
\wh{a}(\xi)\ol{\wh{a}(\xi+\pi)}=e^{i\beta(\xi)} |\wh{a}(\xi)\ol{\wh{a}(\xi+\pi)}|, \qquad \xi\in [0,\pi].
\ee
If $\wh{a}(\xi)\wh{a}(\xi+\pi)=0$, then we simply define $\beta(\xi)=0$.
For $\xi\in [0,\pi]$, we define
\be \label{mbp2}
\wh{\mathring{b}^p}(\xi)=-e^{i\beta(\xi)} \sqrt{1-|\wh{a}(\xi)|^2-|\wh{\mathring{b}^n}(\xi)|^2}
\ee
and
\be \label{mbn2}
\wh{\mathring{b}^n}(\xi+\pi)=-e^{-i\beta(\xi)} \sqrt{1-|\wh{a}(\xi+\pi)|^2-|\wh{\mathring{b}^p}(\xi+\pi)|^2}.
\ee
We now prove that $\wh{\mathring{b}^p}(\xi)$ and $\wh{\mathring{b}^n}(\xi+\pi)$ are well defined by proving that for $\xi\in [0,\pi]$,
\be \label{positive}
1-|\wh{a}(\xi)|^2-|\wh{\mathring{b}^n}(\xi)|^2 \ge 0 \quad \mbox{and} \quad 1-|\wh{a}(\xi+\pi)|^2-|\wh{\mathring{b}^p}(\xi+\pi)|^2 \ge 0
\ee
and
\be \label{magnitude}
|\wh{a}(\xi)\wh{a}(\xi+\pi)|=|\wh{\mathring{b}^p}(\xi) \wh{\mathring{b}^p}(\xi+\pi)|+|\wh{\mathring{b}^n}(\xi)\wh{\mathring{b}^n}(\xi+\pi)|.
\ee
We prove \eqref{positive} and \eqref{magnitude} by considering four cases.

Case 1: $C(\xi)=0$. Since $C(\xi)=0$, it follows from \eqref{mbp} and \eqref{mbn} that $\wh{\mathring{b}^p}(\xi+\pi)=\wh{\mathring{b}^n}(\xi)=\frac{1}{2}$.
By $C(\xi)=0$, it follows from the definition of $C(\xi)$ in \eqref{C} that $1-|\wh{a}(\xi)|^2-|\wh{a}(\xi+\pi)|^2=0$ and $|\wh{a}(\xi)|^2-|\wh{a}(\xi+\pi)|^2=0$. Hence, we must have $|\wh{a}(\xi)|^2=|\wh{a}(\xi+\pi)|^2=\frac{1}{2}$.
Consequently, $1-|\wh{a}(\xi)|^2-|\wh{\mathring{b}^n}(\xi)|^2=1-\frac{1}{2}-\frac{1}{4}=\frac{1}{4}\ge 0$ and
$1-|\wh{a}(\xi+\pi)|^2-|\wh{\mathring{b}^p}(\xi+\pi)|^2=1-\frac{1}{2}-\frac{1}{4}=\frac{1}{4}\ge 0$.
Hence, \eqref{positive} holds. Now by the definition of $\wh{\mathring{b}^p}(\xi)$ in \eqref{mbp2} and $\wh{\mathring{b}^n}(\xi+\pi)$ in \eqref{mbn2}, we have $\wh{\mathring{b}^p}(\xi)=
-e^{i\beta(\xi)}/2$ and
$\wh{\mathring{b}^n}(\xi+\pi)=-e^{-i\beta(\xi)}/2$.
Thus, it is trivial to check that \eqref{magnitude} holds.

Case 2: $C(\xi)\ne 0$ and $A(\xi)=0$. By the definition of  $\wh{\mathring{b}^p}(\xi+\pi)$ in \eqref{mbp} and $\wh{\mathring{b}^n}(\xi)$ in \eqref{mbn}, we have
$\wh{\mathring{b}^p}(\xi+\pi)=\wh{\mathring{b}^n}(\xi)=0$. Clearly, \eqref{positive} holds since $1-|\wh{a}(\xi)|^2-|\wh{a}(\xi+\pi)|^2\ge 0$. It is also easy to see that $A(\xi)=0$ implies $\wh{a}(\xi)\wh{a}(\xi+\pi)=0$. Therefore, \eqref{magnitude} is obviously true.

Case 3: $C(\xi)\ne 0$, $A(\xi)\ne 0$, and $|\wh{a}(\xi)|^2-|\wh{a}(\xi+\pi)|^2=\sqrt{C(\xi)}$ or $-\sqrt{C(\xi)}$. Without loss of any generality, we only consider $|\wh{a}(\xi)|^2-|\wh{a}(\xi+\pi)|^2=\sqrt{C(\xi)}$, from which we deduce that
\[
1-|\wh{a}(\xi)|^2-|\wh{a}(\xi+\pi)|^2=0, \qquad
\wh{\mathring{b}^p}(\xi+\pi)=0, \qquad \wh{\mathring{b}^n}(\xi)=\sqrt{A(\xi)}.
\]
It follows from $1-|\wh{a}(\xi)|^2-|\wh{a}(\xi+\pi)|^2=0$ and the definition of $A(\xi)$ in \eqref{A} that $A(\xi)=\frac{1-|\wh{a}(\xi)|^2+|\wh{a}(\xi+\pi)|^2}{2}
=|\wh{a}(\xi+\pi)|^2$.
Now we see that \eqref{positive} is satisfied, since $1-|\wh{a}(\xi+\pi)|^2-|\wh{\mathring{b}^p}(\xi+\pi)|^2=1-|\wh{a}(\xi+\pi)|^2=|\wh{a}(\xi)|^2\ge 0$ and
\[
1-|\wh{a}(\xi)|^2-|\wh{\mathring{b}^n}(\xi)|^2=
1-|\wh{a}(\xi)|^2-A(\xi)=1-|\wh{a}(\xi)|^2-|\wh{a}(\xi+\pi)|^2=0.
\]
Consequently, we deduce from the above identity and the definition of $\wh{\mathring{b}^p}(\xi)$ in \eqref{mbp2} that $\wh{\mathring{b}^p}(\xi)=0$. Since $\wh{\mathring{b}^p}(\xi+\pi)=0$ and $A(\xi)=|\wh{a}(\xi+\pi)|^2$, from the definition of $\wh{\mathring{b}^n}(\xi+\pi)$ in \eqref{mbn2} we deduce that
\[
|\wh{\mathring{b}^n}(\xi+\pi)|^2=1-|\wh{a}(\xi+\pi)|^2
-|\wh{\mathring{b}^p}(\xi+\pi)|^2=1-|\wh{a}(\xi+\pi)|^2=
|\wh{a}(\xi)|^2.
\]
Therefore, by $\wh{\mathring{b}^p}(\xi)=\wh{\mathring{b}^p}(\xi+\pi)=0$, $\wh{\mathring{b}^n}(\xi)=\sqrt{A(\xi)}$, and $|\wh{\mathring{b}^n}(\xi+\pi)|=|\wh{a}(\xi)|$, we see that
\[
|\wh{\mathring{b}^p}(\xi) \wh{\mathring{b}^p}(\xi+\pi)|+|\wh{\mathring{b}^n}(\xi)\wh{\mathring{b}^n}(\xi+\pi)|
=|\wh{\mathring{b}^n}(\xi) \wh{\mathring{b}^n}(\xi+\pi)|
=\sqrt{A(\xi)} |\wh{a}(\xi)|=|\wh{a}(\xi)\wh{a}(\xi+\pi)|,
\]
where we used the identity $A(\xi)=|\wh{a}(\xi+\pi)|^2$ in the last identity. Hence, \eqref{magnitude} holds.

Case 4: $C(\xi)\ne 0$, $A(\xi)\ne 0$, and $|\wh{a}(\xi)|^2-|\wh{a}(\xi+\pi)|^2\ne \pm \sqrt{C(\xi)}$. Note that the last two conditions imply that $\wh{\mathring{b}^p}(\xi+\pi) \ne 0$ and $\wh{\mathring{b}^n}(\xi)\ne 0$. From the definition of $\wh{\mathring{b}^p}(\xi+\pi)$ in \eqref{mbp} and $\wh{\mathring{b}^n}(\xi)$ in \eqref{mbn}, we see that
\be \label{eq:ratio}
\frac{|\wh{\mathring{b}^p}(\xi+\pi)|^2}{|\wh{\mathring{b}^n}(\xi)|^2}
=\frac{\sqrt{C(\xi)}-(|\wh{a}(\xi)|^2-|\wh{a}(\xi+\pi)|^2)}{
\sqrt{C(\xi)}+(|\wh{a}(\xi)|^2-|\wh{a}(\xi+\pi)|^2)}=
\frac{1-|\wh{a}(\xi)|^2-A(\xi)}{1-|\wh{a}(\xi+\pi)|^2-A(\xi)},
\ee
where we used the relation $\sqrt{C(\xi)}=2-|\wh{a}(\xi)|^2-|\wh{a}(\xi+\pi)|^2-2A(\xi)$ (derived from the definition of $A(\xi)$ in \eqref{A}) in the last identity.
Since  $C(\xi)\ne 0$, we deduce from the definition of  $\wh{\mathring{b}^p}(\xi+\pi)$ in \eqref{mbp} and $\wh{\mathring{b}^n}(\xi)$ in \eqref{mbn}
that $|\wh{\mathring{b}^p}(\xi+\pi)|^2+|\wh{\mathring{b}^n}(\xi)|^2
=A(\xi)$. Now it follows directly from \eqref{eq:ratio} that
\begin{align*}
\frac{|\wh{\mathring{b}^p}(\xi+\pi)|^2}{|\wh{\mathring{b}^n}(\xi)|^2}
&=\frac{1-|\wh{a}(\xi)|^2-A(\xi)}{1-|\wh{a}(\xi+\pi)|^2-A(\xi)}
=\frac{1-|\wh{a}(\xi)|^2-A(\xi)+|\wh{\mathring{b}^p}(\xi+\pi)|^2}
{1-|\wh{a}(\xi+\pi)|^2-A(\xi)+|\wh{\mathring{b}^n}(\xi)|^2}
=\frac{1-|\wh{a}(\xi)|^2-|\wh{\mathring{b}^n}(\xi)|^2}
{1-|\wh{a}(\xi+\pi)|^2-|\wh{\mathring{b}^p}(\xi+\pi)|^2}.
\end{align*}
That is, we proved
\be \label{bpbn}
\frac{|\wh{\mathring{b}^p}(\xi+\pi)|^2}{|\wh{\mathring{b}^n}(\xi)|^2}
=\frac{1-|\wh{a}(\xi)|^2-|\wh{\mathring{b}^n}(\xi)|^2}
{1-|\wh{a}(\xi+\pi)|^2-|\wh{\mathring{b}^p}(\xi+\pi)|^2}.
\ee
From the identity in \eqref{bpbn}, we further deduce that
\begin{align*}
\frac{|\wh{\mathring{b}^p}(\xi+\pi)|^2}{A(\xi)}
&=\frac{|\wh{\mathring{b}^p}(\xi+\pi)|^2}{|\wh{\mathring{b}^p}(\xi+\pi)|^2
+|\wh{\mathring{b}^n}(\xi)|^2}=
\frac{1-|\wh{a}(\xi)|^2-|\wh{\mathring{b}^n}(\xi)|^2}
{(1-|\wh{a}(\xi)|^2-|\wh{\mathring{b}^n}(\xi)|^2)+(1-|\wh{a}(\xi+\pi)|^2-|\wh{\mathring{b}^p}(\xi+\pi)|^2)}\\
&=\frac{1-|\wh{a}(\xi)|^2-|\wh{\mathring{b}^n}(\xi)|^2}
{2-|\wh{a}(\xi)|^2-|\wh{a}(\xi+\pi)|^2-A(\xi)}.
\end{align*}
In other words, we proved
\be \label{bp}
\frac{|\wh{\mathring{b}^p}(\xi+\pi)|^2}{A(\xi)}
=\frac{1-|\wh{a}(\xi)|^2-|\wh{\mathring{b}^n}(\xi)|^2}
{2-|\wh{a}(\xi)|^2-|\wh{a}(\xi+\pi)|^2-A(\xi)}.
\ee
Similarly, we can prove that
\be \label{bn}
\frac{|\wh{\mathring{b}^n}(\xi)|^2}{A(\xi)}
=\frac{1-|\wh{a}(\xi+\pi)|^2-|\wh{\mathring{b}^p}(\xi+\pi)|^2}
{2-|\wh{a}(\xi)|^2-|\wh{a}(\xi+\pi)|^2-A(\xi)}.
\ee
By our assumption $A(\xi)>0$, we see from \eqref{A>0} that $2-|\wh{a}(\xi)|^2-|\wh{a}(\xi+\pi)|^2-A(\xi)\ge A(\xi)>0$.
Since  $\wh{\mathring{b}^p}(\xi+\pi) \ne 0$ and $\wh{\mathring{b}^n}(\xi)\ne 0$, we deduce from \eqref{bp} that we must have $1-|\wh{a}(\xi)|^2-|\wh{\mathring{b}^n}(\xi)|^2>0$.
By the same argument, we deduce from \eqref{bn} that $1-|\wh{a}(\xi+\pi)|^2-|\wh{\mathring{b}^p}(\xi+\pi)|^2>0$. Hence, we proved \eqref{positive}. Therefore, $\wh{\mathring{b}^p}(\xi)$ and $\wh{\mathring{b}^n}(\xi+\pi)$ are well defined. It now follows from \eqref{bpbn} that
\[
\frac{|\wh{\mathring{b}^p}(\xi+\pi)|^2}{|\wh{\mathring{b}^n}(\xi)|^2}
=\frac{1-|\wh{a}(\xi)|^2-|\wh{\mathring{b}^n}(\xi)|^2}
{1-|\wh{a}(\xi+\pi)|^2-|\wh{\mathring{b}^p}(\xi+\pi)|^2}=
\frac{|\wh{\mathring{b}^p}(\xi)|^2}{|\wh{\mathring{b}^n}(\xi+\pi)|^2}
\]
from which we see that the vector $(|\wh{\mathring{b}^p}(\xi+\pi)|, |\wh{\mathring{b}^n}(\xi)|)$ is parallel to the vector
$(|\wh{\mathring{b}^p}(\xi)|, |\wh{\mathring{b}^n}(\xi+\pi)|)$.
Consequently, we must have
\[
|\wh{\mathring{b}^p}(\xi) \wh{\mathring{b}^p}(\xi+\pi)|+|\wh{\mathring{b}^n}(\xi)\wh{\mathring{b}^n}(\xi+\pi)|
=\sqrt{|\wh{\mathring{b}^p}(\xi+\pi)|^2+|\wh{\mathring{b}^n}(\xi)|^2}
\sqrt{|\wh{\mathring{b}^p}(\xi)|^2+|\wh{\mathring{b}^n}(\xi+\pi)|^2}.
\]
By the definition of $\wh{\mathring{b}^p}(\xi+\pi)$ in \eqref{mbp} and
$\wh{\mathring{b}^n}(\xi)$ in \eqref{mbn} and
by the definition of $\wh{\mathring{b}^p}(\xi)$ in \eqref{mbp2} and
$\wh{\mathring{b}^n}(\xi+\pi)$ in \eqref{mbn2},
we conclude that
\begin{align*}
|\wh{\mathring{b}^p}(\xi) \wh{\mathring{b}^p}(\xi+\pi)|+|\wh{\mathring{b}^n}(\xi)\wh{\mathring{b}^n}(\xi+\pi)|
&=\sqrt{|\wh{\mathring{b}^p}(\xi+\pi)|^2+|\wh{\mathring{b}^n}(\xi)|^2}
\sqrt{|\wh{\mathring{b}^p}(\xi)|^2+|\wh{\mathring{b}^n}(\xi+\pi)|^2}\\
&=\sqrt{A(\xi) (2-|\wh{a}(\xi)|^2-|\wh{a}(\xi+\pi)|^2-A(\xi))}
=|\wh{a}(\xi)\wh{a}(\xi+\pi)|,
\end{align*}
where in the last identity we used the fact that $A(\xi)$ and $2-|\wh{a}(\xi)|^2-|\wh{a}(\xi+\pi)|^2-A(\xi)$ are the two roots of $f$ in \eqref{f:func} and $f(0)=-|\wh{a}(\xi)\wh{a}(\xi+\pi)|^2$. Thus, we proved \eqref{magnitude}.

By our construction, it is now trivial to see that $|\wh{\mathring{b}^p}(\xi+\pi)|^2+|\wh{\mathring{b}^n}(\xi)|^2=A(\xi)$ for all $\xi\in [0,\pi]$ such that $C(\xi)\ne 0$. If $C(\xi)=0$, as discussed in Case 1, then we have $A(\xi)=\frac{1}{2}$ and we still have $|\wh{\mathring{b}^p}(\xi+\pi)|^2+|\wh{\mathring{b}^n}(\xi)|^2
=\frac{1}{4}+\frac{1}{4}=\frac{1}{2}=A(\xi)$.
To complete the proof, we now show that $\{a; \mathring{b}^p, \mathring{b}^n\}$ is a tight framelet filter bank. By our construction of $\wh{\mathring{b}^p}$ and $\wh{\mathring{b}^n}$, it is trivial to see that \eqref{eq1} and \eqref{eq2} are satisfied with $b^p$ and $b^n$ being replaced by  $\mathring{b}^p$ and $\mathring{b}^n$, respectively.
To check \eqref{eq3},  we have
\begin{align*}
\wh{a}(\xi)&\ol{\wh{a}(\xi+\pi)}+\wh{\mathring{b}^p}(\xi) \ol{\wh{\mathring{b}^p}(\xi+\pi)}+\wh{\mathring{b}^n}(\xi) \ol{\wh{\mathring{b}^n}(\xi+\pi)}\\
&=
e^{i\beta(\xi)} |\wh{a}(\xi)\wh{a}(\xi+\pi)|-e^{i\beta(\xi)} ( |\wh{\mathring{b}^p}(\xi) \wh{\mathring{b}^p}(\xi+\pi)|+
|\wh{\mathring{b}^n}(\xi) \wh{\mathring{b}^n}(\xi+\pi)|)
=0,
\end{align*}
where in the last identity we used \eqref{magnitude}.
Therefore, $\{a; \mathring{b}^p, \mathring{b}^n\}$ is indeed a tight framelet filter bank.

If the filter $a$ is real-valued, then $\wh{a}(\xi)=\ol{\wh{a}(-\xi)}$ a.e. $\xi \in \R$. Consequently, we have
\be \label{real:a}
|\wh{a}(-\xi)|=|\wh{a}(\xi)| \quad \mbox{and}\quad
C(-\xi)=C(\xi)=C(\pi-\xi), \quad A(-\xi)=A(\xi)=A(\pi-\xi).
\ee
We now prove that $\ol{\wh{\mathring{b}^p}(-\xi)}=\wh{\mathring{b}^n}(\xi)$ a.e. $\xi\in \R$, which is equivalent to verify that
\be \label{mbp:mbn:rel}
\ol{\wh{b^p}(-\xi)}=\wh{\mathring{b}^n}(\xi) \quad \mbox{and}\quad
\ol{\wh{b^p}(\pi-\xi)}=\wh{\mathring{b}^n}(\xi-\pi), \qquad a.e.\; \xi\in [0,\pi].
\ee
By \eqref{real:a} and the definition of $\wh{\mathring{b}^p}(\xi+\pi)$ in \eqref{mbp} and $\wh{\mathring{b}^n}(\xi)$ in \eqref{mbn}, we see that
\[
\ol{\wh{\mathring{b}^p}(-\xi)}=\ol{\wh{\mathring{b}^p}((\pi-\xi)+\pi)}
=\wh{\mathring{b}^p}((\pi-\xi)+\pi)=
\wh{\mathring{b}^n}(\xi), \qquad \xi\in [0,\pi],
\]
which is the first identity in \eqref{mbp:mbn:rel}. Similarly, we have
\begin{align*}
\ol{\wh{b^p}(\pi-\xi)}&=-e^{-i\beta(\pi-\xi)} \sqrt{1-|\wh{a}(\pi-\xi)|^2-|\wh{\mathring{b}^n}(\pi-\xi)|^2}\\
&=-e^{-i\beta(\pi-\xi)} \sqrt{1-|\wh{a}(\xi+\pi)|^2-|\wh{\mathring{b}^p}(\xi+\pi)|^2}
=e^{i(\beta(\xi)-\beta(\pi-\xi))} \wh{\mathring{b}^n}(\xi+\pi),
\end{align*}
where we used \eqref{mbn2} and the first identity in \eqref{mbp:mbn:rel}.
If we can prove that
\be \label{phase:real}
e^{i(\beta(\xi)-\beta(\pi-\xi))}=1,\qquad \xi\in [0,\pi],
\ee
then the second identity in \eqref{mbp:mbn:rel} holds and therefore, we proved $\ol{\wh{b^p}(-\xi)}=\wh{b^n}(\xi)$ a.e. $\xi\in \R$.

We now prove \eqref{phase:real}. Replacing $\xi$ by $\pi-\xi$ in the definition of $\beta(\xi)$ in \eqref{beta} and using \eqref{real:a}, we have
\[
\wh{a}(\pi-\xi)\ol{\wh{a}(2\pi-\xi)}=e^{i\beta(\xi-\pi)}
|\wh{a}(\pi-\xi)\wh{a}(2\pi-\xi)|=e^{i\beta(\xi-\pi)} |\wh{a}(\xi)\wh{a}(\xi+\pi)|.
\]
Since $\wh{a}(\xi)=\ol{\wh{a}(-\xi)}$, we have
\[
\wh{a}(\pi-\xi)\ol{\wh{a}(2\pi-\xi)}=\ol{\wh{a}(\xi-\pi)}\; \ol{\wh{a}(-\xi)}=\ol{\wh{a}(\xi+\pi)} \wh{a}(\xi)=
\wh{a}(\xi)\ol{\wh{a}(\xi+\pi)}.
\]
Consequently, comparing with \eqref{beta}, we conclude that for $\xi\in [0,\pi]$ such that $\wh{a}(\xi)\wh{a}(\xi+\pi)\ne 0$, we must have $e^{i\beta(\pi-\xi)}=e^{i\beta(\xi)}$, which is simply \eqref{phase:real}. For the case that $\wh{a}(\xi)\wh{a}(\xi+\pi)=0$, \eqref{phase:real} is trivially true since $\beta(\xi)=\beta(\pi-\xi)=0$. This completes the proof of Theorem~\ref{thm:lowerbound}.
\end{proof}

As demonstrated by Theorem~\ref{thm:A:Bspline},
the frequency separation function $A$ in \eqref{A} is often small for many known low-pass filters.

\begin{proof}[Proof of Theorem~\ref{thm:A:Bspline}]
Define $x:=|\wh{a}(\xi)|^2$ and $y:=|\wh{a}(\xi+\pi)|^2$. Then $0\le x, y\le 1$ and $0\le x+y\le 1$.
In terms of $x$ and $y$, the function $A(\xi)$ in \eqref{A} can be rewritten as
\be \label{pA}
A(\xi)=\tfrac{1}{2}\pA(x,y) \quad \mbox{with}\quad
\pA(x,y):=2-x-y-\sqrt{4(1-x-y)+(x-y)^2}.
\ee
By a simple direct calculation, we have
\be \label{est:A:detail}
\tfrac{1}{2}\pA(x,y)=x-\frac{4x(1-x-y)}{g(x,y)} \le x,
\ee
where
\[
g(x,y):=2-3x-y+\sqrt{4(1-x-y)+(x-y)^2}\ge 
2-3x-y+(x-y)=2(1-x-y)\ge 0.
\]
If $g(x,y)>0$, by the symmetry between $x$ and $y$ in $\pA(x,y)$, then it follows from \eqref{est:A:detail} that $A(\xi)=\frac{1}{2}\pA(x,y)\le \min(x,y)=\min(|\wh{a}(\xi)|^2, |\wh{a}(\xi+\pi)|^2)$.
Note that $g(x,y)=0$ if and only if $x+y=1$ and $x\ge y$. If $g(x,y)=0$, then we also have $A(\xi)=\frac{1}{2}\pA(x,y)=y=\min(x,y)=\min(|\wh{a}(\xi)|^2,|\wh{a}(\xi+\pi)|^2)$. Therefore, we proved the inequality \eqref{est:A}.

Item (i) follows directly from the definition of $A(\xi)$ and the relation in \eqref{C}.
Item (ii) follows directly from \eqref{est:A:detail}.
For item (iii), by the definition of the function $A$ in \eqref{A} with $a=a^B_m$, we have $A(\xi)=\tfrac{1}{2}\pA(x,y)$ and
$\sin^{2m}(\xi)=2^{2m} \sin^{2m}(\xi/2) \cos^{2m}(\xi/2)=4^m xy$.
%
%
Note that
\[
\pA(x,y)=(2-x-y)-\sqrt{(2-x-y)^2-4xy}
=\frac{4xy}{(2-x-y)+\sqrt{(2-x-y)^2-4xy}}.
\]
Since $0\le x,y \le 1$, we obviously have
$0\le \sqrt{(2-x-y)^2-4xy}\le 2-x-y$.
Therefore, we conclude that
$\frac{2xy}{2-x-y}\le \pA(x,y)\le \frac{4xy}{2-x-y}$.
Consequently, by $0\le x,y\le 1$ and $x+y\le 1$, we deduce that
\[
xy\le \frac{2xy}{2-x-y} \le \pA(x,y)\le \frac{4xy}{2-x-y}\le 4xy.
\]
This completes the proof of \eqref{A:aBm}.
\end{proof}

The following result shows that for a tight framelet filter bank $\{a; b^p, b^n\}$, if the high-pass filters $b^p$ and $b^n$ are real-valued (but the filter $a$ can be complex-valued), then its frequency separation between $b^p$ and $b^n$ cannot be good. Moreover, the best possible frequency separation between two real-valued high-pass filters $b^p$ and $b^n$ in a tight framelet filter bank $\{a; b^p, b^n\}$ is achieved when $a$ is an orthogonal filter. On the other hand, Theorem~\ref{thm:A:Bspline} tells us that the frequency separation between two complex-valued high-pass filters $b^p$ and $b^n$ in a complex-valued tight framelet filter bank $\{a; b^p, b^n\}$ is the worst when $a$ is an orthogonal filter.

\begin{theorem}\label{thm:real}
Let $a, b^p, b^n\in \lp{2}$ such that $\{a; b^p, b^n\}$ is a tight framelet filter bank and the two high-pass filters $b^p$ and $b^n$ are real-valued (but the filter $a$ may be complex-valued). Then
\be \label{lowerbound:real}
\int_0^\pi \big[|\wh{b^p}(\xi+\pi)|^2+|\wh{b^n}(\xi)|^2\big] d\xi =\frac{1}{2} \int_0^\pi \big[2-
|\wh{a}(\xi)|^2-|\wh{a}(\xi+\pi)|^2\big] d\xi\ge \frac{\pi}{2},
\ee
where the equal sign holds if and only if $a$ is an orthogonal filter (that is, $|\wh{a}(\xi)|^2+|\wh{a}(\xi+\pi)|^2=1$ a.e. $\xi\in \R$).
\end{theorem}

\begin{proof} Define $B(\xi):=|\wh{b^p}(\xi+\pi)|^2+|\wh{b^n}(\xi)|^2$.
Note that a general filter $u$ has real coefficients if and only if $\ol{\wh{u}(\xi)}=\wh{u}(-\xi)$.
Therefore, we have $\wh{b^p}(\xi+\pi)=\wh{b^p}(\xi-\pi)=\ol{\wh{b^p}(\pi-\xi)}$. Hence,
$B(\xi)=|\wh{b^p}(\pi-\xi)|^2+|\wh{b^n}(\xi)|^2$.

By $|\wh{a}(\xi)|^2+|\wh{b^p}(\xi)|^2+|\wh{b^n}(\xi)|^2=1$, we have $|\wh{a}(\pi-\xi)|^2+|\wh{b^p}(\pi-\xi)|^2+|\wh{b^n}(\pi-\xi)|^2=1$. Therefore,
\be \label{B:eq1}
B(\xi)+B(\pi-\xi)=|\wh{b^p}(\pi-\xi)|^2+|\wh{b^n}(\xi)|^2+|\wh{b^p}(\xi)|^2+|\wh{b^n}(\pi-\xi)|^2
=2-|\wh{a}(\xi)|^2-|\wh{a}(\pi-\xi)|^2.
\ee
Note that
\[
1=|\wh{a}(-\xi)|^2+|\wh{b^p}(-\xi)|^2+|\wh{b^n}(-\xi)|^2
=|\wh{a}(-\xi)|^2+|\wh{b^p}(\xi)|^2+|\wh{b^n}(\xi)|^2=1+|\wh{a}(-\xi)|^2-|\wh{a}(\xi)|^2,
\]
from which we must have $|\wh{a}(-\xi)|=|\wh{a}(\xi)|$. Therefore, it follows from \eqref{B:eq1} that
\[
B(\xi)+B(\pi-\xi)=2-|\wh{a}(\xi)|^2-|\wh{a}(\xi+\pi)|^2,
\]
from which we have
\[
\int_0^\pi \big[ 2-|\wh{a}(\xi)|^2-|\wh{a}(\xi+\pi)|^2\big] d\xi=\int_0^\pi \big[ B(\xi)+B(\pi-\xi)\big] d\xi=2\int_0^\pi B(\xi) d\xi.
\]
Since $|\wh{a}(\xi)|^2+|\wh{a}(\xi+\pi)|^2\le 1$ a.e. $\xi\in \R$, we conclude from the above identity that \eqref{lowerbound:real} holds.
\end{proof}

\section{Structure of Finitely Supported Complex Tight Framelet Filter Banks}

In order to design finitely supported complex tight framelet filter banks $\{a; b^p, b^n\}$ with good directionality, we have to investigate the structure of all possible finitely supported complex-valued high-pass filters $b^p, b^n$ such that $\{a; b^p, b^n\}$ is a tight framelet filter bank. More precisely, from any given finitely supported filter $a\in \lp{0}$, we are interesting in finding all possible finitely supported complex tight framelet filter banks $\{a; b^p, b^n\}$ derived from a given low-pass filter $a$.
For prescribed filter lengths of $b^p$ and $b^n$, such a result enables us to find the best possible complex tight framelet filter bank $\{a; b^p, b^n\}$ with the best possible frequency separation, that is, $|\wh{b^p}(\xi+\pi)|^2+|\wh{b^n}(\xi)|^2\approx A(\xi), \xi\in [0,\pi]$.

To construct finitely supported tight framelet filter banks, it is convenient to use Laurent polynomials instead of $2\pi$-periodic trigonometric polynomials. Recall that $\lp{0}$ denotes the linear space of all finitely supported sequences on $\Z$. For a sequence $u=\{u(k)\}_{k\in \Z}\in \lp{0}$, its $z$-transform is a Laurent polynomial defined to be
\be \label{p:u}
\pu(z):=\sum_{k\in \Z} u(k)z^k, \qquad z\in \C \bs \{0\}.
\ee
Let $u: \Z \rightarrow \C^{r\times s}$ be a sequence of $r\times s$ matrices. We define $u^\star$ to be its associated adjoint sequence by $u^\star(k):=\ol{u(-k)}^\tp$, $k\in \Z$.
In terms of Fourier series, we have $\wh{u^\star}(\xi)=\ol{\hu(\xi)}^\tp$ and $\wh{u}(\xi)=\pu(e^{-i\xi})$.
Using Laurent polynomials, we have
\be \label{p:star}
\pu^\star(z):=[\pu(z)]^\star:=\sum_{k\in \Z} \ol{u(k)}^\tp z^{-k}, \qquad z\in \C \bs \{0\}.
\ee
In terms of Laurent polynomials, for $a, b_1, b_2\in \lp{0}$, $\{a; b_1, b_2\}$ is a tight framelet filter bank if
\be \label{def:tffb:lp}
\left[ \begin{matrix}
\pa(z) &\pb_1(z) &\pb_2(z)\\
\pa(-z) &\pb_1(-z) &\pb_2(-z)\end{matrix} \right]
\left[ \begin{matrix}
\pa(z) &\pb_1(z) &\pb_2(z)\\
\pa(-z) &\pb_1(-z) &\pb_2(-z)\end{matrix} \right]^\star=I_2
\ee
for all $z\in \C \bs \{0\}$, where $I_2$ is the $2\times 2$ identity matrix. For a $2\times 2$ matrix $\pU$ of Laurent polynomials, we say that $\pU$ is paraunitary if $\pU(z) \pU^\star(z)=I_2$ for all $z\in \T:=\{\zeta\in \C \setsp |\zeta|=1\}$, or equivalently, $\pU(e^{-i\xi})\ol{\pU(e^{-i\xi})}^\tp=I_2$ for all $\xi\in \R$.

For a Laurent polynomial $\pu$, we shall use the notation $\pu\equiv 0$ to mean that $\pu$ is identically zero, and the notation $\pu \not \equiv 0$ to mean that $\pu$ is not identically zero. We say that $u$ is an orthogonal filter if $\pu(z)\pu^\star(z)+\pu(-z)\pu^\star(-z)=1$ for all $z\in \C \bs \{0\}$.

The main result in this section is as follows:

\begin{theorem}\label{thm:tffb}
Let $a, b_1, b_2, b^p, b^n\in \lp{0}$ such that $\{a; b_1, b_2\}$ is a tight framelet filter bank and the filter $a$ is not identically zero.
Suppose that
\be \label{tffb:cond}
|\pa(z)|^2+|\pa(-z)|^2 \le 1, \qquad \forall\; z\in \T.
\ee
%
Then the following are equivalent:
\begin{enumerate}
\item[(i)] $\{a; b^p, b^n\}$ is a finitely supported tight framelet filter bank and
%
\be \label{det:relation}
\pb^p(z)\pb^n(-z)-\pb^p(-z)\pb^n(z)=\gl z^{2k}[\pb_1(z)\pb_2(-z)-\pb_1(-z)\pb_2(z)]
\ee
for some $k\in \Z$ and $\gl\in \T$. Remove condition \eqref{det:relation} if $a$ is an orthogonal filter.

\item[(ii)] There exists a $2\times 2$ paraunitary matrix $\pU$ of Laurent polynomials such that
\be \label{eqU}
\left[ \begin{matrix} \pb^p(z) &\pb^n(z)\end{matrix} \right]
=\left[ \begin{matrix} \pb_1(z) &\pb_2(z) \end{matrix} \right]
\pU(z^2), \qquad \forall\; z\in \C \bs \{0\}.
\ee
\end{enumerate}
\end{theorem}

From \eqref{def:tffb:lp}, we see that $\{a; b_1, b_2\}$ is a tight framelet filter bank if and only if
\be \label{tffb:2}
\left[ \begin{matrix}
\pb_1(z) &\pb_2(z)\\
\pb_1(-z) &\pb_2(-z)\end{matrix} \right]
\left[ \begin{matrix}
\pb_1(z) &\pb_2(z)\\
\pb_1(-z) &\pb_2(-z)\end{matrix} \right]^\star=\cM_{a}(z)
\ee
with
\be \label{Ma}
\cM_{a}(z):=\left [ \begin{matrix} 1-\pa(z) \pa^\star(z) &-\pa(z) \pa^\star(-z)\\
-\pa(-z) \pa^\star(z) &1-\pa(-z) \pa^\star(-z)\end{matrix}\right].
\ee
We define a Laurent polynomial $\df_{b_1, b_2}$ by
\be \label{db}
\df_{b_1, b_2}(z^2):=z [\pb_1(z)\pb_2(-z)-\pb_1(-z)\pb_2(z)].
\ee
Note that $\df_{b_1, b_2}$ is a well-defined Laurent polynomial.
Then it follows from \eqref{tffb:2} that
\be \label{df}
|\df_{b_1, b_2}(z^2)|^2=\det(\cM_a(z))=1-|\pa(z)|^2-|\pa(-z)|^2, \qquad \forall\; z\in \T.
\ee
If $a$ is an orthogonal filter, then we must have $\df_{b_1, b_2}\equiv 0$. For $\df_{b_1, b_2}\not \equiv 0$, by Fej\'er-Riesz lemma, we see that up to a monomial factor there are essentially only finitely many Laurent polynomials $\df_{b_1, b_2}$ satisfying \eqref{df}. As we shall discuss in Section~4, all finitely supported complex-valued tight framelet filter banks $\{a; b_1, b_2\}$ having the shortest possible filter supports can be derived from the low-pass filter $a$ by solving a system of linear equations.
Consequently, Theorem~\ref{thm:tffb} allows us to obtain all finitely supported complex-valued tight framelet filter banks $\{a; b_1, b_2\}$ with the low-pass filter $a$ being given in advance.
Using Theorem~\ref{thm:tffb}, we shall discuss in Section~4 how to find the best possible complex tight framelet filter bank $\{a; b^p, b^n\}$ with the best possible frequency separation for any prescribed filter lengths of the high-pass filters $b^p$ and $b^n$.

To prove Theorem~\ref{thm:tffb}, we need several auxiliary results.
Let us first introduce some definitions.
We say that $\pu$ is a trivial factor if it is a nonzero monomial, that is, $\pu(z)=\gl z^k$ for some $\gl\in \C \bs \{0\}$ and $k\in \Z$.
For two Laurent polynomials $\pu$ and $\pv$, by $\gcd(\pu, \pv)$ we denote the greatest common factor of $\pu$ and $\pv$.
In particular, we use the notation $\gcd(\pu, \pv)=1$ to mean that $\pu$ and $\pv$ do not have a nontrivial common factor.

\begin{lemma}\label{lem:det0}
Let $\pp_1, \pp_2, \pp_3, \pp_4$ be Laurent polynomials. Define
\be \label{P}
\pP(z):=\left [\begin{matrix} \pp_1(z) &\pp_3(z)\\ \pp_2(z) &\pp_4(z) \end{matrix}\right].
\ee
Then the following are equivalent:
\begin{enumerate}
\item $\det (\pP(z))=0$ for all $z\in \C \bs \{0\}$.

\item $\pp_1(z) \pp_4(z)-\pp_2(z)\pp_3(z)=0$ for all $z\in \C \bs \{0\}$.

\item There exist Laurent polynomials $\pq_1, \pq_2, \pq_3, \pq_4$ such that
\be \label{pq}
\pp_1(z)=\pq_1(z) \pq_3(z), \quad \pp_2(z)=\pq_2(z) \pq_3(z), \qquad
 \pp_3(z)=\pq_1(z) \pq_4(z), \quad \pp_4(z)=\pq_2(z) \pq_4(z).\
\ee

\item There exist Laurent polynomials $\pq_1, \pq_2, \pq_3, \pq_4$ such that
\[
\pP(z)=\left [ \begin{matrix} \pq_1(z)\\ \pq_2(z)\end{matrix} \right]
\left[ \begin{matrix} \pq_3(z) &\pq_4(z) \end{matrix}\right].
\]
\end{enumerate}
\end{lemma}

\begin{proof} If $\pP$ is identically zero, then all claims hold obviously. Hence, we assume that at least one of $\pp_1, \pp_2, \pp_3, \pp_4$ is not identically zero.
It is trivial that (1)$\imply$(2) and (3)$\imply$(4)$\imply$(1). To complete the proof, it suffices to prove (2)$\imply$(3).

If both $\pp_1$ and $\pp_2$ are identically zero, then the claim in item (3) obviously holds by taking $\pq_1=\pp_3, \pq_2=\pp_4, \pq_3=0$ and $\pq_4=1$. Now we assume that either $\pp_1\not \equiv 0$ or $\pp_2\not \equiv 0$, that is, at least one of $\pp_1$ and $\pp_2$ is not identically zero. Define
\be \label{q123}
\pq_3:=\gcd(\pp_1, \pp_2)\quad \mbox{and}\quad \pq_1:=\pp_1/\pq_3,\quad \pq_2:=\pp_2/\pq_3.
\ee
Since $\pq_3$ is not identically zero, all $\pq_1, \pq_2, \pq_3$ are well-defined Laurent polynomials and at least one of $\pq_1$ and $\pq_2$ are not identically zero.
Moreover, $\pp_1=\pq_1\pq_3$, $\pp_2=\pq_2\pq_3$, and $\gcd(\pq_1, \pq_2)=1$, which means that $\pq_1$ and $\pq_2$ have no nontrivial common factor. By item (2), we have
\[
0=\pp_1 \pp_4-\pp_2 \pp_3=\pq_3(\pq_1\pp_4-\pq_2\pp_3).
\]
Since $\pq_3$ is not identically zero, from the above identity we must have
$\pq_1\pp_4=\pq_2 \pp_3$.
Because at least one of $\pq_1$ and $\pq_2$ is not identically zero, without loss of generality, we may assume that $\pq_1$ is not identically zero.
By $\gcd(\pq_1, \pq_2)=1$ and $\pq_1\pp_4=\pq_2 \pp_3$, we must have $\pq_1 \mid \pp_3$.
Then we define $\pq_4=\pp_3/\pq_1$, which is a well-defined Laurent polynomial. By $\pq_1\pp_4=\pq_2 \pp_3$, we see that $\pp_4=\pq_2 \pq_4$.
Using \eqref{q123}, now one can directly check that \eqref{pq} holds.
Therefore, we complete the proof of (2)$\imply$(3).
\end{proof}

\begin{prop}\label{prop:UV}
Let $\pQ$ and $\pV$ be $2\times 2$ matrices of Laurent polynomials.
If
\be \label{VQ}
\pV(z) \pQ(z)=\left[ \begin{matrix} \pc(z) &0\\ 0 &\pd(z)\end{matrix} \right],
\ee
%
then there exist Laurent polynomials $\pu_1, \pu_2, \pu_3, \pu_4, \pv_1, \pv_2, \pv_3, \pv_4$ such that
\be \label{VQ:expr}
\pV(z)=\left [ \begin{matrix} \pv_1(z) &0\\ 0 &\pv_2(z)\end{matrix} \right]
\left[ \begin{matrix} \pu_1(z) &-\pu_3(z)\\ \pu_2(z) &\pu_4(z)\end{matrix}\right],
\qquad
\pQ(z)=
\left[ \begin{matrix} \pu_4(z) &\pu_3(z)\\ -\pu_2(z) &\pu_1(z)\end{matrix}\right]
\left [ \begin{matrix} \pv_3(z) &0\\ 0 &\pv_4(z)\end{matrix} \right]
\ee
and
\be \label{VQ:d}
\pc(z)=\pv_1(z)\pv_3(z)\big(\pu_1(z)\pu_4(z)+\pu_2(z)\pu_3(z)\big),\qquad
\pd(z)=\pv_2(z)\pv_4(z)\big(\pu_1(z)\pu_4(z)+\pu_2(z)\pu_3(z)\big).
\ee
If $\pc=1$, then we can particularly take $\pv_1=\pv_3=1$ so that $\pu_1(z)\pu_4(z)+\pu_2(z)\pu_3(z)=1$ and $\pd(z)=\pv_2(z)\pv_4(z)$.
\end{prop}

\begin{proof}
By our assumption in \eqref{VQ}, we have $[\pV(z)\pQ(z)]_{1,2}(z)=\pV_{1,1}(z) \pQ_{1,2}(z)+\pV_{1,2}(z) \pQ_{2,2}(z)=0$ for all $z\in \C \bs \{0\}$.
By Lemma~\ref{lem:det0}, there exist Laurent polynomials $\pu_1, \pu_3, \pv_1, \pv_4$ such that
\[
\left[ \begin{matrix} \pV_{1,1}(z) &\pV_{1,2}(z)\\ -\pQ_{2,2}(z) &\pQ_{1,2}(z)\end{matrix}\right]=\left[
\begin{matrix} \pv_1(z)\\ -\pv_4(z)\end{matrix}\right] \left[ \begin{matrix} \pu_1(z) &-\pu_3(z)\end{matrix}\right].
\]
Similarly, by our assumption in \eqref{VQ}, we have $[\pV(z)\pQ(z)]_{2,1}(z)=\pV_{2,1}(z) \pQ_{1,1}(z)+\pV_{2,2}(z) \pQ_{2,1}(z)=0$ for all $z\in \C \bs \{0\}$.
By Lemma~\ref{lem:det0}, there exist Laurent polynomials $\pu_2, \pu_4, \pv_2, \pv_3$ such that
\[
\left[ \begin{matrix} \pV_{2,1}(z) &\pV_{2,2}(z)\\ -\pQ_{2,1}(z) &\pQ_{1,1}(z)\end{matrix}\right]=\left[
\begin{matrix} \pv_2(z)\\ \pv_3(z)\end{matrix}\right] \left[ \begin{matrix} \pu_2(z) &\pu_4(z)\end{matrix}\right].
\]
Now we can directly check that both \eqref{VQ:expr} and \eqref{VQ:d} are satisfied.

If $\pc=1$, then it follows from \eqref{VQ:d} that all $\pv_1, \pv_3$ and $\pu_1\pu_4+\pu_2\pu_3$ must be monomials. Now it follows directly from \eqref{VQ:expr} that
\[
\pV(z)=\left [ \begin{matrix} 1 &0\\ 0 &\pv_2(z)/\pv_3(z)\end{matrix} \right]
\left[ \begin{matrix} \pu_1(z)\pv_1(z) &-\pu_3(z)\pv_1(z)\\ \pu_2(z)\pv_3(z) &\pu_4(z)\pv_3(z)\end{matrix}\right]
\]
and
\[
\pQ(z)=
\left[ \begin{matrix} \pu_4(z)\pv_3(z) &\pu_3(z)\pv_1(z)\\ -\pu_2(z)\pv_3(z) &\pu_1(z)\pv_1(z)\end{matrix}\right]
\left [ \begin{matrix} 1 &0\\ 0 &\pv_4(z)/\pv_1(z)\end{matrix} \right].
\]
Redefine $\pu_1, \pu_2, \pu_3, \pu_4, \pv_2, \pv_4$ as $\pu_1\pv_1, \pu_2\pv_3, \pu_3\pv_1, \pu_4\pv_3, \pv_2/\pv_3, \pv_4/\pv_1$, respectively. We now see that the claim holds for the particular case of $\pc=1$.
\end{proof}

As a direct consequence of Proposition~\ref{prop:UV}, we have the following two corollaries.

\begin{cor}\label{cor:paraunitary}
Let $\pP$ be a $2\times 2$ matrix of Laurent polynomials defined in \eqref{P}. Then $\pP$ is paraunitary, that is, $\pP(z) \pP^\star(z)=I_2$ for all $z\in \C \bs \{0\}$, if and only if
\be \label{pu:u12}
\pp_3(z)=-\gl z^k \pp_2^\star(z),\quad \pp_4(z)=\gl z^k \pp_1^\star(z), \quad \pp_1(z)\pp_1^\star(z)+\pp_2(z)\pp_2^\star(z)=1 \quad \mbox{with}\quad \gl \in \T, k\in \Z.
\ee
\end{cor}

\begin{proof} Let $\pQ$ and $\pV$ be the $2\times 2$ matrix of Laurent polynomials defined by
\[
\pV(z):=\pP(z) \quad \mbox{and} \quad
\pQ(z):=\pP^\star(z)=\left[ \begin{matrix} \pp^\star_1(z) &\pp^\star_2(z)\\ \pp_3^\star(z)
&\pp_4^\star(z)\end{matrix} \right].
\]
If $\pP$ is paraunitary, then $\pV(z)\pQ(z)=I_2$. By Proposition~\ref{prop:UV} with $\pc=1$, we see that \eqref{pu:u12} must hold.

Conversely, if \eqref{pu:u12} is satisfied, then we can directly check that $\pP$ is a paraunitary matrix.
\end{proof}

\begin{cor}\label{cor:UV12}
Let $\pQ, \pV, \mathring{\pQ}, \mathring{\pV}$ be $2\times 2$ matrices of Laurent polynomials. If
\be \label{VQ12:cond1}
\pV(z)\pQ(z)=\left[ \begin{matrix} 1 &0\\ 0 &\pd(z)\end{matrix}\right]
=\mathring{\pV}(z)\mathring{\pQ}(z)
\ee
and
\be \label{VQ12:cond2}
\det( \mathring{\pV}(z))= \gl z^k \det(\pV(z)) \qquad \mbox{for some}\;\; \gl\in \C \bs \{0\}, k\in \Z.
\ee
Then there exists a $2\times 2$ matrix $\pU$ of Laurent polynomials such that $\det (\pU(z))=\gl z^k$ and
\be \label{V12}
\mathring{\pV}(z)=\pV(z) \pU(z).
\ee
\end{cor}

\begin{proof} By Proposition~\ref{prop:UV} with $\pc=1$, we see that
\[
\pV(z)=\left [ \begin{matrix} 1 &0\\ 0 &\det(\pV(z)) \end{matrix} \right]
\pU_1(z), \quad
\mathring{\pV}(z)=\left [ \begin{matrix} 1 &0\\ 0 &\det(\mathring{\pV}(z)) \end{matrix} \right] \pU_2(z),
\]
where $\pU_1, \pU_2$ are $2\times 2$ matrices of Laurent polynomials such that $\det(\pU_1(z))=\det(\pU_2(z))=1$. Therefore, $[\pU_1(z)]^{-1}$ is also a matrix of Laurent polynomials. Define
\[
\pU(z):=[\pU_1(z)]^{-1} \left[ \begin{matrix} 1 &0\\ 0 &\gl z^k \end{matrix}\right] \pU_2(z).
\]
Now it is trivial to check that \eqref{V12} holds and $\det(\pU(z))=\gl z^k$ is a nontrivial monomial.
\end{proof}

Now we have the following result about the essential uniqueness of factorization of a positive semidefinite $2\times 2$ matrix of Laurent polynomials.

\begin{theorem}\label{thm:uni}
Let $\pP$ be a $2\times 2$ matrix of Laurent polynomials given in \eqref{P}
such that $\det(\pP(z)) \not \equiv 0$ (that is, the determinant of $\pP$ is not identically zero) and $\gcd(\pp_1, \pp_2, \pp_3, \pp_4) = 1$.
If $\pV$ and $\mathring{\pV}$ are $2\times 2$ matrices of Laurent polynomials satisfying
\be\label{V:cond1}
\pV(z) \pV^\star(z)= \pP(z)=\mathring{\pV}(z)\mathring{\pV}^\star(z)
\ee
and
\be \label{V:cond2}
\det(\mathring{\pV}(z))=\gl z^k \det(\pV(z)) \qquad \mbox{for some}\;\; \gl\in \C \bs \{0\}, k\in \Z,
\ee
then there exists a $2\times 2$ paraunitary matrix $\pU$ of Laurent polynomials such that $\mathring{\pV}(z)=\pV(z) \pU(z)$, $\det(\pU(z))=\gl z^k$, and $\pU(z)\pU^\star(z)=I_2$ for all $z\in \C \bs \{0\}$.
\end{theorem}

\begin{proof}
It is a basic result in linear algebra that there exist
two $2\times 2$ matrices $\pA$ and $\pB$ of Laurent polynomials satisfying
$\det(\pA(z))=\det(\pB(z))= 1$ and
\[
\pA(z) \pP(z) \pB(z) = \left[ \begin{matrix}
\pc(z) & 0 \\ 0 & \pd(z) \end{matrix} \right]
\]
with $\pc, \pd$ being Laurent polynomials satisfying $\pc \mid \pd$.
The above result can be proved using elementary matrix forms and Euclidian division of Laurent polynomials. The diagonal matrix $\mbox{diag}(\pc, \pd)$ is called the Smith normal form of $\pP$ and such Laurent polynomials $\pc, \pd$ are essentially unique.
See \cite{Smith} for a detailed proof of the above result.
Moreover, one can directly verify that $\pc=\gcd(\pp_1, \pp_2, \pp_3, \pp_4)=1$ and $\pd = \det(\pP) /\pc\not \equiv 0$. Consequently, by \eqref{V:cond1}, we have
\[
(\pA(z) \pV(z)) (\pV^\star(z)\pB(z)) = \left[ \begin{matrix}
1 & 0 \\ 0 & \pd(z) \end{matrix}\right]
=(\pA(z)\mathring{\pV}(z)) (\mathring{\pV}^\star(z)\pB(z)).
\]
Note that
\[
\det(\pA(z) \mathring{\pV}(z))= \det(\pA(z)) \det(\mathring{\pV}(z))=\det(\mathring{\pV}(z))
=\gl z^k \det(\pV(z))=\gl z^k \det(\pA(z) \pV(z)).
\]
Consequently, it follows from Corollary~\ref{cor:UV12} that there exists a $2\times 2$ matrix
$\pU$ of Laurent polynomials such that $\det(\pU(z))=\gl z^k$ and
$\pA(z) \mathring{\pV}(z) = \pA(z) \pV(z)\pU(z)$, from which we have $\mathring{\pV}(z)=\pV(z)\pU(z)$ since $\det(\pA(z))=1$.
Therefore, it follows from \eqref{V:cond1} that $\pV(z) \pV^\star(z)=\mathring{\pV}(z)\mathring{\pV}^\star(z)$ which leads to
\[
\pV(z)\big( \pU(z)\pU^\star(z)-I_2\big) \pV^\star(z)=0.
\]
By \eqref{V:cond1}, we have $\det (\pV(z)) \det(\pV^\star(z))=\det(\pP(z))\not \equiv 0$ and therefore, $\det (\pV(z))\not \equiv 0$. Thus, $\pV(z)$ is invertible for all $z$ satisfying $\det (\pV(z))\neq 0$. Now we deduce from the above identity that we must have
$\pU(z)\pU^\star(z)= I_2$ for all $z\in \C \bs \{0\}$.
\end{proof}

We are now ready to prove Theorem~\ref{thm:tffb}.

\begin{proof}[Proof of Theorem~\ref{thm:tffb}]
(ii)$\imply$(i) is trivial. Note that \eqref{eqU} is equivalent to
\be \label{eqU:2}
\left[ \begin{matrix} \pb^p(z) &\pb^n(z)\\
\pb^p(-z) &\pb^n(-z)\end{matrix} \right]
=\left[ \begin{matrix} \pb_1(z) &\pb_2(z)\\
\pb_1(-z) &\pb_2(-z)\end{matrix} \right]
\pU(z^2), \qquad \forall\; z\in \C \bs \{0\}.
\ee
Since $\{a; b_1, b_2\}$ is a tight framelet filter bank and $\pU$ is paraunitary, it follows directly from \eqref{tffb:2} and \eqref{eqU:2} that $\{a; b^p, b^n\}$ is a finitely supported tight framelet filter bank. Moreover, it follows directly from
\eqref{eqU:2} that \eqref{det:relation} holds with $\gl z^{k}:=\det(\pU(z))$.

We now prove (i)$\imply$(ii).
For a sequence $u:\Z \rightarrow \C$ and $\gamma\in \Z$, its coset sequence $u^{[\gamma]}$ is defined to be $u^{[\gamma]}(k):=u(\gamma+2k), k\in \Z$.
Since both $\{a; b_1, b_2\}$ and $\{a; b^p, b^n\}$ are
finitely supported tight framelet filter banks, using coset sequences, we see from \eqref{tffb:2} that
\be \label{PR}
\left[ \begin{matrix}
\pb^{p,[0]}(z) &\pb^{n,[0]}(z)\\ \pb^{p,[1]}(z) &\pb^{n,[1]}(z)\end{matrix} \right]
\left[ \begin{matrix}
\pb^{p,[0]}(z) &\pb^{n,[0]}(z)\\
\pb^{p,[1]}(z) &\pb^{n,[1]}(z)\end{matrix} \right]^\star
=\cN_{a}(z)=
\left[ \begin{matrix}
\pb_1^{[0]}(z) &\pb_2^{[0]}(z)\\
\pb_1^{[1]}(z) &\pb_2^{[1]}(z)\end{matrix} \right]
\left[ \begin{matrix}
\pb_1^{[0]}(z) &\pb_2^{[0]}(z)\\
\pb_1^{[1]}(z) &\pb_2^{[1]}(z)\end{matrix} \right]^\star,
\ee
where
\[
\cN_a(z):=\left [ \begin{matrix} \tfrac{1}{2}-\pa^{[0]}(z) (\pa^{[0]}(z))^\star &-\pa^{[0]}(z) (\pa^{[1]}(z))^\star\\
-(\pa^{[0]}(z))^\star \pa^{[1]}(z) & \tfrac{1}{2}-\pa^{[1]}(z) (\pa^{[1]}(z))^\star\end{matrix} \right].
\]
Define $\pc(z) := \gcd([\cN_a(z)]_{1,1}, [\cN_a(z)]_{1,2},
[\cN_a(z)]_{2,1}, [\cN_a(z)]_{2,2})$.
By direct calculation, we have $2\det(\cN_a(z))=\frac{1}{2}-\pa^{[0]}(z)(\pa^{[0]}(z))^\star-\pa^{[1]}(z)
(\pa^{[1]}(z))^\star$
and $\operatorname{trace}(\cN_a(z)) = 1-\pa^{[0]}(z)(\pa^{[0]}(z))^\star-\pa^{[1]}(z)(\pa^{[1]}(z))^\star$.
Therefore, $\pc$ must be a factor of $\operatorname{trace}(\cN_a(z))
- 2\det(\cN_a(z))=1/2$. Consequently, we conclude that $\pc = 1$.
We now consider two cases. We first consider the case that $a$ is not an orthogonal filter. Then $\det(\cN_a(z))\not \equiv 0$. By Theorem \ref{thm:uni}, there
must exist a $2\times 2$ paraunitary matrix $\pU$ of Laurent polynomials such that
\be \label{b:U}
\left[ \begin{matrix} \pb^{p,[0]}(z) &\pb^{n,[0]}(z)\\
\pb^{p,[1]}(z) &\pb^{n,[1]}(z)\end{matrix} \right]
=\left[ \begin{matrix} \pb_1^{[0]}(z) &\pb_2^{[0]}(z)\\
\pb_1^{[1]}(z) &\pb_2^{[1]}(z)\end{matrix} \right]\pU(z)
\ee
for all $z\in \C \bs \{0\}$. Since $\pu(z)=\pu^{[0]}(z^2)+z\pu^{[1]}(z^2)$ holds for any $u\in \lp{0}$, it is straightforward to deduce from \eqref{b:U} that \eqref{eqU} holds. Hence item (ii) is proved if $a$ is not an orthogonal filter.

We now consider the case that $a$ is an orthogonal filter.
Define a filter $b$ by $\pb(z):=z\pa^\star(-z)$. Then $\{a; b\}$ is a tight framelet filter bank. It suffices to prove item (ii) with $b_1=b$ and $b_2=0$. Since $a$ is an orthogonal filter, we must have
\be \label{orth:eq} \pa^{[0]}(z)(\pa^{[0]}(z))^\star+\pa^{[1]}(z)(\pa^{[1]}(z))^\star=
\pb^{[0]}(z)(\pb^{[0]}(z))^\star+\pb^{[1]}(z)(\pb^{[1]}(z))^\star=
1/2
\ee
and $\det(\cN_a(z))=0$.
By \eqref{PR} and $\det(\cN_a(z))=0$, we must have $\pb^{p,[0]}(z)\pb^{n,[1]}(z)-\pb^{p,[1]}(z)\pb^{n,[0]}(z)=0$. By Lemma~\ref{lem:det0}, there exist Laurent polynomials $\pp_1, \pp_2, \pp_3, \pp_4$ such that
\[
\left[ \begin{matrix}
\pb^{p,[0]}(z) &\pb^{n,[0]}(z)\\ \pb^{p,[1]}(z) &\pb^{n,[1]}(z)\end{matrix} \right]
=\left[ \begin{matrix} \pp_1(z)\\ \pp_2(z)\end{matrix}\right]
\left[\begin{matrix} \pp_3(z)  &\pp_4(z)\end{matrix}\right].
\]
Since $b_1=b$ and $b_2=0$,  now \eqref{PR} and \eqref{orth:eq} imply
\[
\left[ \begin{matrix} \pp_1(z)\\ \pp_2(z)\end{matrix}\right]
\left[\begin{matrix} \pp_3(z)  &\pp_4(z)\end{matrix} \right ] \left[ \begin{matrix} \pp_3^\star(z)\\ \pp_4^\star(z)\end{matrix} \right] \left[\begin{matrix} \pp_1^\star(z) &\pp_2^\star(z)\end{matrix}\right ]=
\left[ \begin{matrix} \pb^{[0]}(z)\\ \pb^{[1]}(z)\end{matrix} \right] \left[ \begin{matrix} (\pb^{[0]}(z))^\star &(\pb^{[1]}(z))^\star\end{matrix} \right].
\]
Multiplying $\left[\begin{matrix} \pb^{[0]}(z) &\pb^{[1]}(z)\end{matrix}\right]^\tp$ from the right on both sides of the above identity, by \eqref{orth:eq}, we see that
\be \label{orth:eq2}
\pq(z)\left[ \begin{matrix} \pp_1(z)\\ \pp_2(z)\end{matrix}\right]=
\left[ \begin{matrix} \pb^{[0]}(z)\\ \pb^{[1]}(z)\end{matrix} \right]
\quad \mbox{with}\quad \pq(z):=2[\pp_3(z) \pp_3^\star(z)+\pp_4(z)\pp_4^\star(z)]
[\pp_1^\star(z)\pb^{[0]}(z)+\pp_2^\star(z)\pb^{[1]}(z)].
\ee
Since $\gcd(\pb^{[0]}, \pb^{[1]})=1$ by \eqref{orth:eq}, $\pq$ must be a nontrivial monomial.
Consequently, without loss of any generality, we can assume that $\pp_1=\pb^{[0]}$ and $\pp_2=\pb^{[1]}$. Then it follows from \eqref{orth:eq} and \eqref{orth:eq2} that $\pq=1$ and $\pp_3(z)\pp_3^\star(z)+\pp_4(z)\pp_4^\star(z)=1$.
Consequently,
\[
\pU(z):=
\left[ \begin{matrix} \pp_3(z) &\pp_4(z)\\
-\pp_4^\star(z) &\pp_3^\star(z)\end{matrix} \right]
\]
is a paraunitary matrix and it is trivial to check that \eqref{b:U} is satisfied, since
\[
\left[ \begin{matrix} \pb_1^{[0]}(z) &\pb_2^{[0]}(z)\\
\pb_1^{[1]}(z) &\pb_2^{[1]}(z)\end{matrix} \right]\pU(z)
=\left[ \begin{matrix} \pb^{[0]}(z)\\ \pb^{[1]}(z)\end{matrix}\right]
\left[ \begin{matrix} 1 &0\end{matrix}\right] \pU(z)=
\left[ \begin{matrix} \pb^{[0]}(z)\\ \pb^{[1]}(z)\end{matrix}\right]
\left[\begin{matrix} \pp_3(z)  &\pp_4(z)\end{matrix}\right]
=\left[ \begin{matrix} \pb^{p,[0]}(z) &\pb^{n,[0]}(z)\\
\pb^{p,[1]}(z) &\pb^{n,[1]}(z)\end{matrix} \right].
\]
This proves item (ii) for the case that $a$ is an orthogonal filter.
\end{proof}

\section{Algorithms and Examples of Finitely Supported Complex Tight Framelet Filter Banks with Directionality}

In this section we shall propose an algorithm to construct finitely supported complex tight framelet filter banks $\{a; b^p, b^n\}$ with good frequency separation from any given finitely supported low-pass filter $a$ satisfying \eqref{tffb:cond}. Then we shall provide several examples to illustrate our algorithm.

For a finitely supported sequence $u=\{u(k)\}_{k\in \Z}$ such that $u(k)=0$ for all $k\in \Z \bs [m,n]$ and $u(m)u(n)\ne 0$, we define $\fs(u):=\fs(\pu):=[m,n]$ to be the filter support of $u$ and define $\len(u):=\len(\pu):=n-m$ to be the length of the filter $u$.

In order to employ Theorem~\ref{thm:tffb} to obtain all finitely supported tight framelet filter banks derived from a given low-pass filter, we now recall
an algorithm, which is a special case of \cite[Algorithm~4]{Han:symdf}, to construct all possible complex tight framelet filter banks $\{a; b_1, b_2\}$ having the shortest filter support, that is, $\max(\len(b_1),\len(b_2))\le \len(a)$.

\begin{algorithm} \label{alg:tffb}
Let $a \in \lp{0}$ be a finitely supported filter on $\Z$ satisfying \eqref{tffb:cond}.

\begin{enumerate}
\item[(S1)] Define $\pA(z):=1-\pa(z)\pa^\star(z)$, $\pB(z):=-\pa(z)\pa^\star(-z)$, and
$\pD(z^2):=1-\pa(z)\pa^\star(z)-\pa(-z)\pa^\star(-z)$;

\item[(S2)] Select $\eps, s_1, s_2 \in \{0,1\}$ and a polynomial $\pd$ satisfying
$\pd(z)\pd^\star(z)=\pD(z)$ with $\lceil \tfrac{s_1+s_2-1}{2} \rceil \le m_{\pd} \le n_{\pd}\le \lfloor \tfrac{s_1+s_2-1}{2}\rfloor+n_0+\eps$, where $[-n_0, n_0]:=\fs(\pA)$ and $[m_{\pd}, n_{\pd}]:=\fs(\pd)$;

\item[(S3)] Parameterize a filter $\pb_1$ by $\pb_1(z)=z^{s_1} \sum_{j=0}^{n_0+\eps} t_j z^j$. Find the unknown coefficients $\{t_0, \ldots, t_{n_0+\eps}\}$ by solving a system $X$ of linear equations induced by $\cR(z)\equiv 0$ and
\[
\mbox{coeff}(\pb_2^\star,z,j)=0, \quad j=s_1-n_0-2m_{\pd}-1, \ldots, s_2-1 \quad \mbox{and}\quad j=s_2+n_0+\eps+1,\ldots, s_1+2n_0-2n_{\pd}+\eps-1,
\]
where $\cR$ and $\pb_1^\star$ are uniquely determined by $\fs(\cR)\subseteq [2m_{\pd}, 2 n_{\pd}-1]$ and
\[
\pB(-z) \pb_1(z)- \pA(z) \pb_1(-z)=\pd(z^2) z \pb^\star_2(z)+\cR(z);
\]

\item[(S4)]
For any nontrivial solution to the system $X$ in (S3),
there must exist $\gl>0$ such that
\[
\gl \pd(z^2)=z^{-1}[\pb_1(z)\pb_2(-z)-\pb_1(-z)\pb_2(z)]
\]
holds. Replace $\pb_1, \pb_2$ by $\gl^{-1/2}\pb_1,\gl^{-1/2}\pb_2$, respectively;
\end{enumerate}
Then $\{\ta; \tb_1, \tb_2\}$ is a finitely supported tight framelet filter bank satisfying $\max(\len(b_1),\len(b_2))\le \len(a)+\eps$.
\end{algorithm}

We are now ready to present an algorithm to construct finitely supported complex tight framelet filter banks with frequency separation property.

\begin{algorithm}\label{alg:main}
Let $a\in \lp{0}$ be a finitely supported filter on $\Z$ satisfying \eqref{tffb:cond}.
\begin{enumerate}
\item[(S1)] Construct a finitely supported tight framelet filter
bank $\{a; b_1, b_2\}$ by Algorithm~\ref{alg:tffb};

\item[(S2)] Choose a suitable filter length $N\in \N\cup\{0\}$ and parameterize filters $u_1$ and $u_2$ by
\[
\pu_1(z) := c_0 + c_1 z + \cdots + c_N z^N, \qquad
\pu_2(z) := d_0 + d_1 z + \cdots + d_N z^N,
\]
where $c_0,  \ldots, c_N, d_0,  \ldots, d_N$ are complex numbers to be determined later.
We can further assume $c_0\in \R$ by normalizing the first filter $u_1$;

\item[(S3)] Define new high-pass filters $b^p$ and $b^n$ by
\[
\pb^p(z) := \pb_1(z) \pu_1(z^2)+ \pb_2(z) \pu_2(z^2), \qquad
\pb^n(z) := z^{2m}[\pb_2(z) \pu_1^\star(z^2)-\pb_1(z) \pu_2^\star(z^2)],
\]
where $m$ is an integer such that the centers of $\fs(\pb^p)$ and $\fs(\pb^n)$ are close to each other;

\item[(S4)] If in addition the given filter $a$ is real-valued, then we further require that the initial filters $b_1, b_2$ should be real-valued and $c_0, \ldots, c_N, d_0, \ldots, d_N\in \R$. Further replace the filters $\pb^p$ and $\pb^n$ in (S3) by  $[\pb^p(z)+i\pb^n(z)]/\sqrt{2}$ and $[\pb^p(z)-i\pb^n(z)]/\sqrt{2}$, respectively;

\item[(S5)] Find a solution $\{c_0, \ldots, c_N, d_0, \ldots, d_N\}$ of the following constrained optimization problem:
\[
\min_{u_1, u_2} \int_0^\pi [|\pb^p(-e^{-i\xi})|^2+|\pb^n(e^{-i\xi})|^2] d\xi
\]
under the constraint $|\pu_1(e^{-i\xi})|^2 + |\pu_2(e^{-i\xi})|^2 = 1$ for all $\xi\in \R$ (such constraint on $u_1, u_2$ can be rewritten as equations using $c_0, \ldots, c_N, d_0, \ldots, d_N$).
\end{enumerate}
Then $\{a; b^p, b^n\}$ is a tight framelet filter bank.
For a real-valued filter $a$, in addition we have $b^n=\ol{b^p}$.
\end{algorithm}

Using Algorithms~\ref{alg:tffb} and \ref{alg:main}, many examples of finitely supported complex tight framelet filter banks with good directionality can be easily constructed. Here we only present several examples to illustrate Algorithms~\ref{alg:tffb} and \ref{alg:main}. In order to see the improvement of directionality of a tight framelet filter bank $\{a; b^p, b^n\}$, we shall use the following quantities:
\be \label{direction}
d_{\R}:=\frac{1}{2} \int_0^\pi [2-|\wh{a}(\xi)|^2-|\wh{a}(\xi+\pi)|^2]d\xi,\quad
d_A:=\int_0^\pi A(\xi) d\xi, \quad
d_{B}:=\int_0^\pi [|\wh{b^p}(\xi+\pi)|^2+|\wh{b^n}(\xi)|^2]d\xi,
\ee
where the sharp theoretical lower bound frequency separation function $A$ is defined in \eqref{A} and the subscript $\R$ in $d_{\R}$ refers to the case of real-valued high-pass filters.
By Theorem~\ref{thm:lowerbound}, we always have $d_A\le d_B$. If both $b^p$ and $b^n$ are real-valued filters, by Theorem~\ref{thm:real} we always have $d_{\R}=d_B$.

\begin{example}\label{ex1} {\rm Let $\pa(z) =(z^{-1}+2+z)/4=\{\tfrac{1}{4}, \tfrac{1}{2}, \tfrac{1}{4}\}_{[-1,1]}$ be the B-spline filter of order $2$.
Using Algorithm~\ref{alg:tffb}, we obtain a tight framelet filter bank $\{a; b_1, b_2\}$ with
$\pb_1(z) = \frac{\sqrt{6}}{6}(1- z^{-1})$ and $\pb_2(z) =\frac{\sqrt{3}}{12} (1 - z^{-1})(1+3z)$.
Applying Algorithm~\ref{alg:main} with $N =0$, we have a finitely supported complex tight framelet filter bank $\{a; b^p, b^n\}$ with
$b^n=\ol{b^p}$ and
\begin{align*}
\pb^p(z):=\tfrac{1}{24}(1-z^{-1})[(-3\sqrt{2}+6i)z+(3\sqrt{2}+6i)].
\end{align*}
By calculation we have $d_{\R}=\frac{5}{8}\pi \approx 1.96349$, $d_A\approx 0.05339$, and $d_B\approx 0.549282$. 
If we take $N=2$, then
\begin{align*}
\pb^p(z)=&(-0.0296422357615+0.0245498453274i)z^{-3}+(0.0659915437767-0.0546545208555i)z^{-2}\\
&-(0.134097034665+0.310569363502i)z^{-1}-(0.199259492568+0.279133899130i)\\
&+(0.256396707846-0.0503651650867i)z+(0.00392785810334+0.00474261627250i)z^2\\
&+(0.0366826532674+0.0442917599692i)z^{3}.
\end{align*}
By calculation, we have $d_B\approx 0.329559$. 
See Figure~\ref{fig1} for the graphs of the eight tight framelet generators in the associated two-dimensional real-valued tight framelet for $L_2(\R^2)$ in \eqref{tp:ctf}.
} \end{example}

\begin{figure}[ht]
\begin{center}
\subfigure{
\includegraphics[width=1.5in,height=1.0in]
{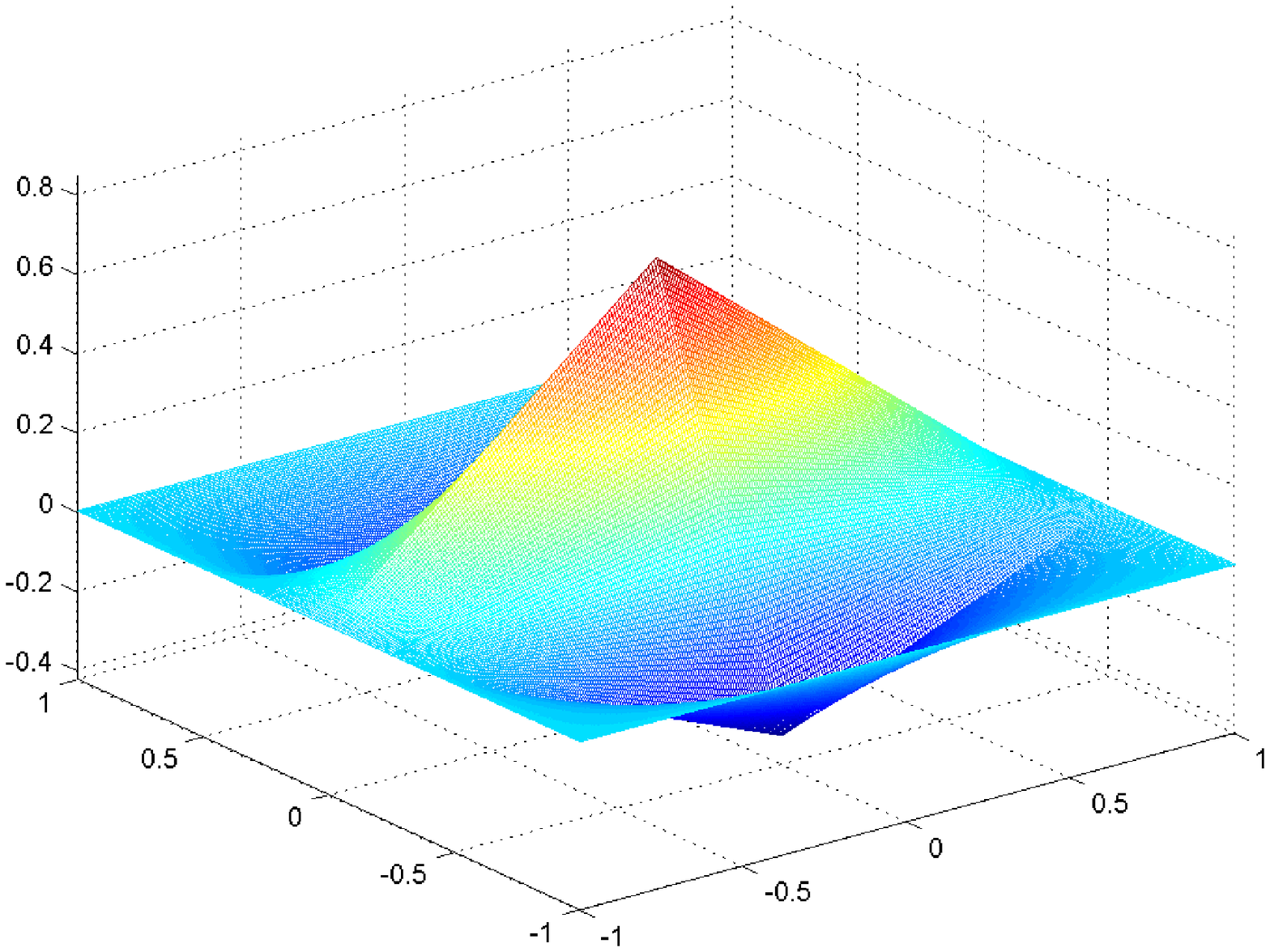}}
\subfigure{
\includegraphics[width=1.5in,height=1.0in]
{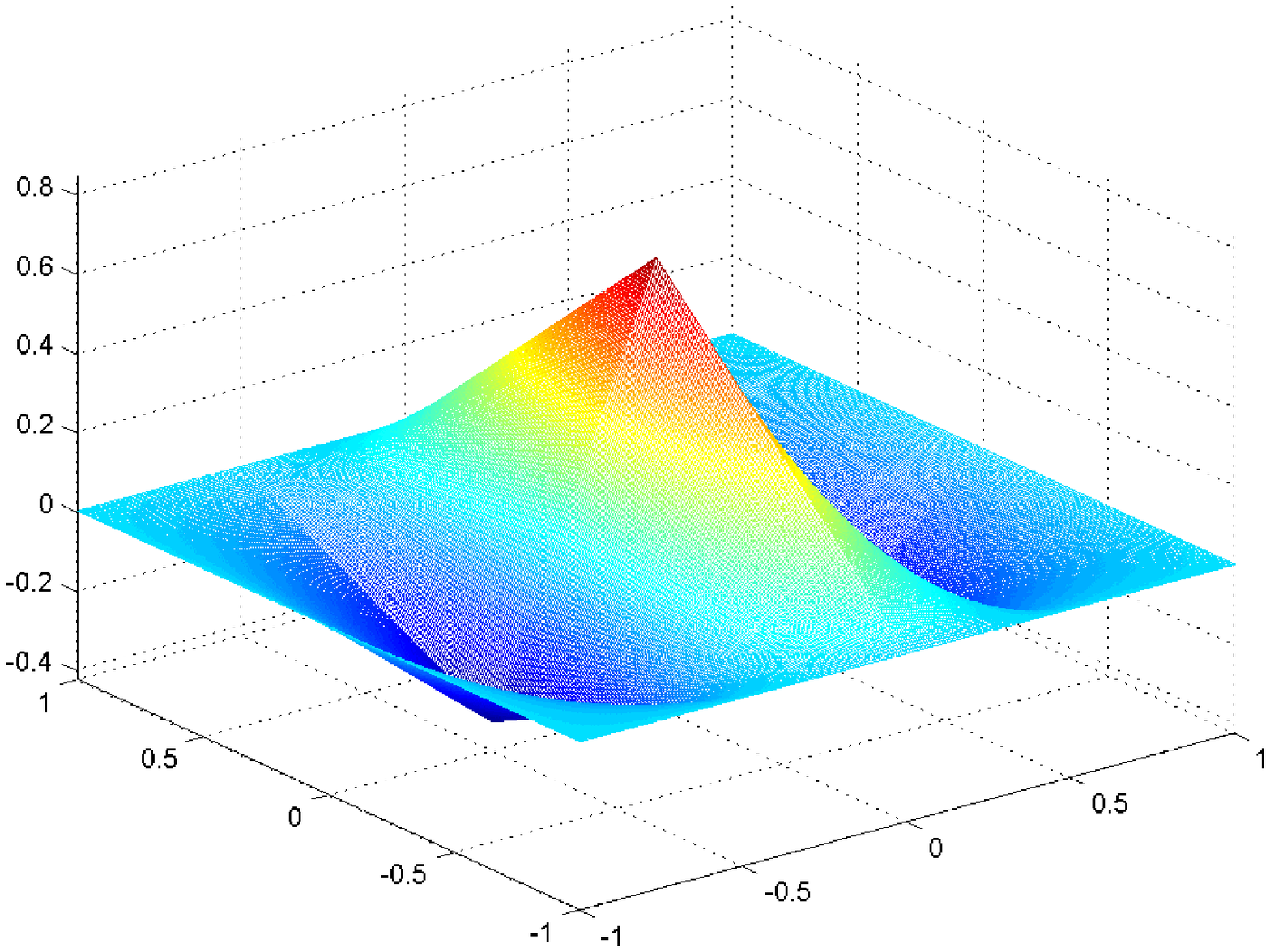}}
\subfigure{
\includegraphics[width=1.5in,height=1.0in]
{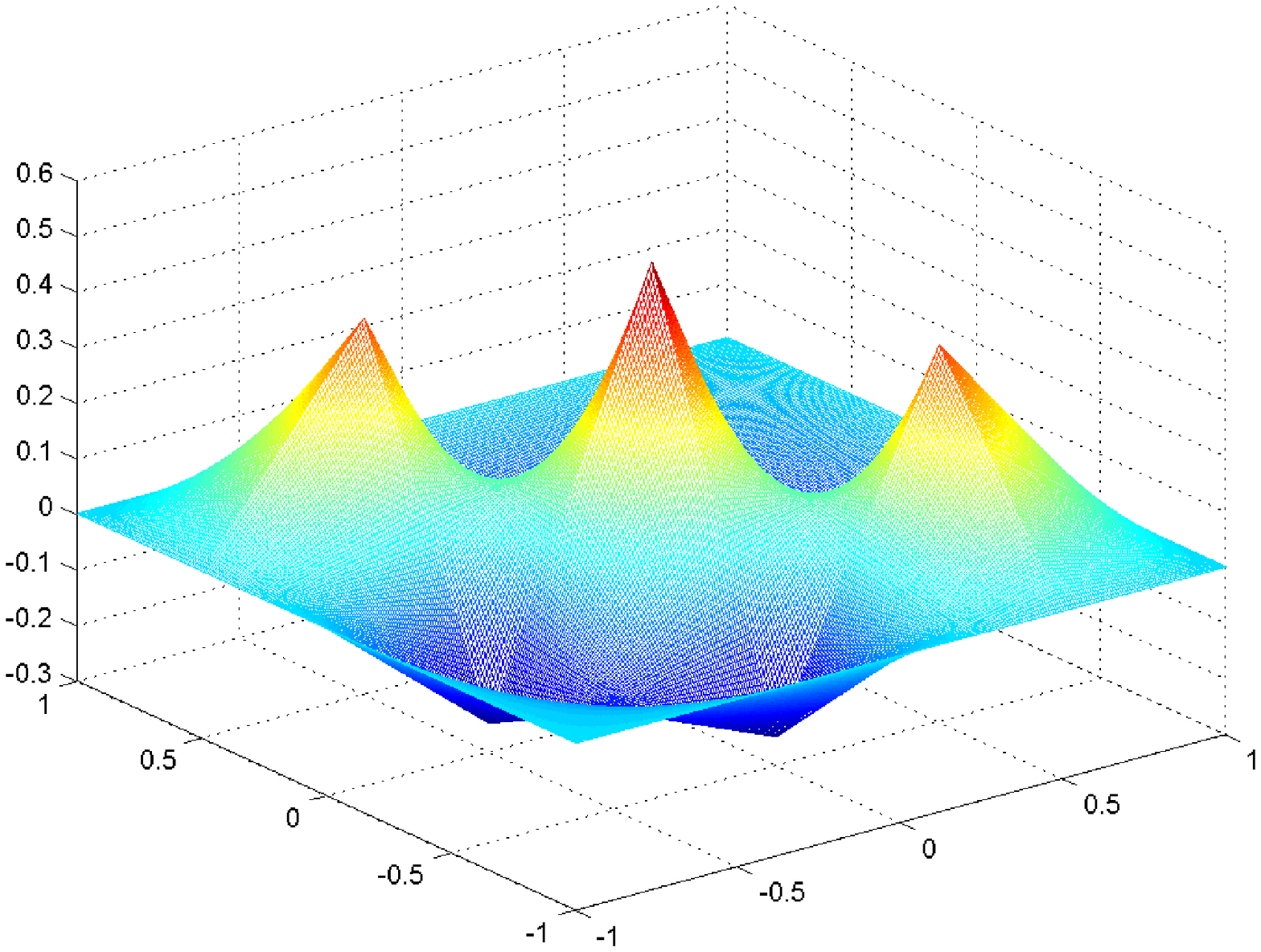}}
\subfigure{
\includegraphics[width=1.5in,height=1.0in]
{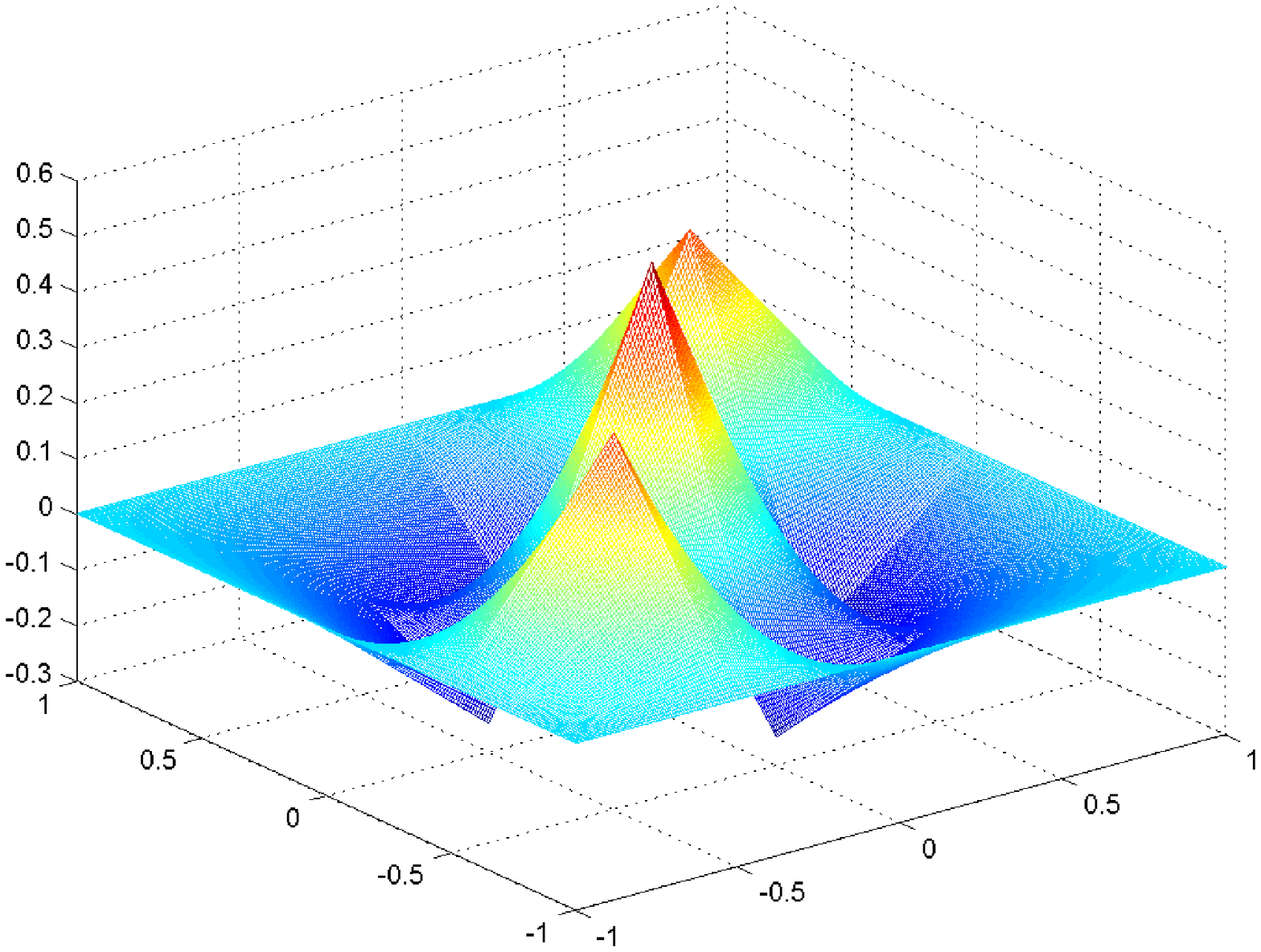}}\\
\subfigure{
\includegraphics[width=1.5in,height=1.0in]
{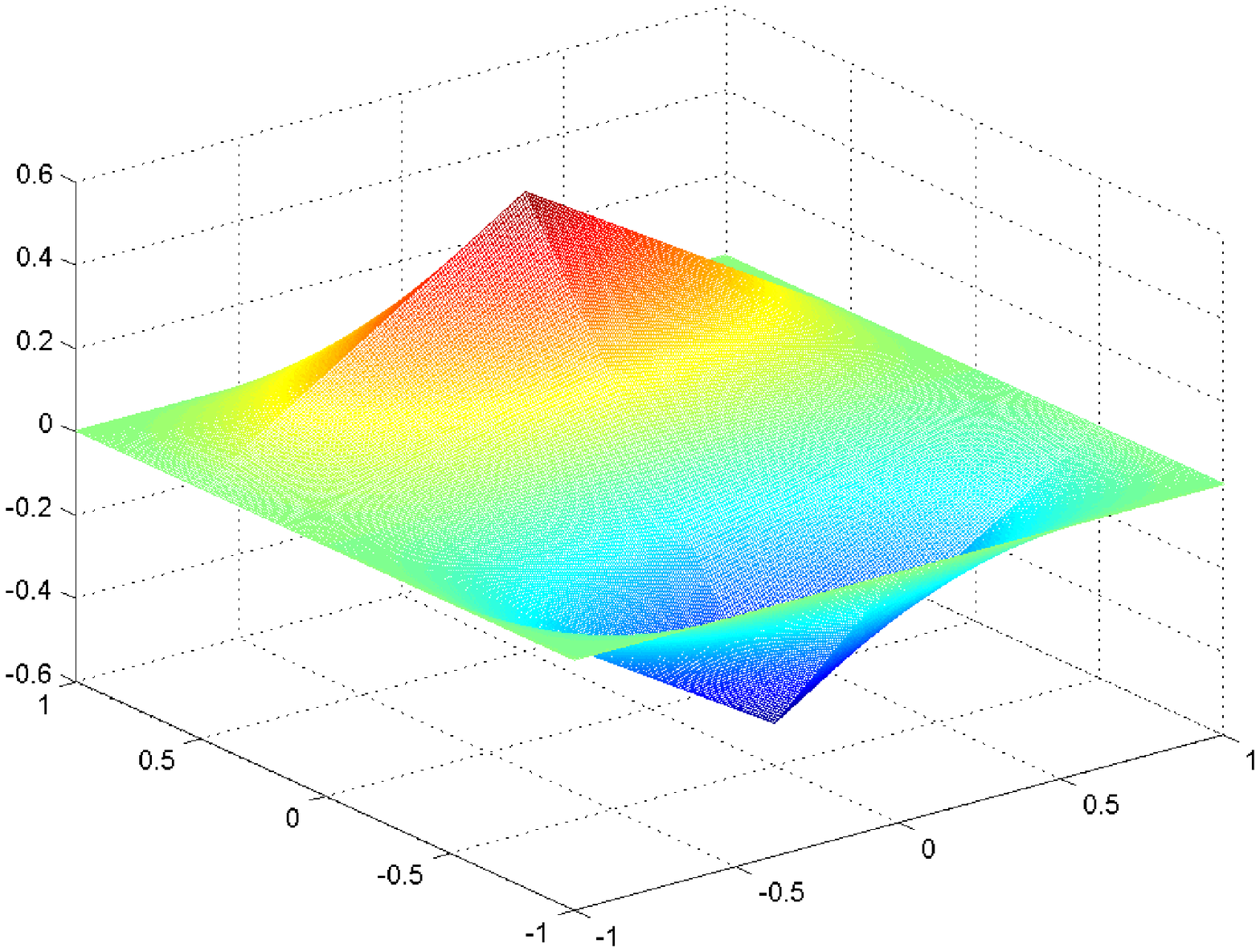}}
\subfigure{
\includegraphics[width=1.5in,height=1.0in]
{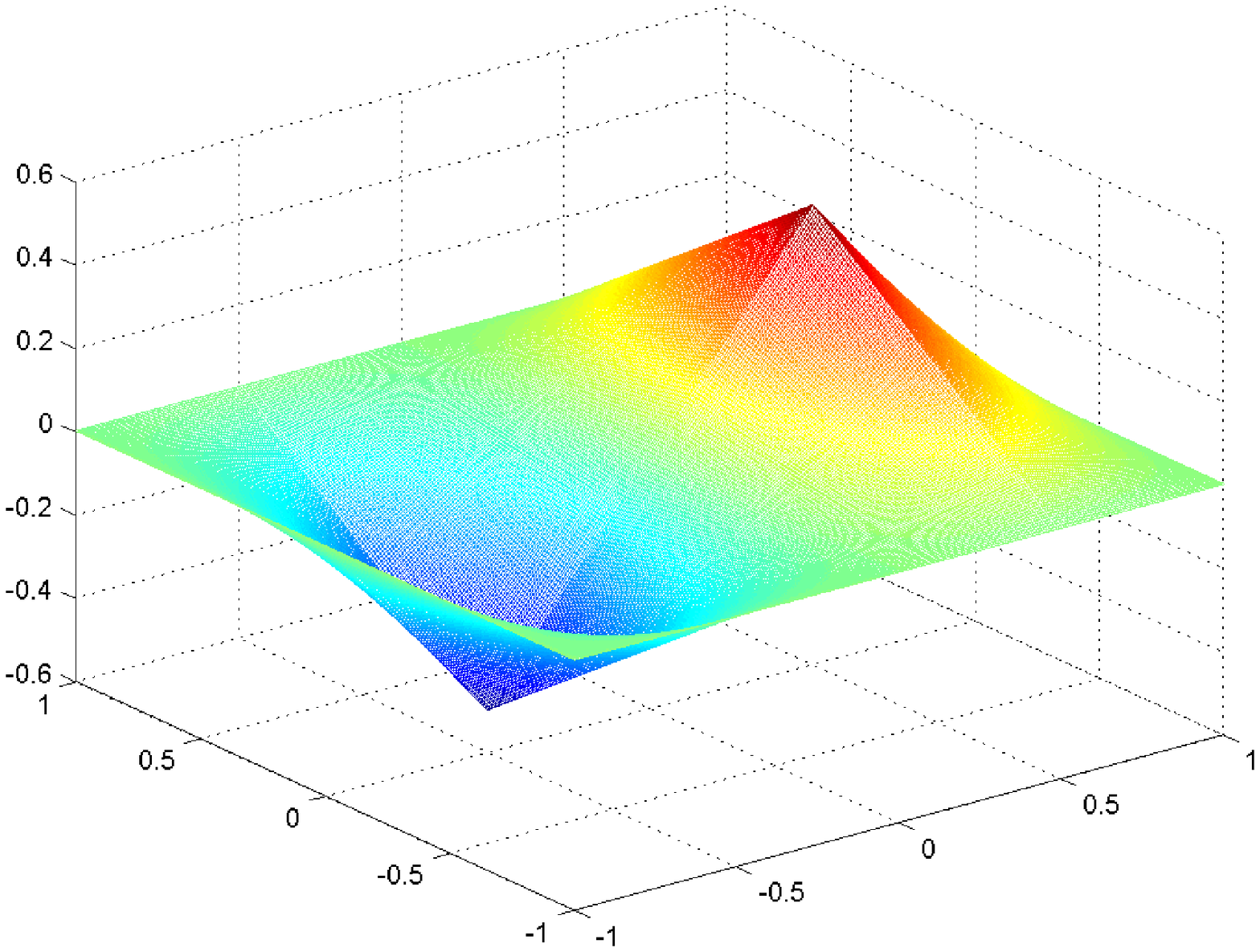}}
\subfigure{
\includegraphics[width=1.5in,height=1.0in]
{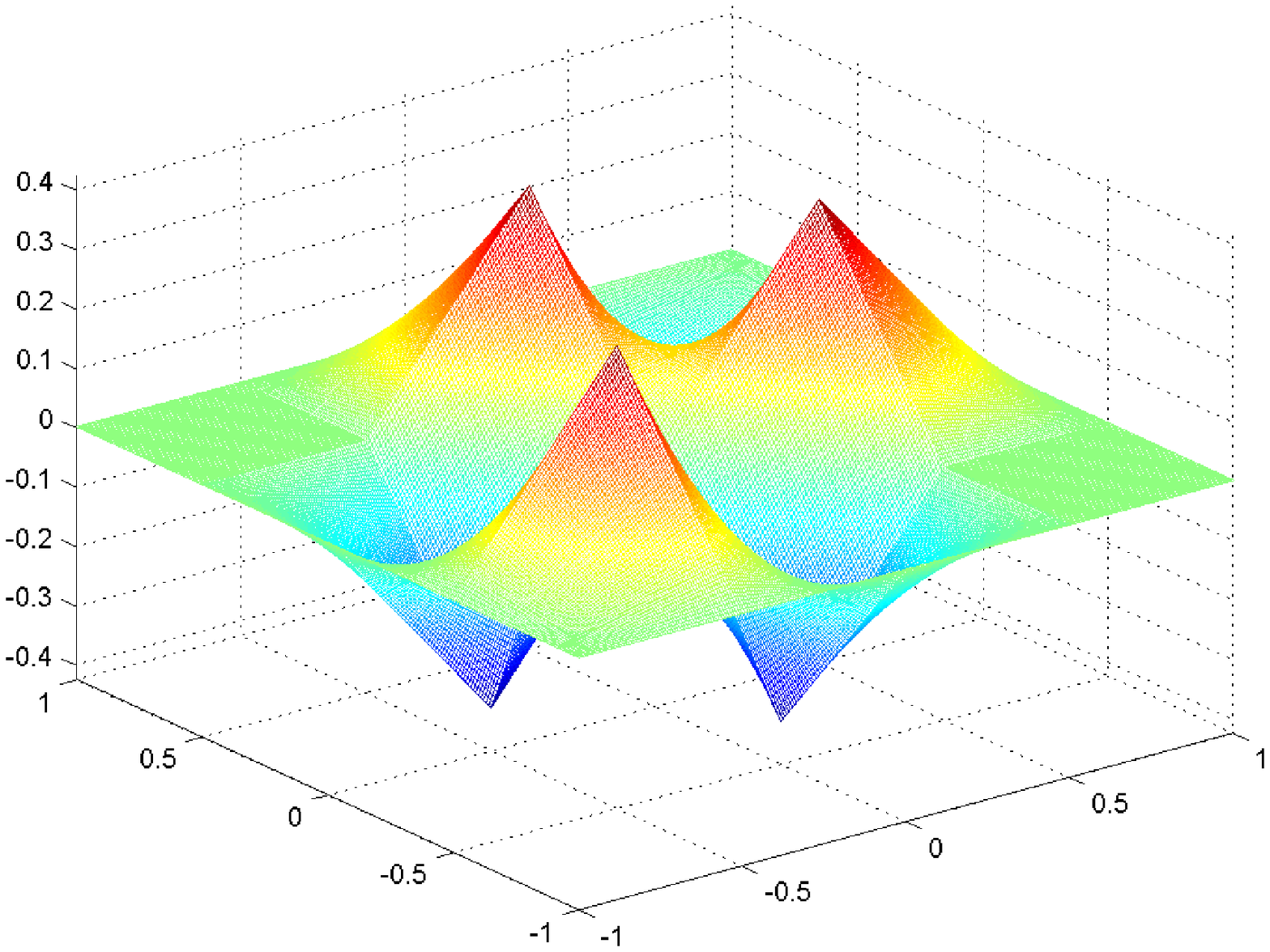}}
\subfigure{
\includegraphics[width=1.5in,height=1.0in]
{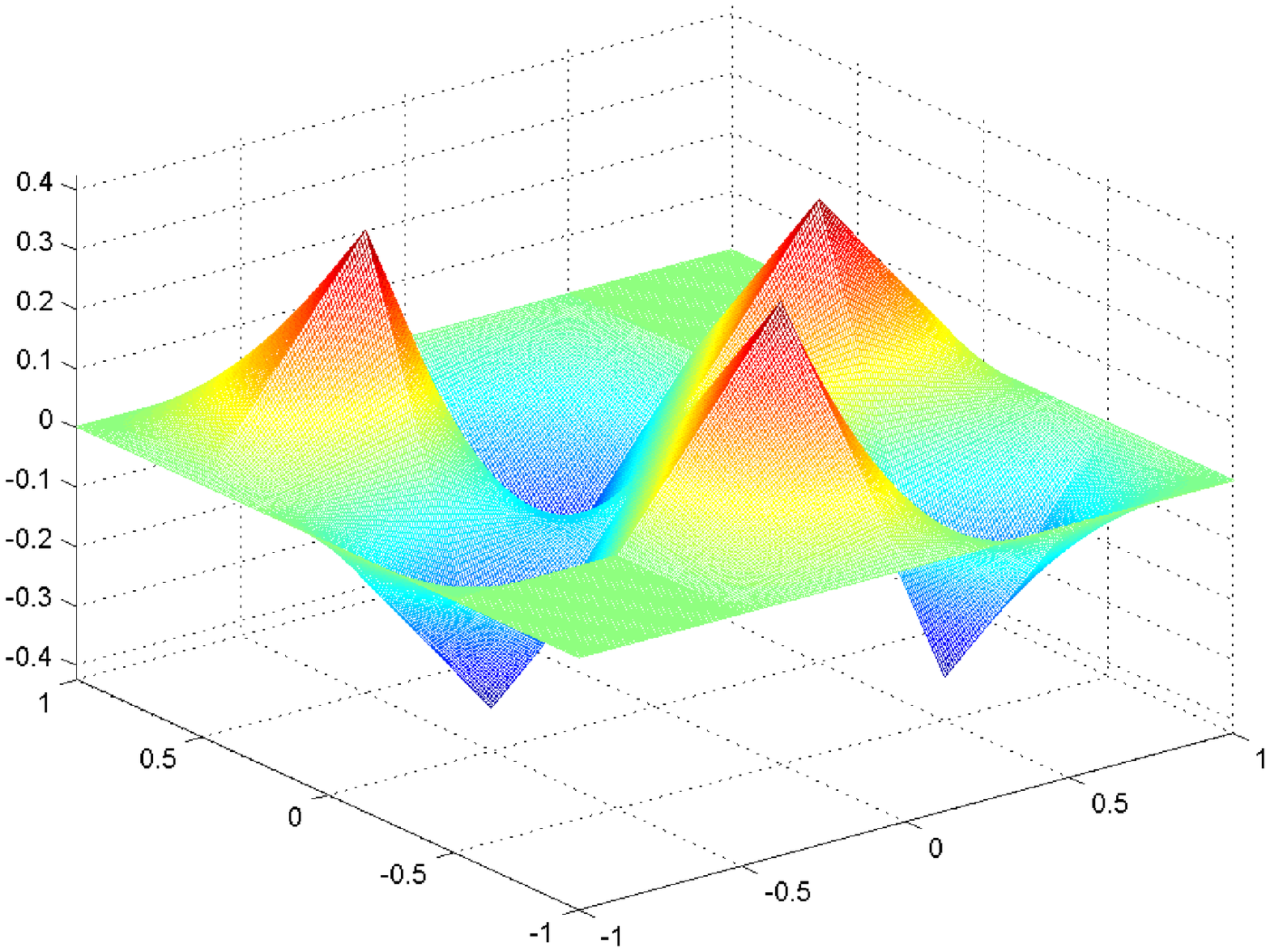}}\\
\subfigure{
\includegraphics[width=0.7in,height=0.7in]
{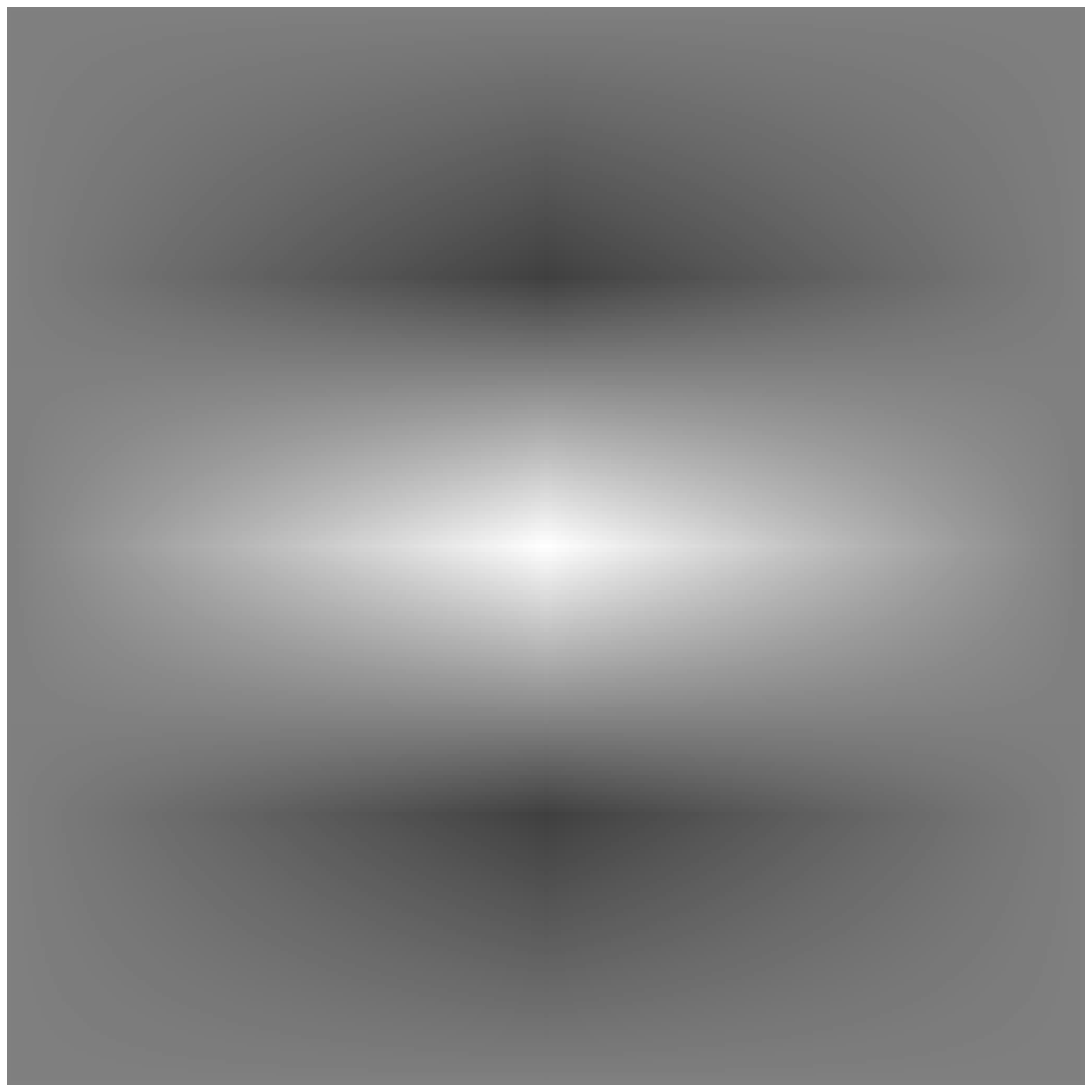}}
\subfigure{
\includegraphics[width=0.7in,height=0.7in]
{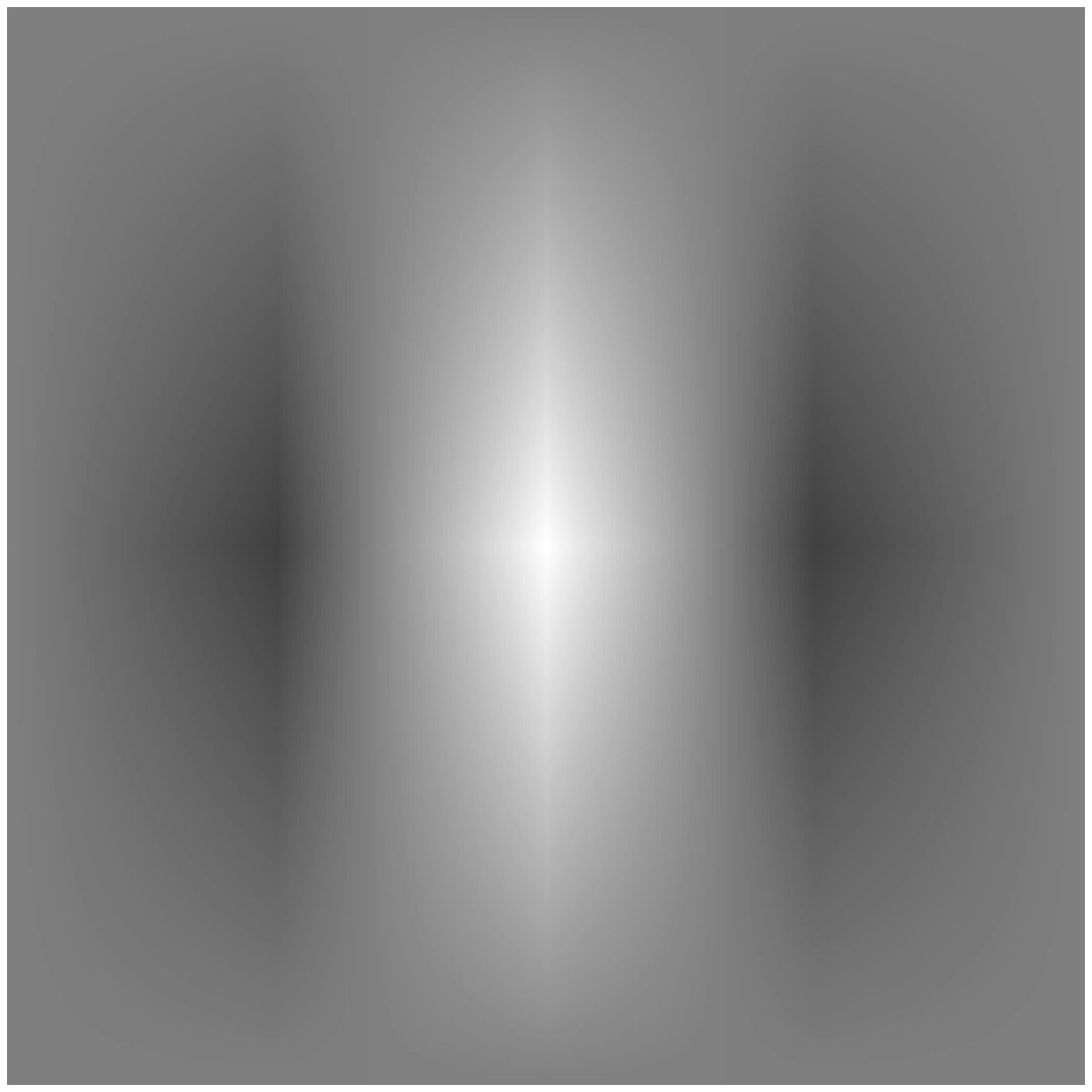}}
\subfigure{
\includegraphics[width=0.7in,height=0.7in]
{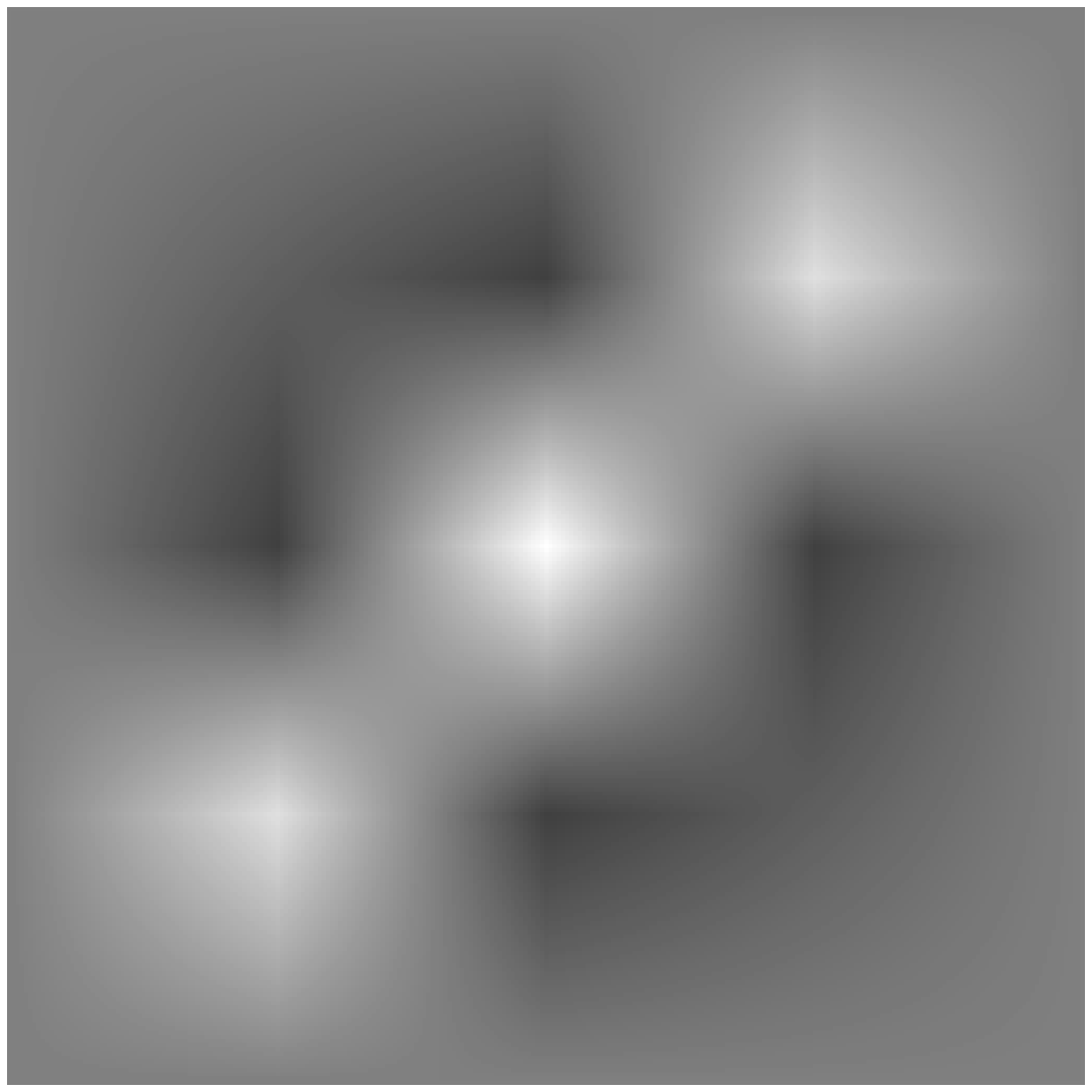}}
\subfigure{
\includegraphics[width=0.7in,height=0.7in]
{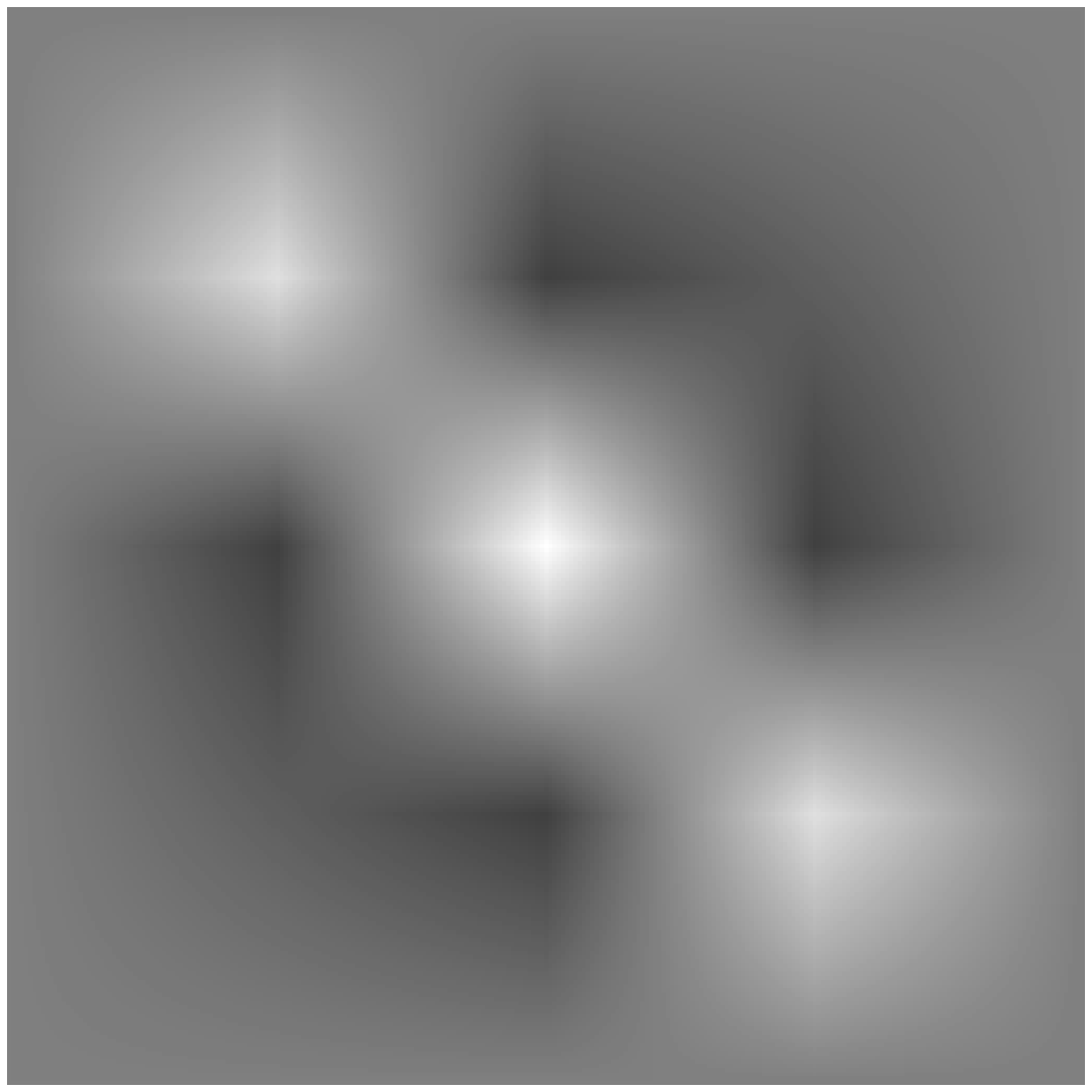}}
\subfigure{
\includegraphics[width=0.7in,height=0.7in]
{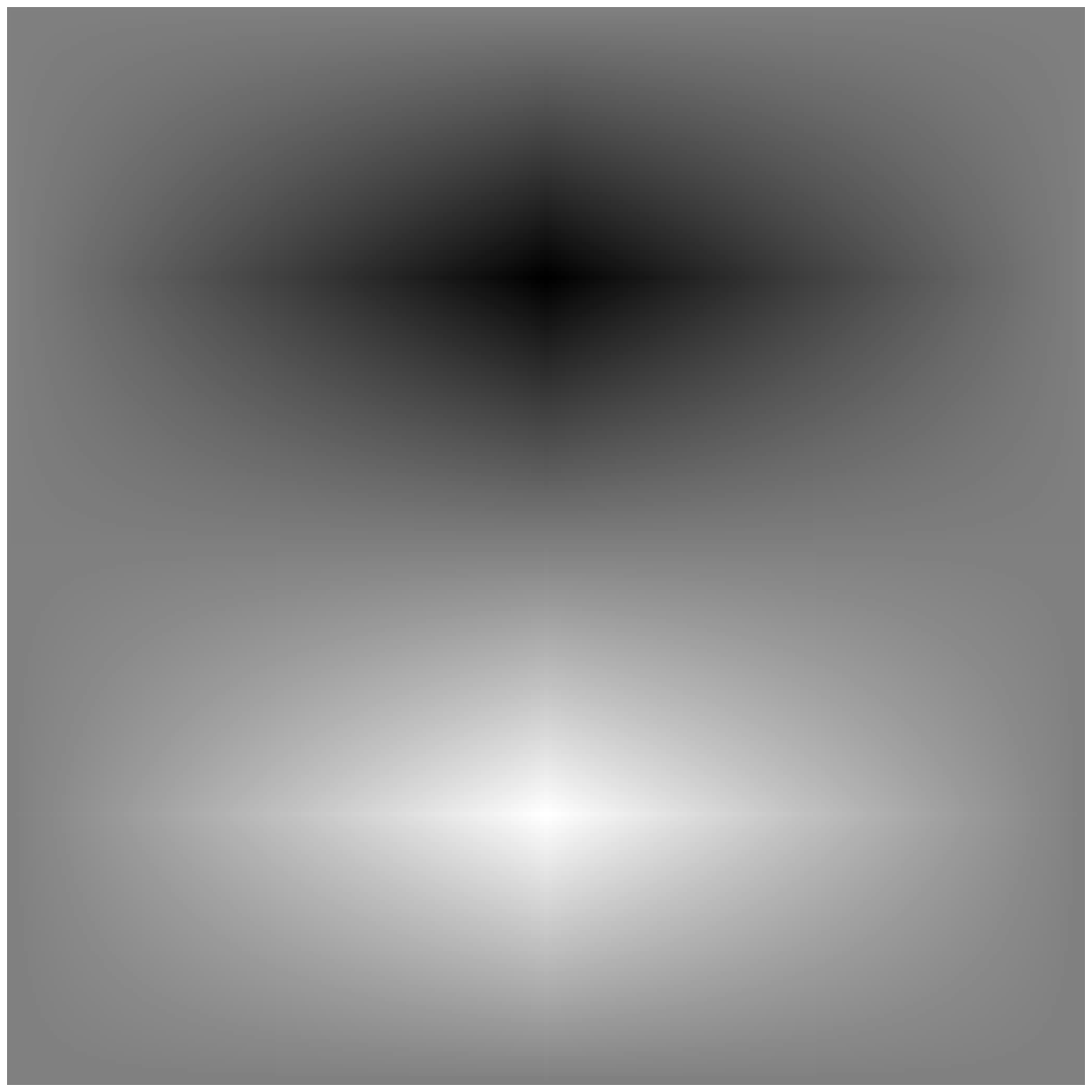}}
\subfigure{
\includegraphics[width=0.7in,height=0.7in]
{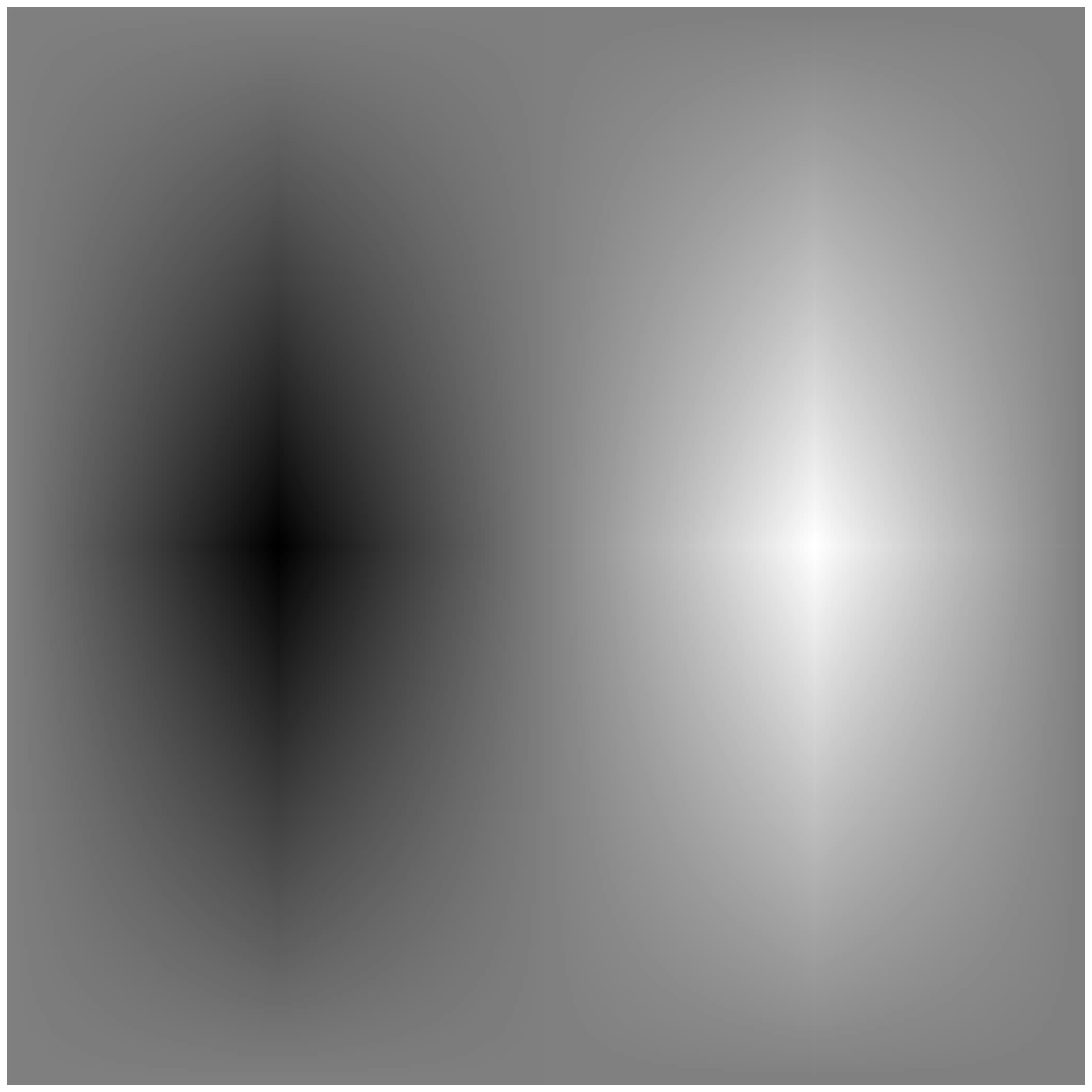}}
\subfigure{
\includegraphics[width=0.7in,height=0.7in]
{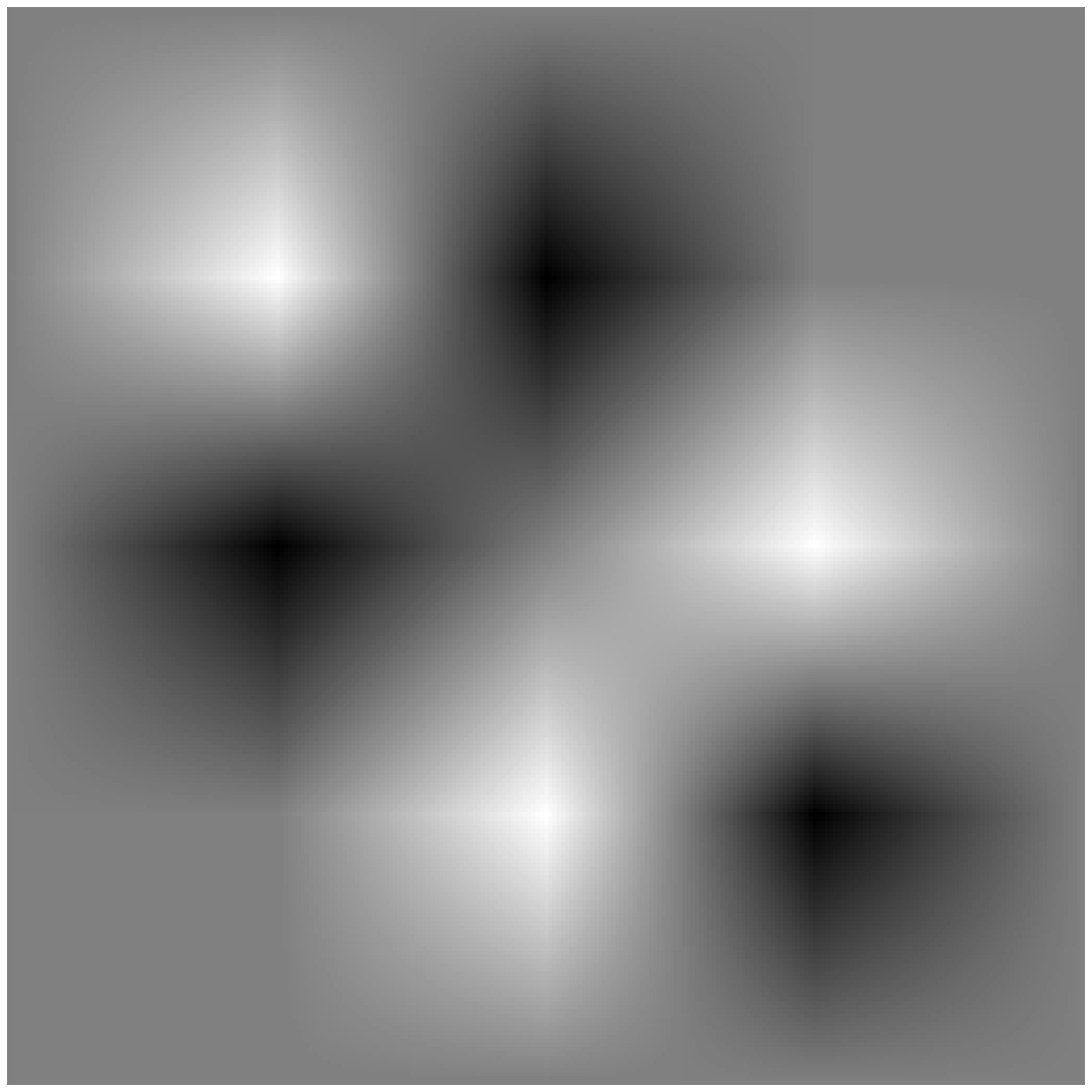}}
\subfigure{
\includegraphics[width=0.7in,height=0.7in]
{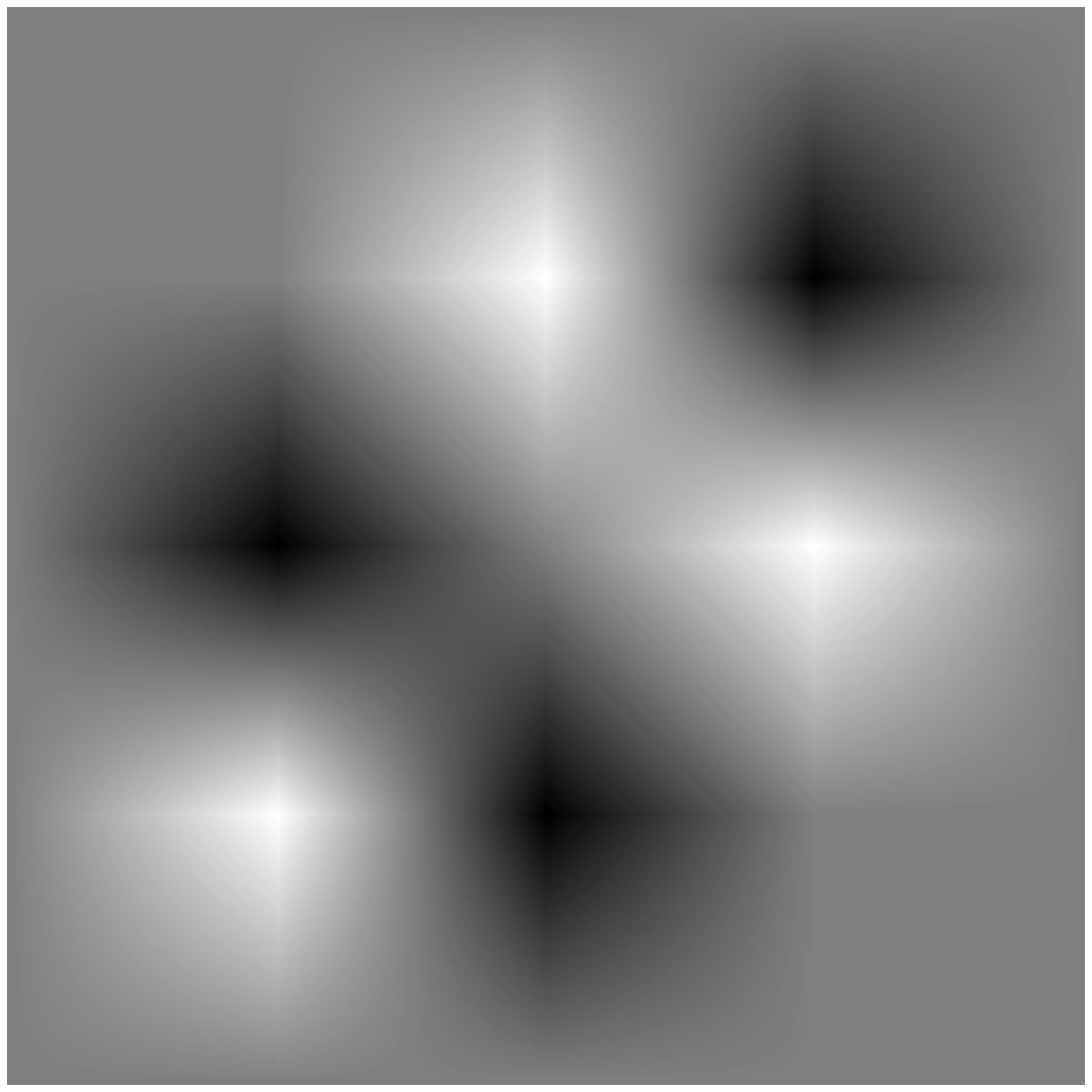}}
\begin{caption}{
The first row is for the real part and the second row is for the imaginary part of the tight framelet generators in Example~\ref{ex1} with $N=0$.
The third row is the greyscale image of the eight generators: the first four for real part and the last four for imaginary part.
} \label{fig1}
\end{caption}
\end{center}
\end{figure}

\begin{example}\label{ex2} {\rm Let $\pa(z) =z^{-2}(1 + z)^4/16=\{\tfrac{1}{16}, \tfrac{1}{4}, \tfrac{3}{8}, \tfrac{1}{4}, \tfrac{1}{16} \}_{[-2,2]}$ be the B-spline filter of order $4$.
Using Algorithm~\ref{alg:tffb}, we obtain a tight framelet filter bank $\{a; b_1, b_2\}$  with
\begin{align*}
&\pb_1(z) =\frac{\sqrt{34+8\sqrt{14}}(\sqrt{14}-4)}{2080}(1-z)[65z^3+(64\sqrt{14}+261)z^2+(40\sqrt{14}+155)z
+8\sqrt{14}+31],\\
&\pb_2(z) =\frac{\sqrt{34+8\sqrt{14}}(4\sqrt{14}-17)}{1300}(1-z) [ 10z^2-(5\sqrt{14}+15)z-\sqrt{14}-3].
\end{align*}
Applying Algorithm~\ref{alg:main} with $N =0$, we have a finitely supported complex tight framelet filter bank $\{a; b^p, b^n\}$ with
$b^n=\ol{b^p}$ and
\begin{align*}
\pb^p(z) =&(-0.00557113140380+0.0731731460340i)z^{-2}+(-0.0222840645179+0.292693813728i)z^{-1}\\
&-(0.318362332504+0.258768579113i)+(0.307215151326-0.0833740786820i)z\\
&+(0.0389934625526-0.0237234566220i)z^2.
\end{align*}
By calculation we have $d_{\R}=\frac{93}{128}\pi \approx 2.28256$, $d_A\approx 0.00187$, and $d_B\approx 0.762678$. 
If we take $N=2$, then
\begin{align*}
\pb^p(z)=&(0.0136421172460-0.00936826775525i)z^{-4}+(0.0545694833985-0.0374729096370i)z^{-3}\\
&-(0.117756260732-0.0384187816047i)z^{-2}+(0.176658675556-0.291095343052i)z^{-1}\\
&(0.215356267335+0.333766056656i)-(0.226650692255-0.0670707536790i)z\\
&-(0.0454230034494-0.00120115849369i)z^2-(0.0601885689225+0.0876476822545i)z^3\\
&-(0.0102063889665+0.0148634718020i)z^4.
\end{align*}
By calculation, we have $d_B\approx 0.283860$. 
See Figure~\ref{fig2} for the graphs of the eight tight framelet generators in the associated two-dimensional real-valued tight framelet for $L_2(\R^2)$ in \eqref{tp:ctf}.
} \end{example}

\begin{figure}[hb]
\begin{center}
\subfigure{
\includegraphics[width=1.5in,height=1.0in]
{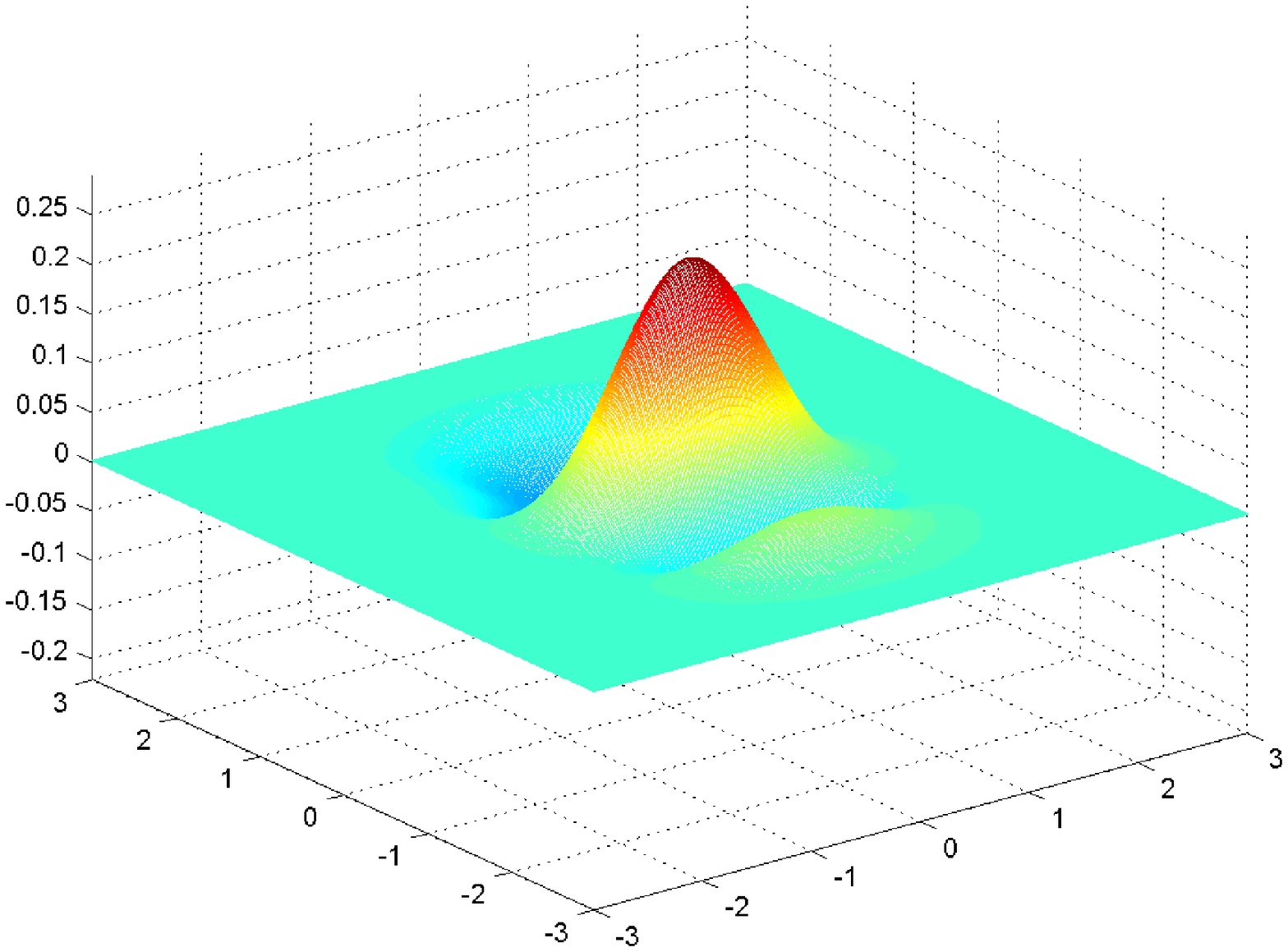}}
\subfigure{
\includegraphics[width=1.5in,height=1.0in]
{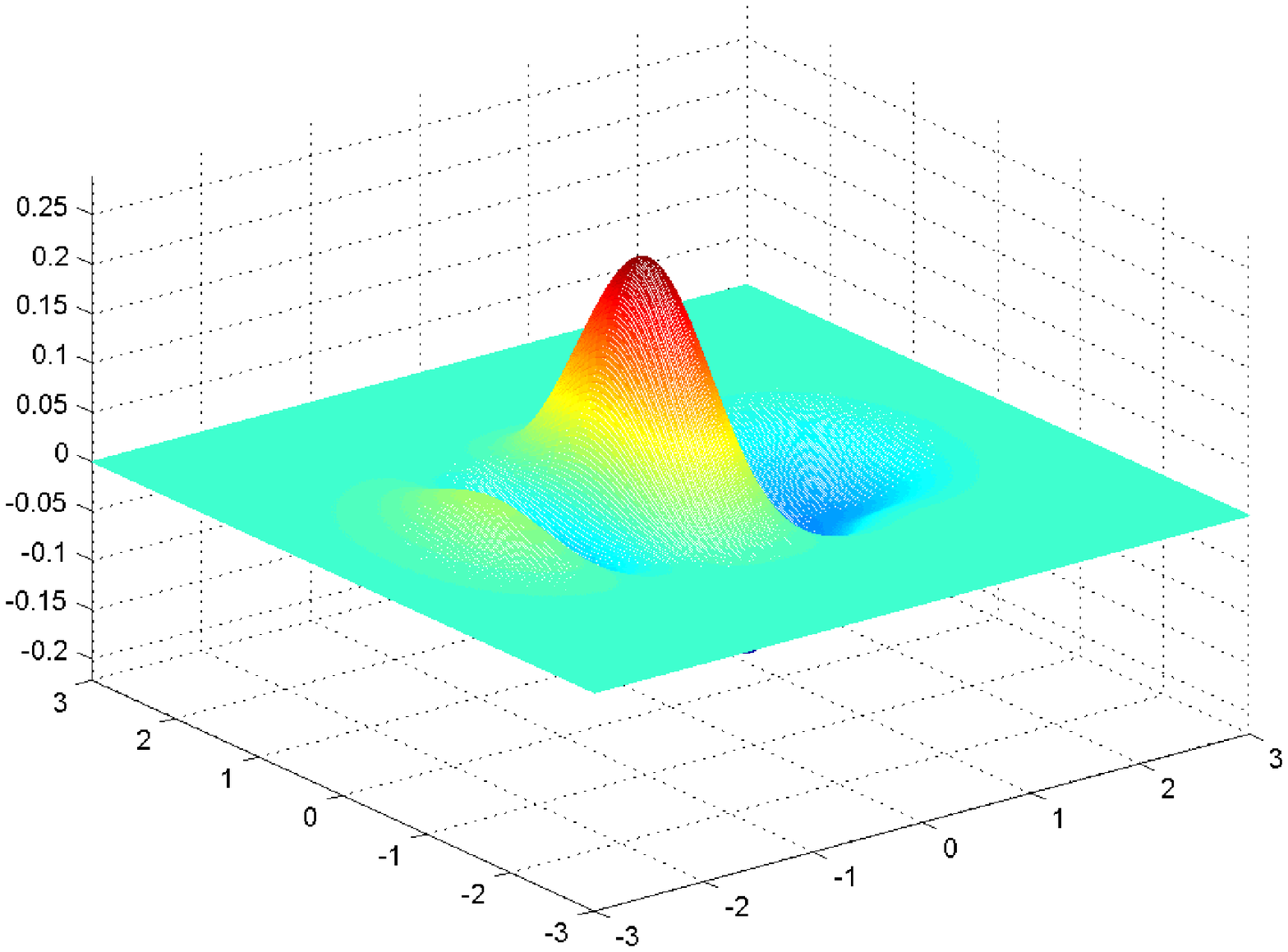}}
\subfigure{
\includegraphics[width=1.5in,height=1.0in]
{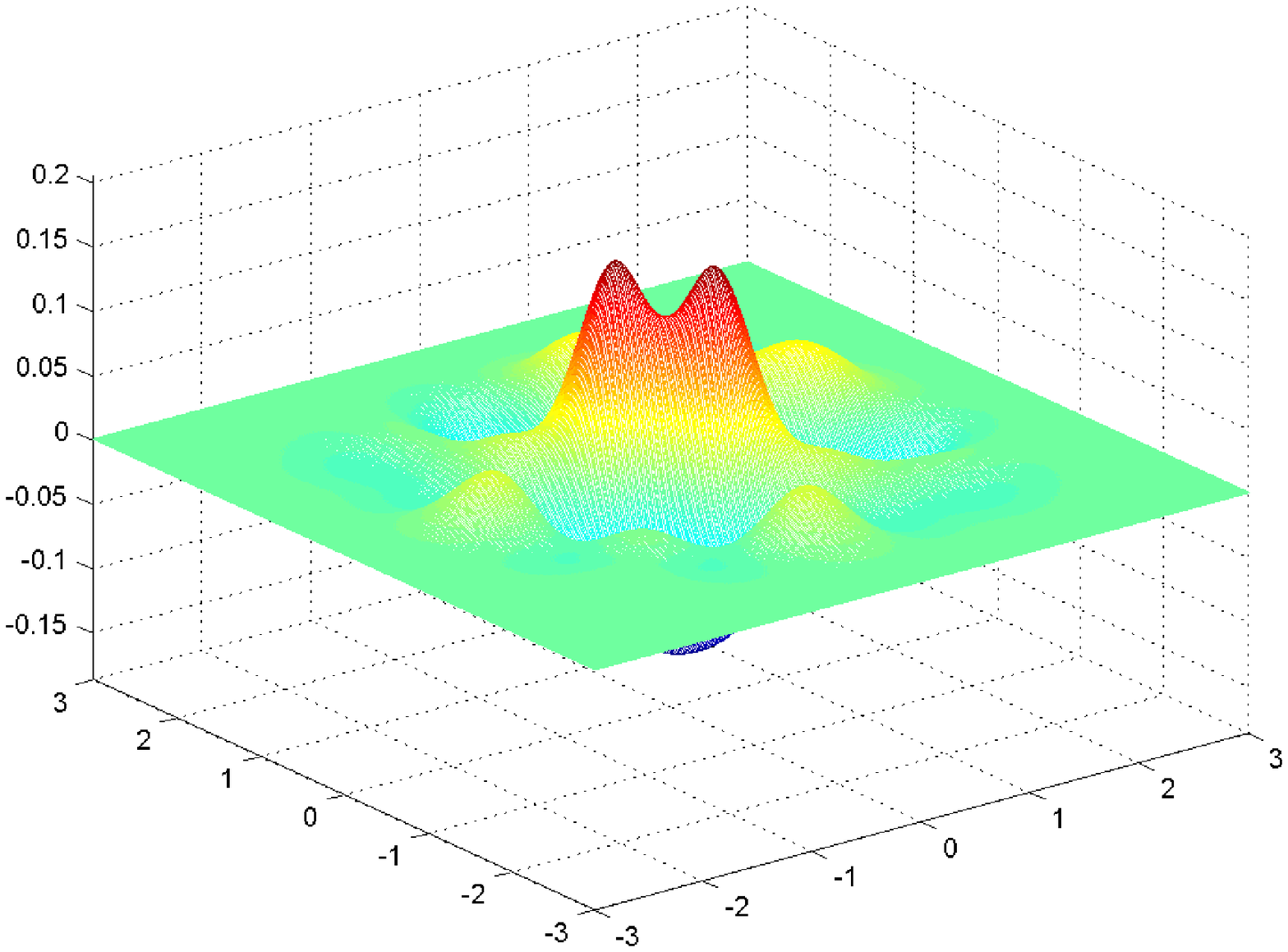}}
\subfigure{
\includegraphics[width=1.5in,height=1.0in]
{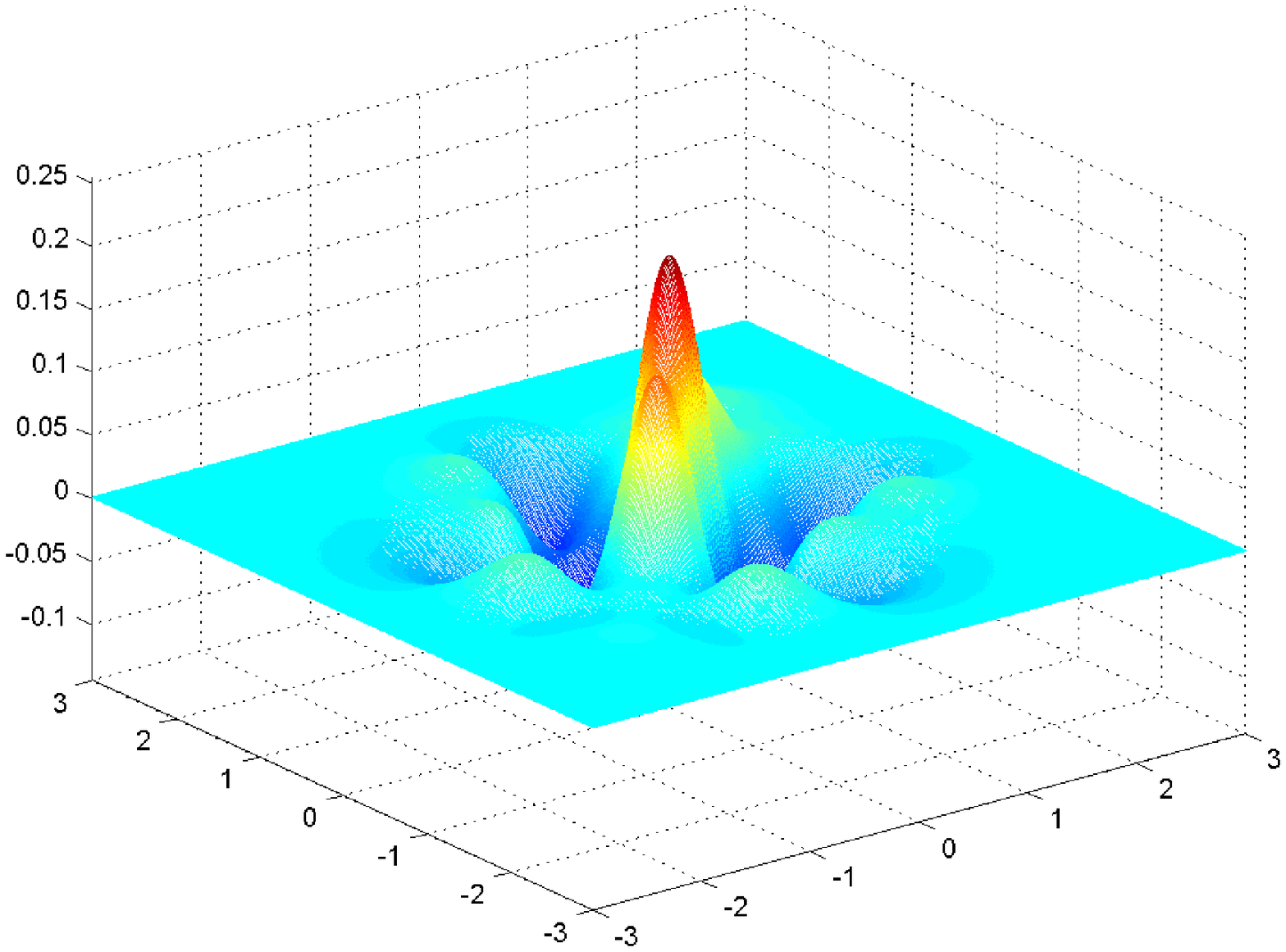}}\\
\subfigure{
\includegraphics[width=1.5in,height=1.0in]
{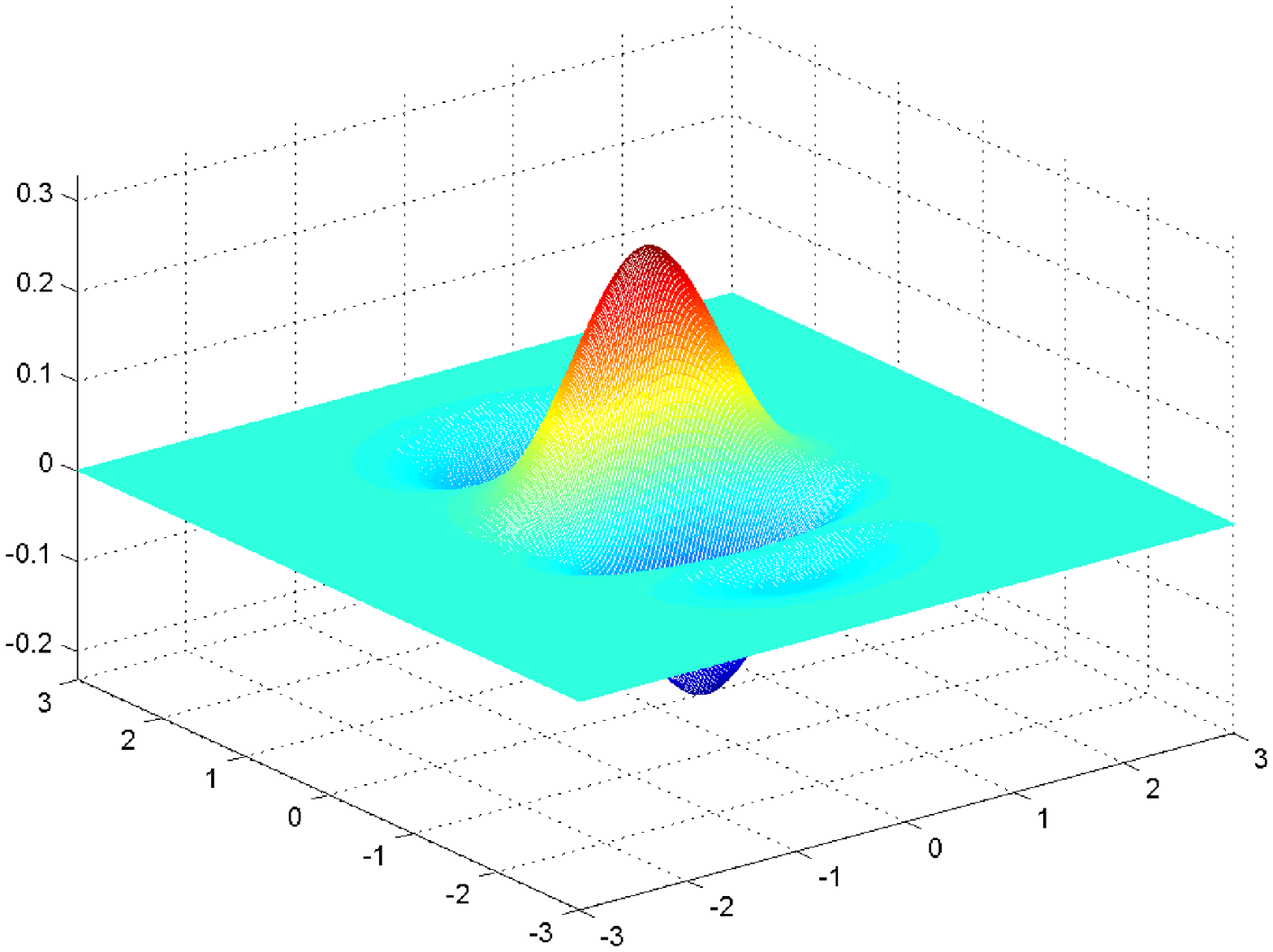}}
\subfigure{
\includegraphics[width=1.5in,height=1.0in]
{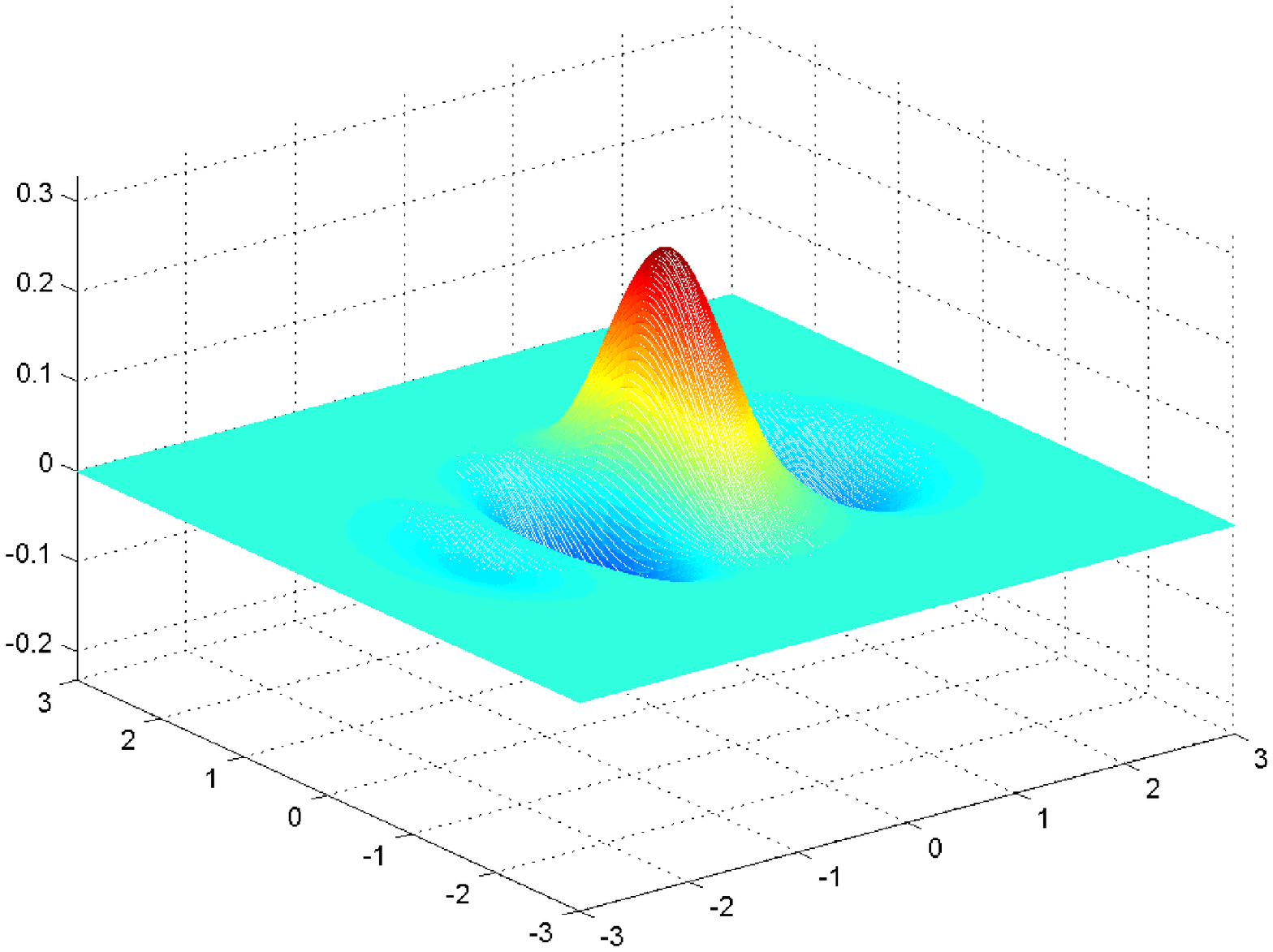}}
\subfigure{
\includegraphics[width=1.5in,height=1.0in]
{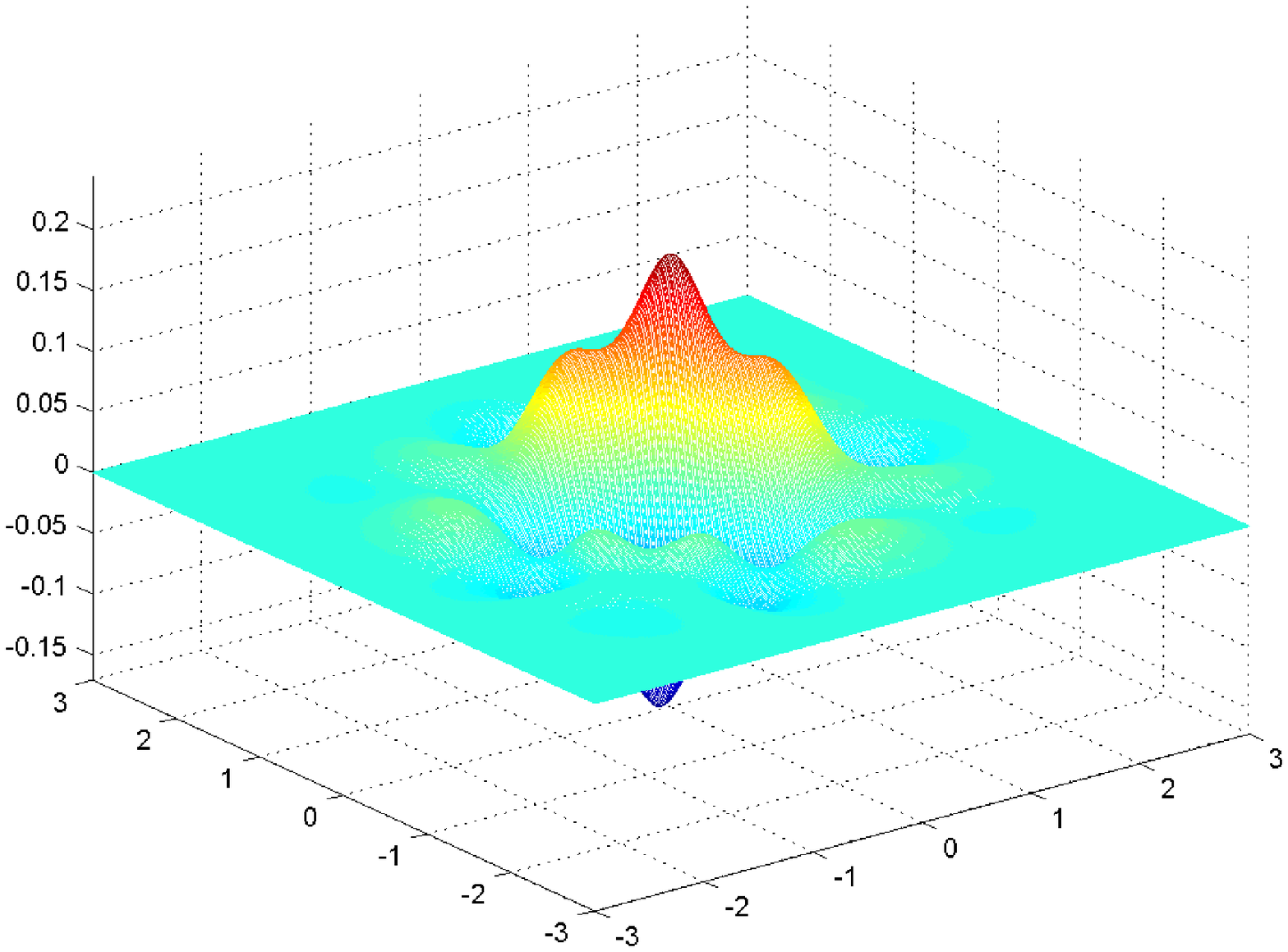}}
\subfigure{
\includegraphics[width=1.5in,height=1.0in]
{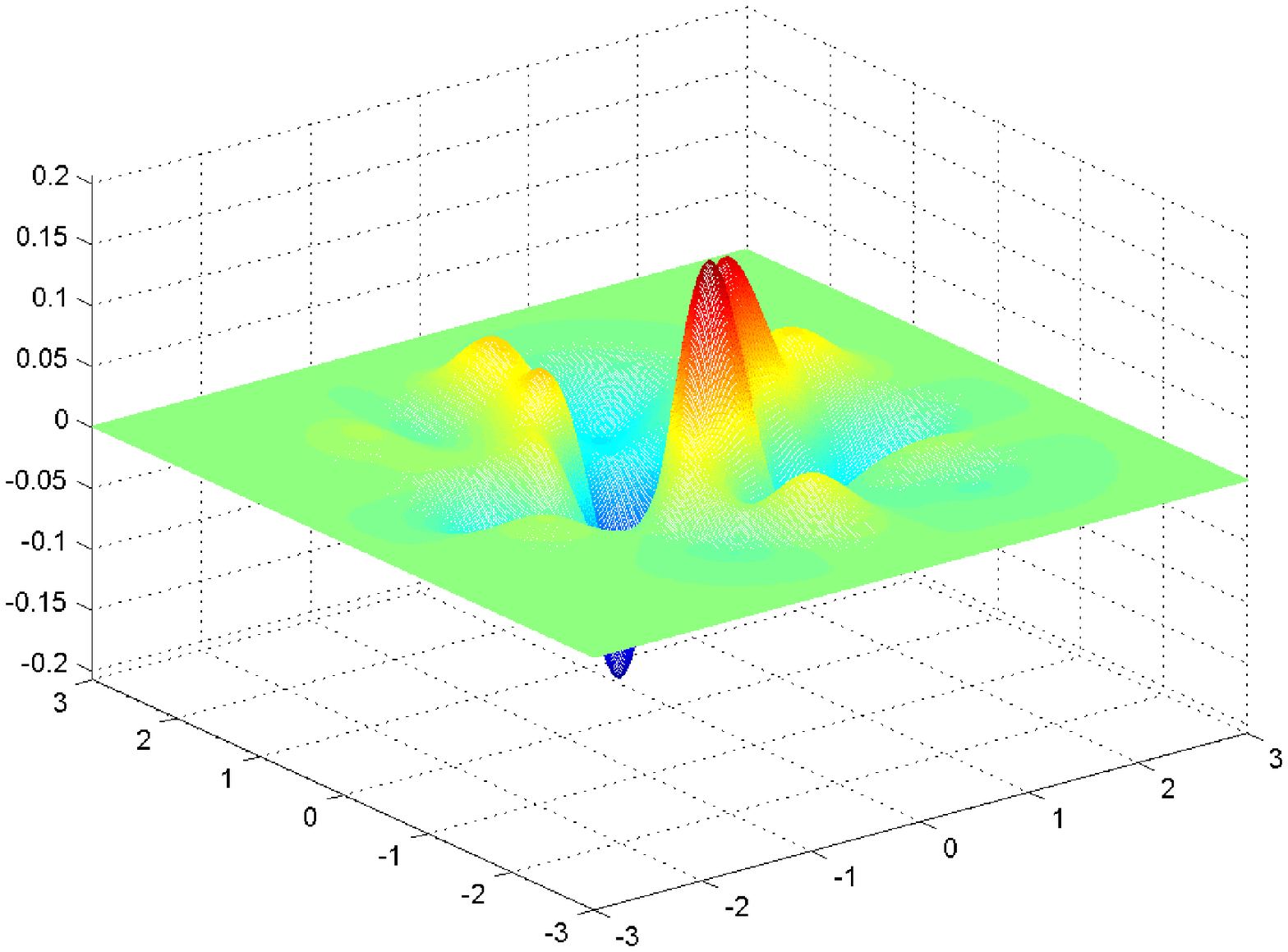}}\\
\subfigure{
\includegraphics[width=0.7in,height=0.7in]
{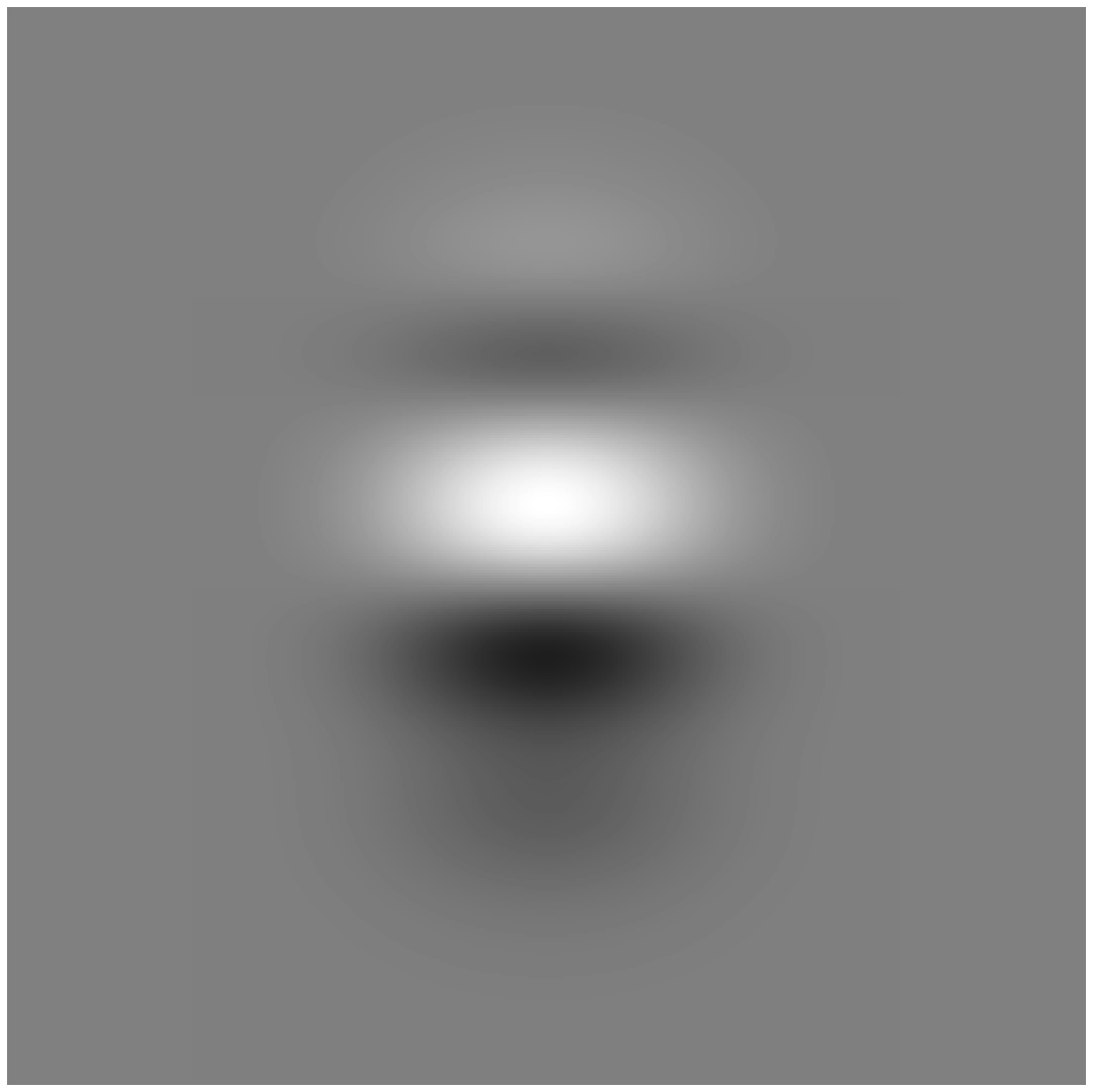}}
\subfigure{
\includegraphics[width=0.7in,height=0.7in]
{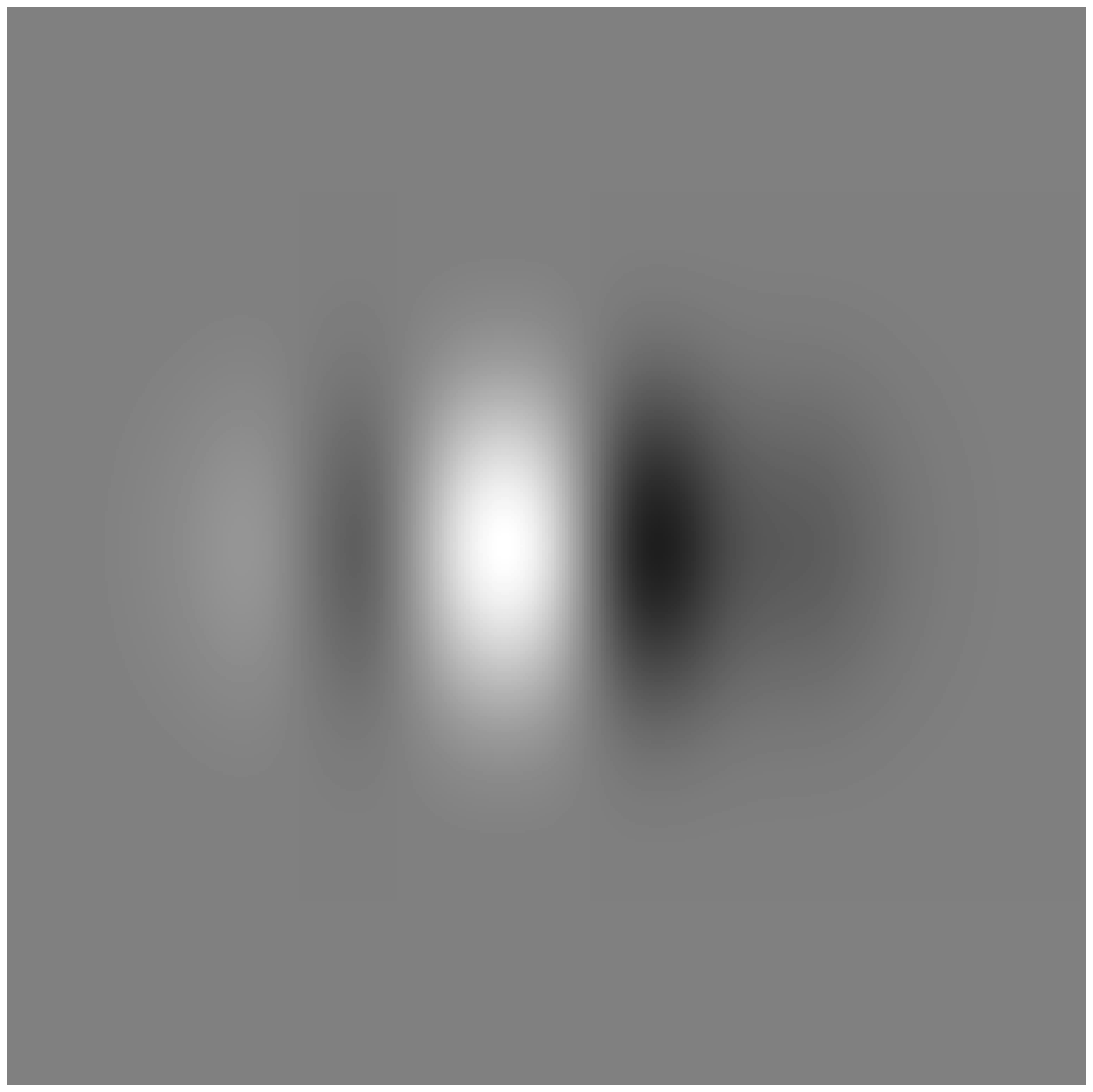}}
\subfigure{
\includegraphics[width=0.7in,height=0.7in]
{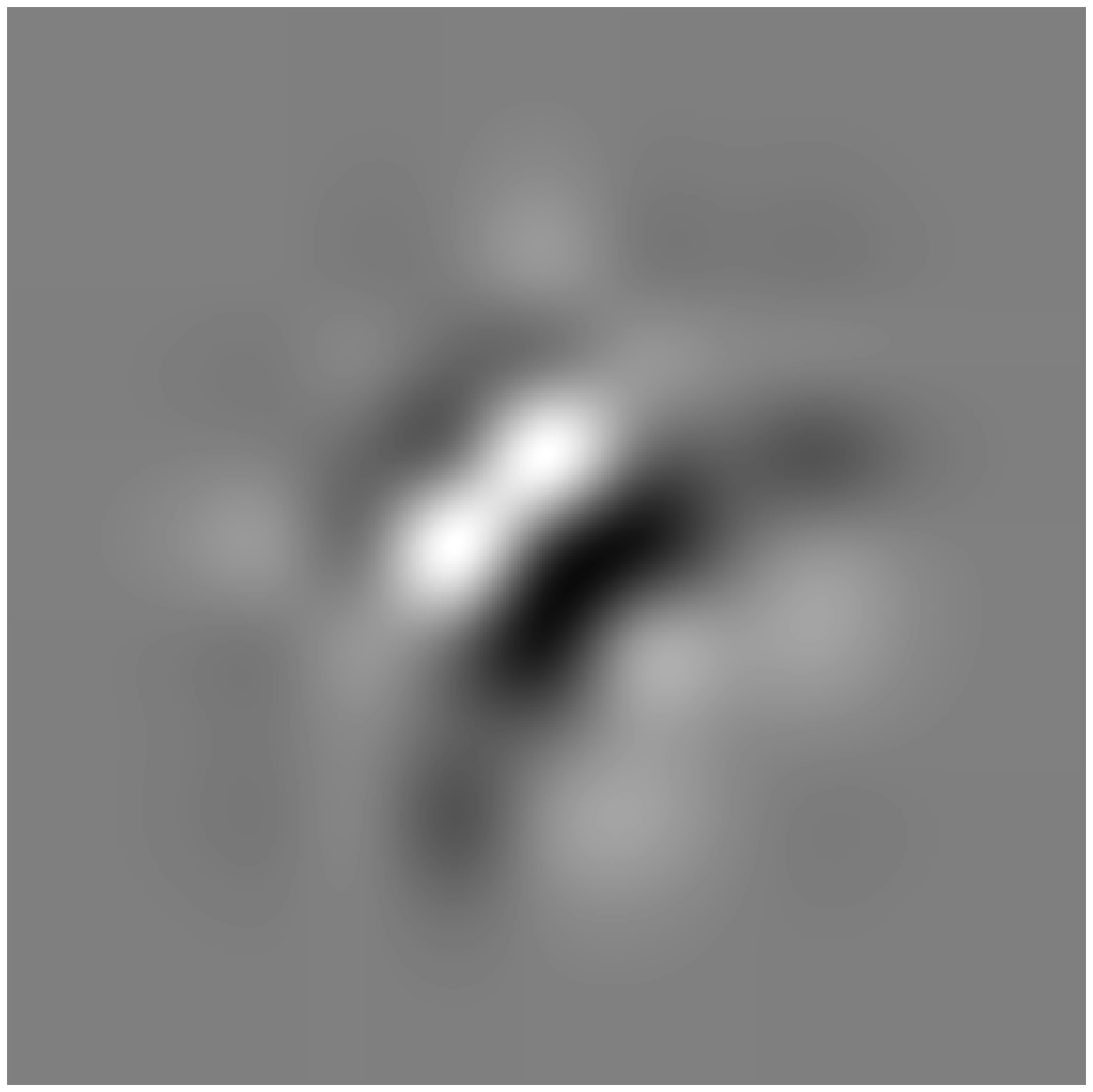}}
\subfigure{
\includegraphics[width=0.7in,height=0.7in]
{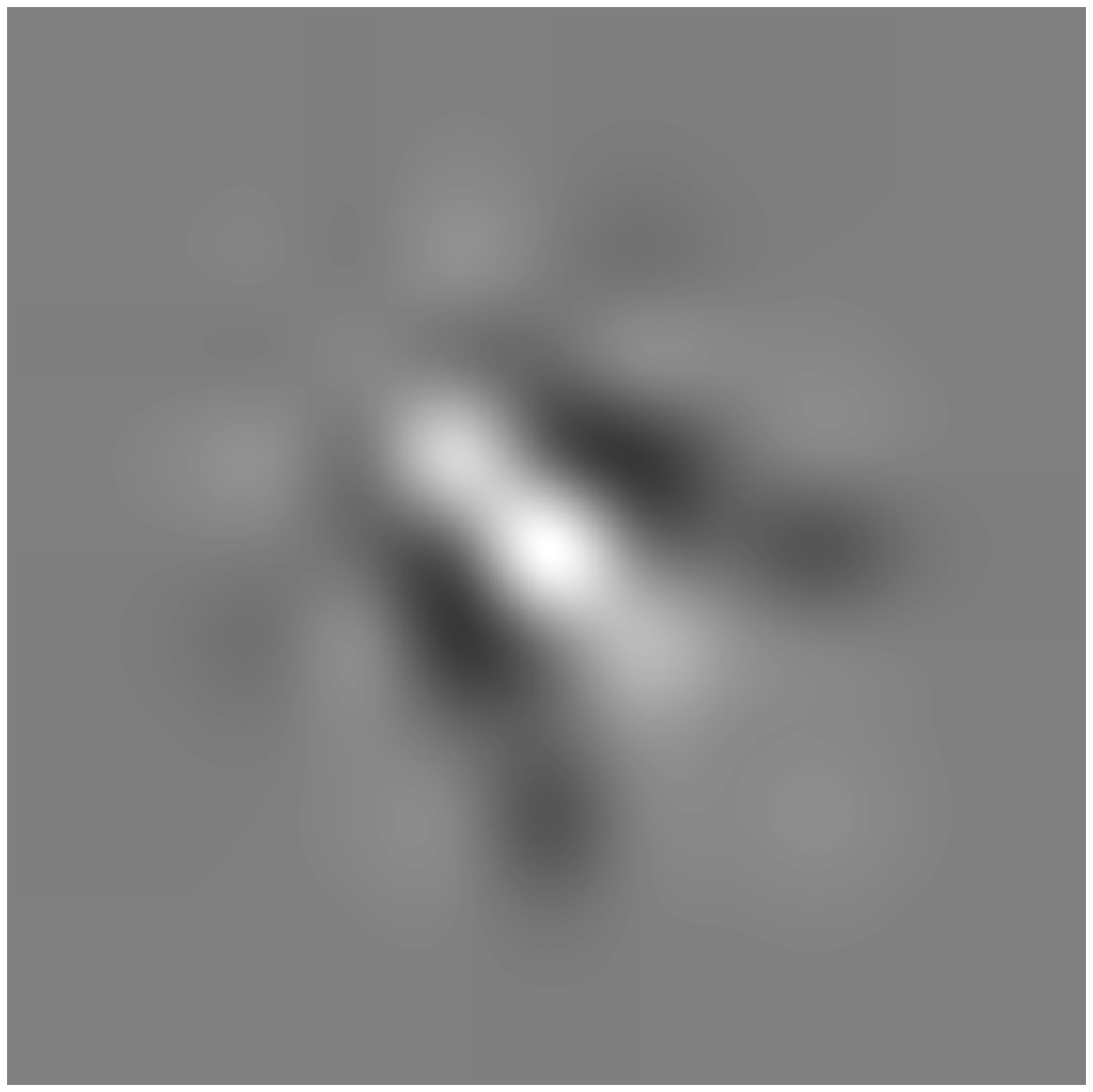}}
\subfigure{
\includegraphics[width=0.7in,height=0.7in]
{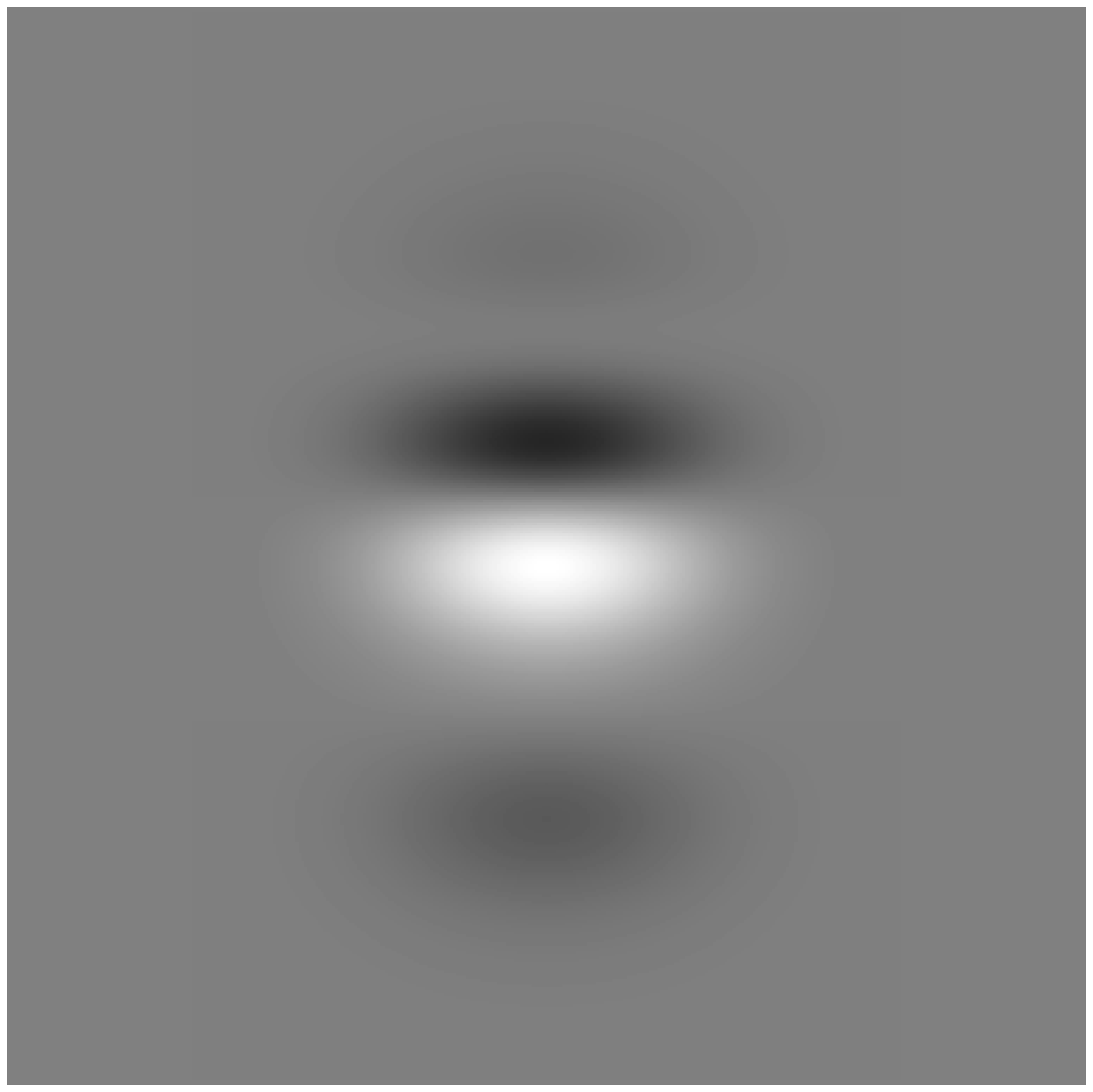}}
\subfigure{
\includegraphics[width=0.7in,height=0.7in]
{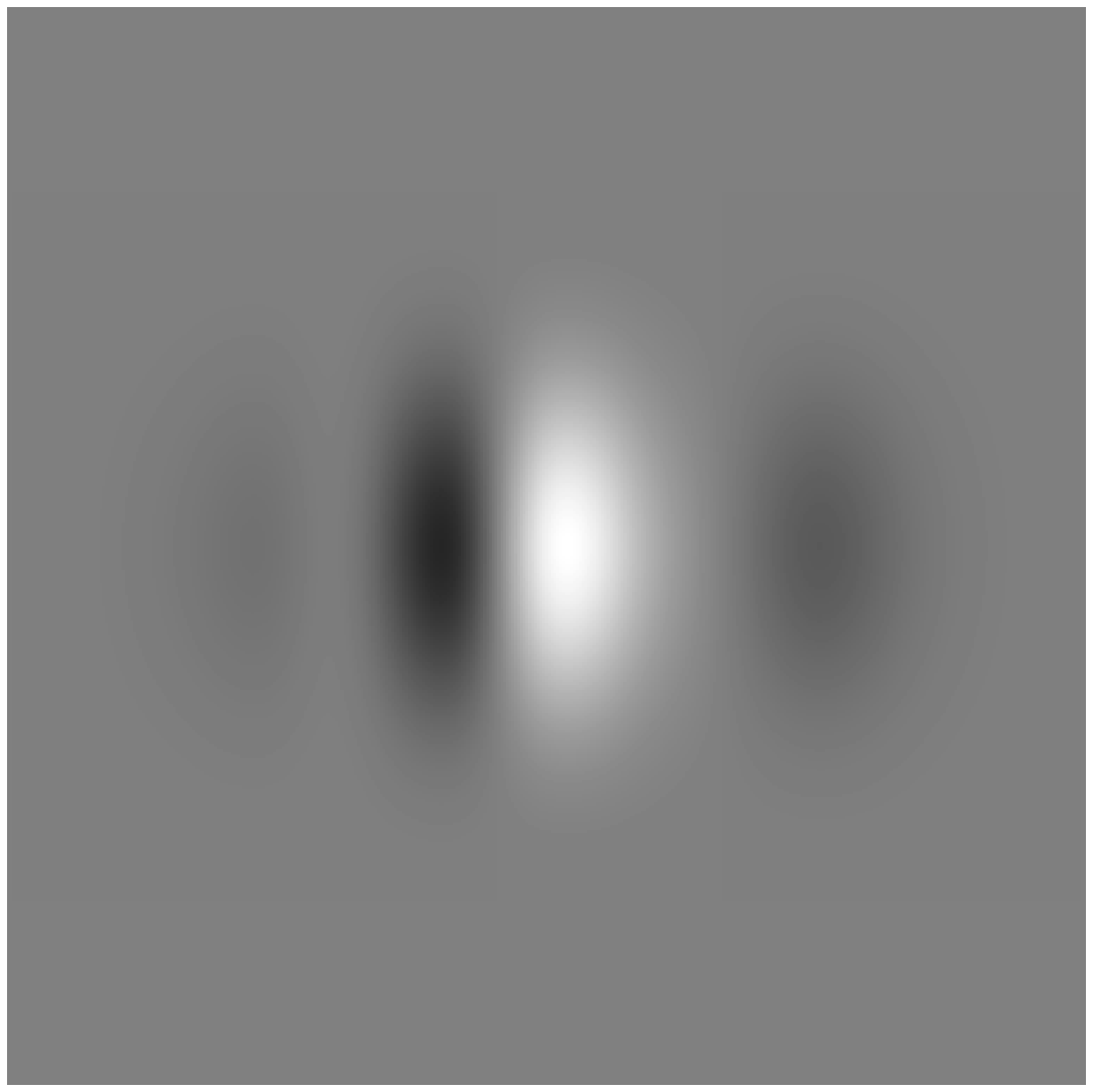}}
\subfigure{
\includegraphics[width=0.7in,height=0.7in]
{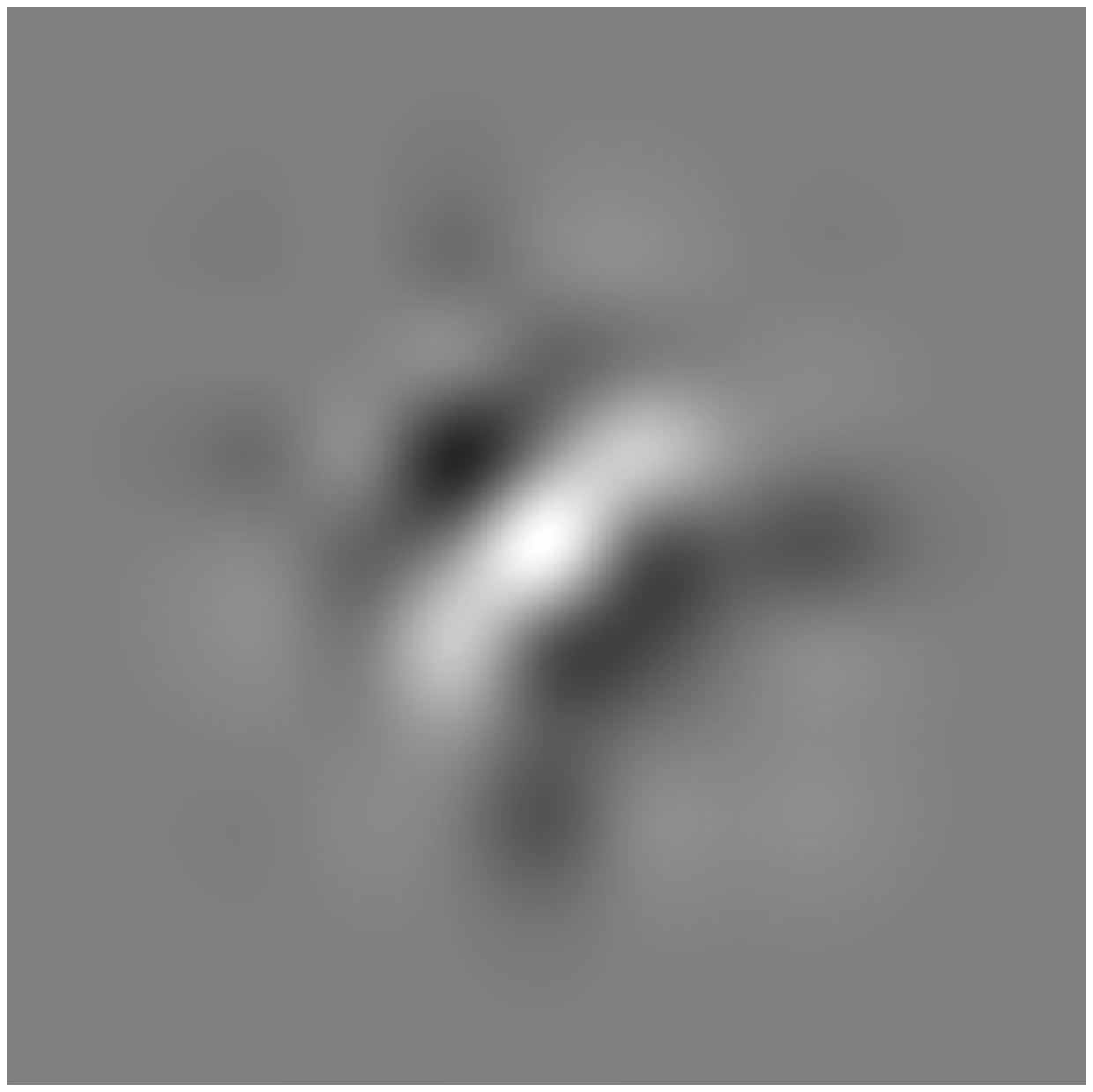}}
\subfigure{
\includegraphics[width=0.7in,height=0.7in]
{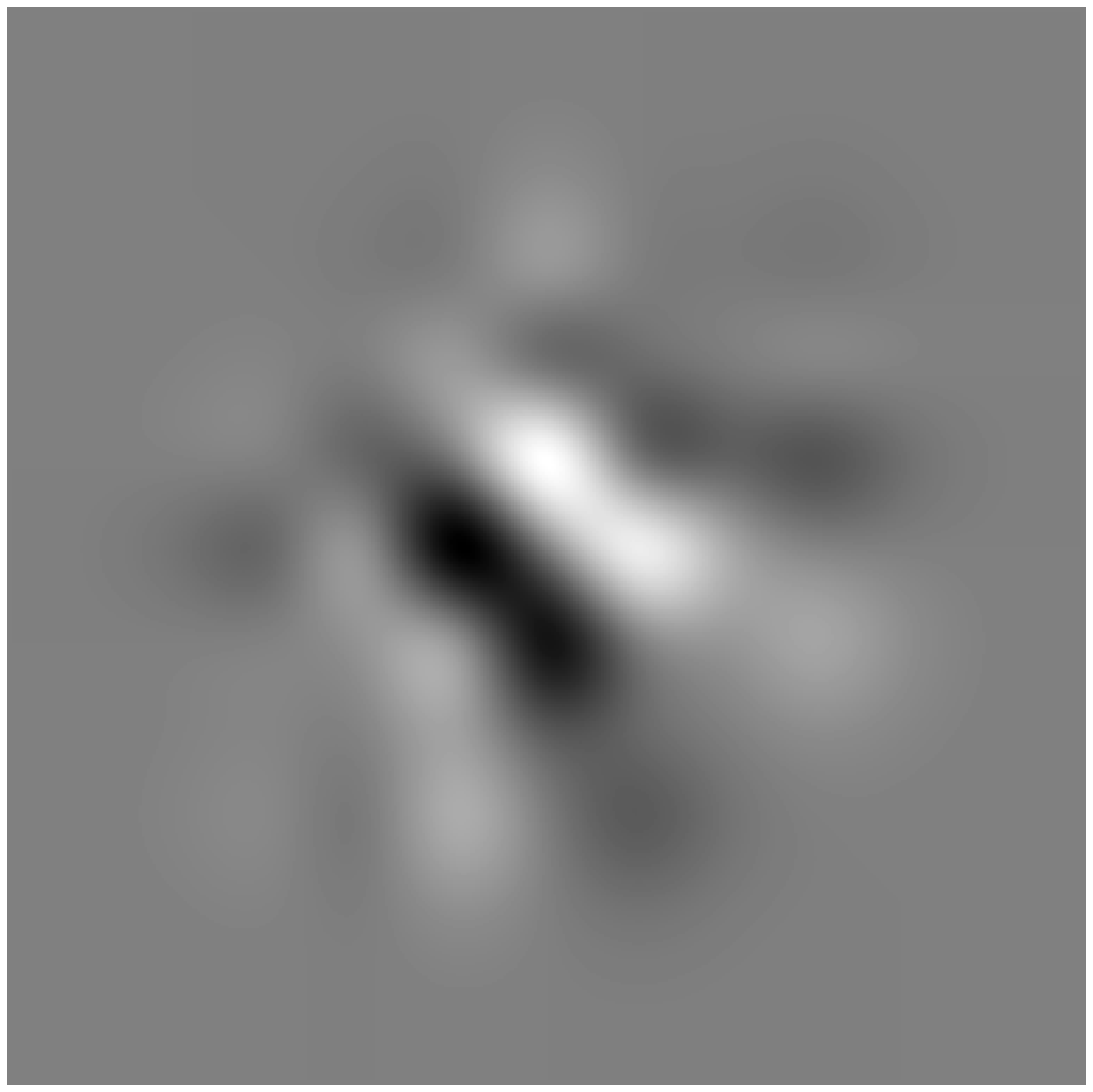}}
\begin{caption}{
The first row is for the real part and the second row is for the imaginary part of the tight framelet generators in Example~\ref{ex2} with $N=2$.
The third row is the greyscale image of the eight generators:
the first four for real part and the last four for imaginary part.
} \label{fig2}
\end{caption}
\end{center}
\end{figure}

\begin{example}\label{ex3} {\rm Let $\pa(z) =-\frac{1}{32}z^{-3}+\frac{9}{32}z^{-1}+\frac{1}{2}+\frac{9}{32}z-\frac{1}{32}z^3=\{-\tfrac{1}{32},0, \tfrac{9}{32}, \tfrac{1}{2}, \tfrac{9}{32},0, -\tfrac{1}{32} \}_{[-3,3]}$ be an interpolatory filter.
Using Algorithm~\ref{alg:tffb}, we obtain a tight framelet filter bank $\{a; b_1, b_2\}$ with
\begin{align*}
&\pb_1(z) =\frac{\sqrt{298527-142344\sqrt{3}}(72\sqrt{3}+151)}{458600736}z^{-3}(z-1)^2(z+2-\sqrt{3})[
1977z^3+(512+57\sqrt{3})z^2\\
&\qquad +(21+86\sqrt{3})z-86-7\sqrt{3}],\\
&\pb_2(z) =\frac{\sqrt{298527-142344\sqrt{3}}(2\sqrt{2}+\sqrt{6})}{173976}z^{-3}(z-1)^2(z+2-\sqrt{3})[
-44z^2+(\sqrt{3}-6)z+2\sqrt{3}-1].
\end{align*}
Applying Algorithm~\ref{alg:main} with $N =0$, we have a finitely supported complex tight framelet filter bank $\{a; b^p, b^n\}$ with
$b^n=\ol{b^p}$ and
\begin{align*}
\pb^p(z) =&(0.000765760176753+0.00404161855341i)z^{-3}-(0.0403653729400+0.0880450827053i)z^{-1}\\
&-(0.0122521628281+0.0646658968547i)+(0.267462323473+0.228631206605i)z\\
&-(0.341301227764-0.0646658968553i)z^2+(0.125690679881-0.144627742454i)z^3.
\end{align*}
By calculation we have $d_{\R}=\frac{151}{256}\pi \approx 1.85305$, $d_A\approx 0.03719$, and $d_B\approx 0.690756$. 
If we take $N=2$, then
\begin{align*}
\pb^p(z)=&(0.000127813163113+0.000468578346236i)z^{-5}-(0.0030678318507+0.0157028980677i)z^{-3}\\
&-(0.00204501060981+0.00749725353983i)z^{-2}+(-0.0374047192912+0.0481138677939i)z^{-1}\\
&-(0.0665960959764+0.172855502748i)+(0.350214784761+0.131605792364i)z\\
&-(0.245342403088-0.169559360298i)z^2-(0.0151368278980+0.148441081755i)z^3\\
&-(0.0395698809180e-0.0107933959919i)z^4+(0.0588201717066-0.0160442586839i)z^5.
\end{align*}
By calculation, we have $d_B\approx 0.307271$. 
See Figure~\ref{fig3} for the graphs of the eight tight framelet generators in the associated two-dimensional real-valued tight framelet for $L_2(\R^2)$ in \eqref{tp:ctf}.
} \end{example}

\begin{figure}[ht]
\begin{center}
\subfigure{
\includegraphics[width=1.5in,height=1.0in]
{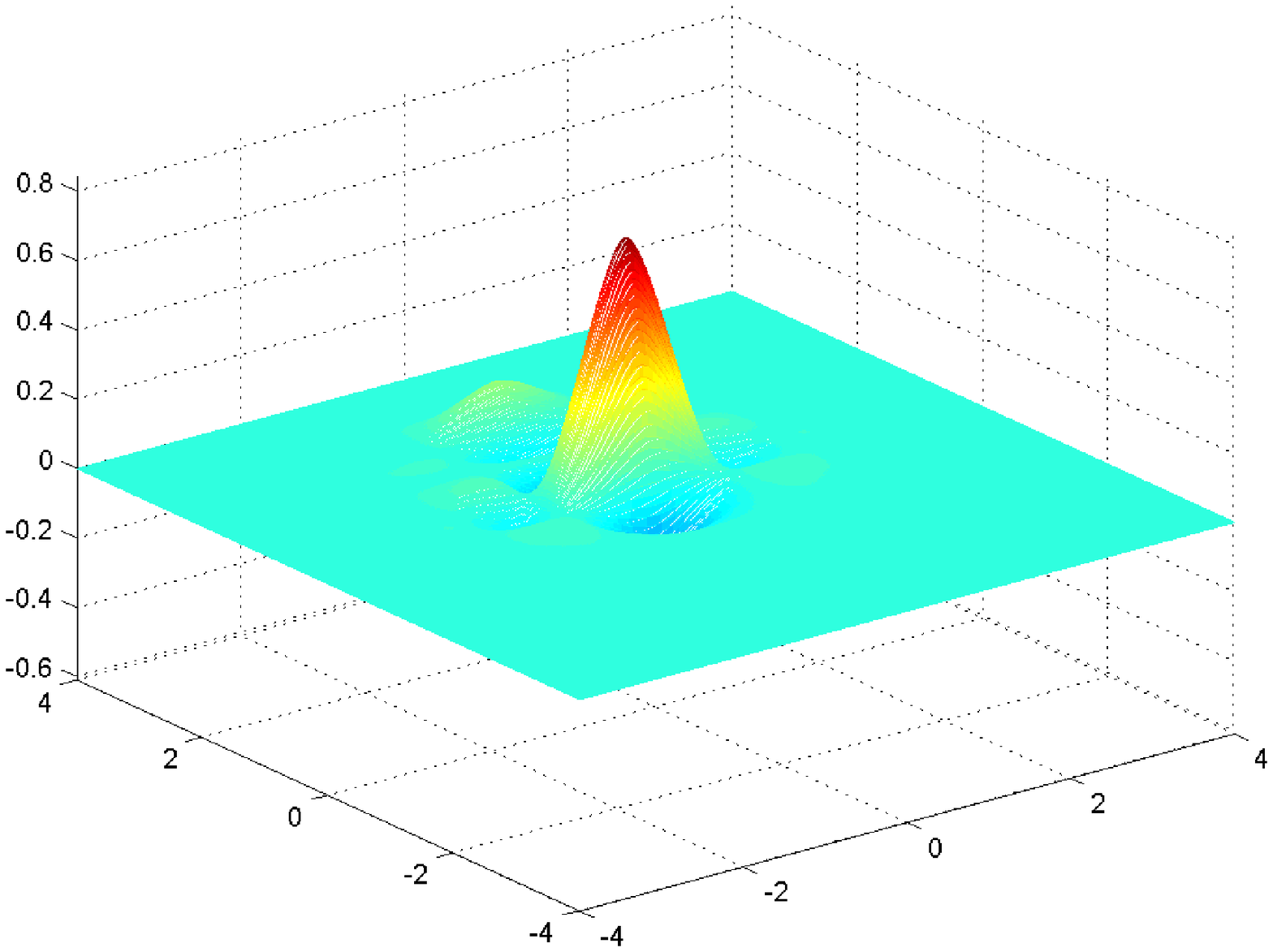}}
\subfigure{
\includegraphics[width=1.5in,height=1.0in]
{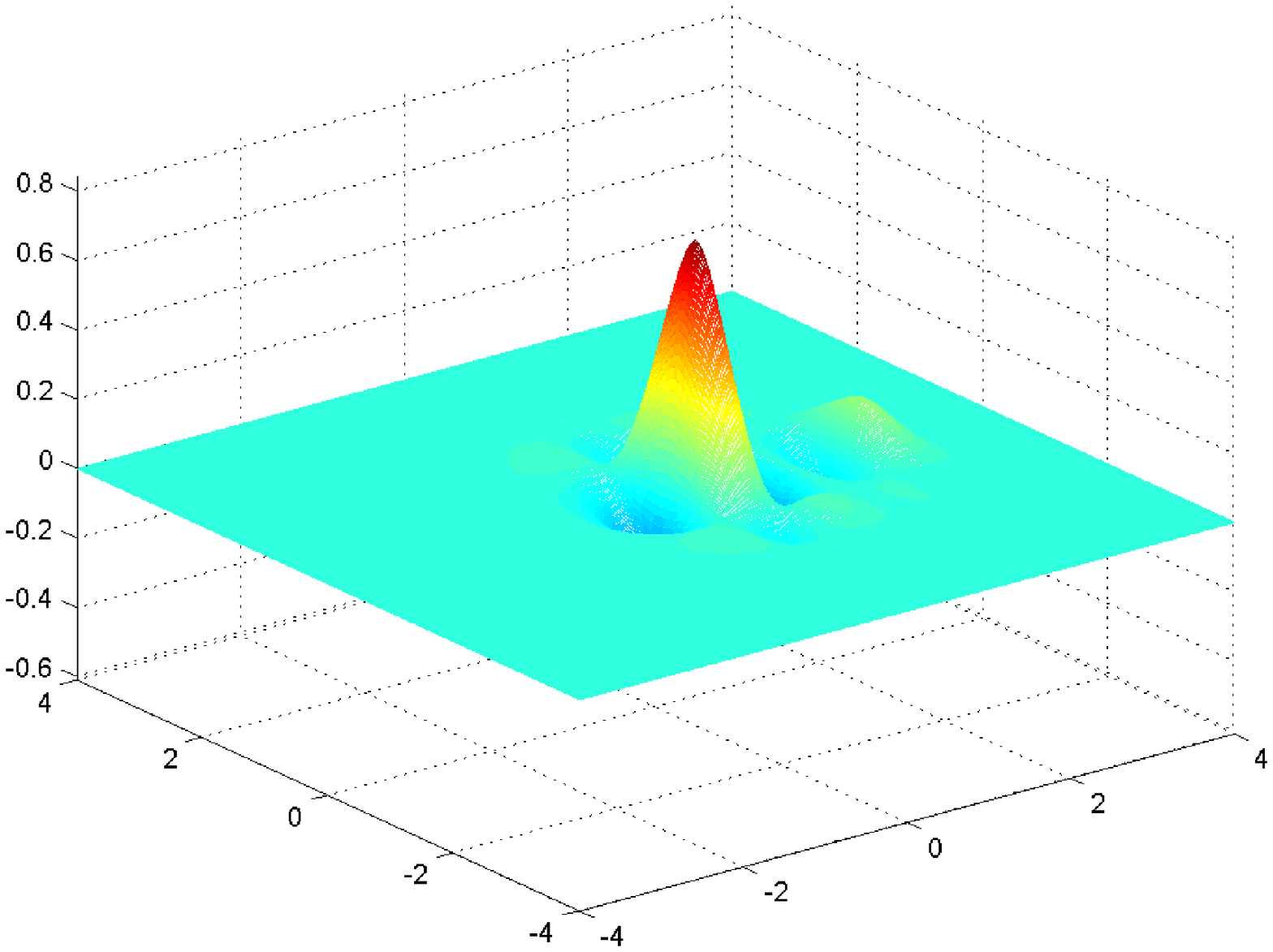}}
\subfigure{
\includegraphics[width=1.5in,height=1.0in]
{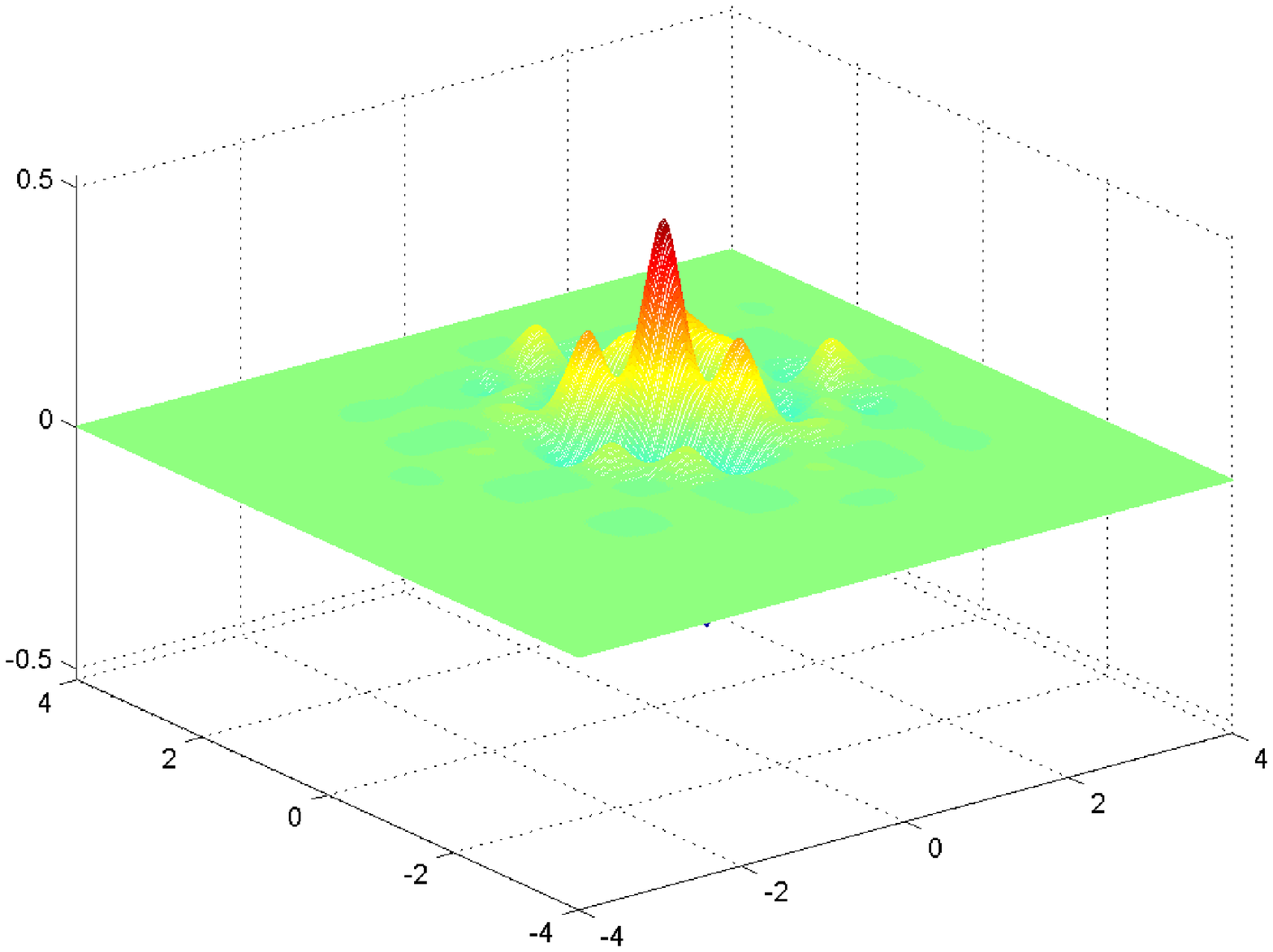}}
\subfigure{
\includegraphics[width=1.5in,height=1.0in]
{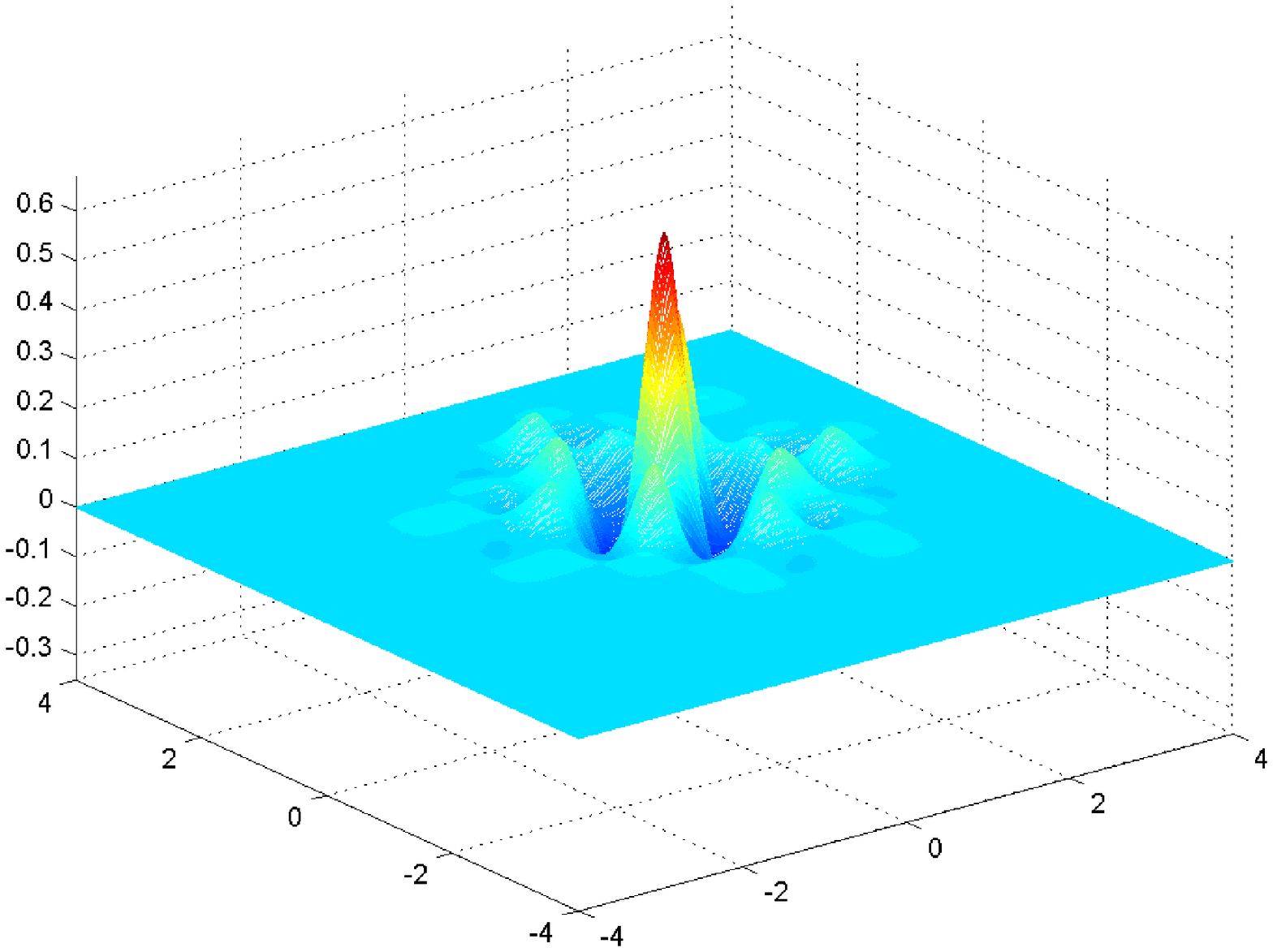}}\\
\subfigure{
\includegraphics[width=1.5in,height=1.0in]
{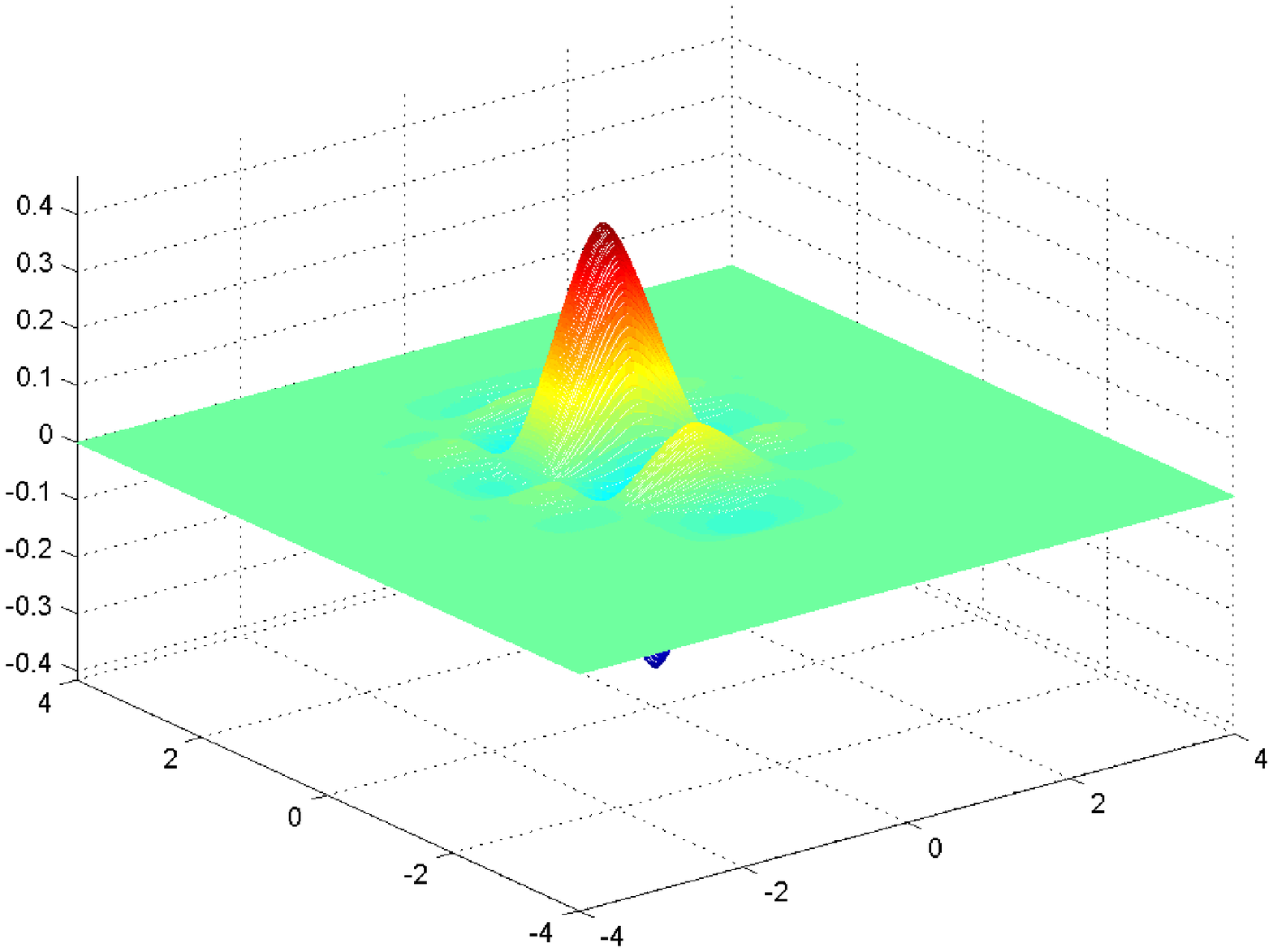}}
\subfigure{
\includegraphics[width=1.5in,height=1.0in]
{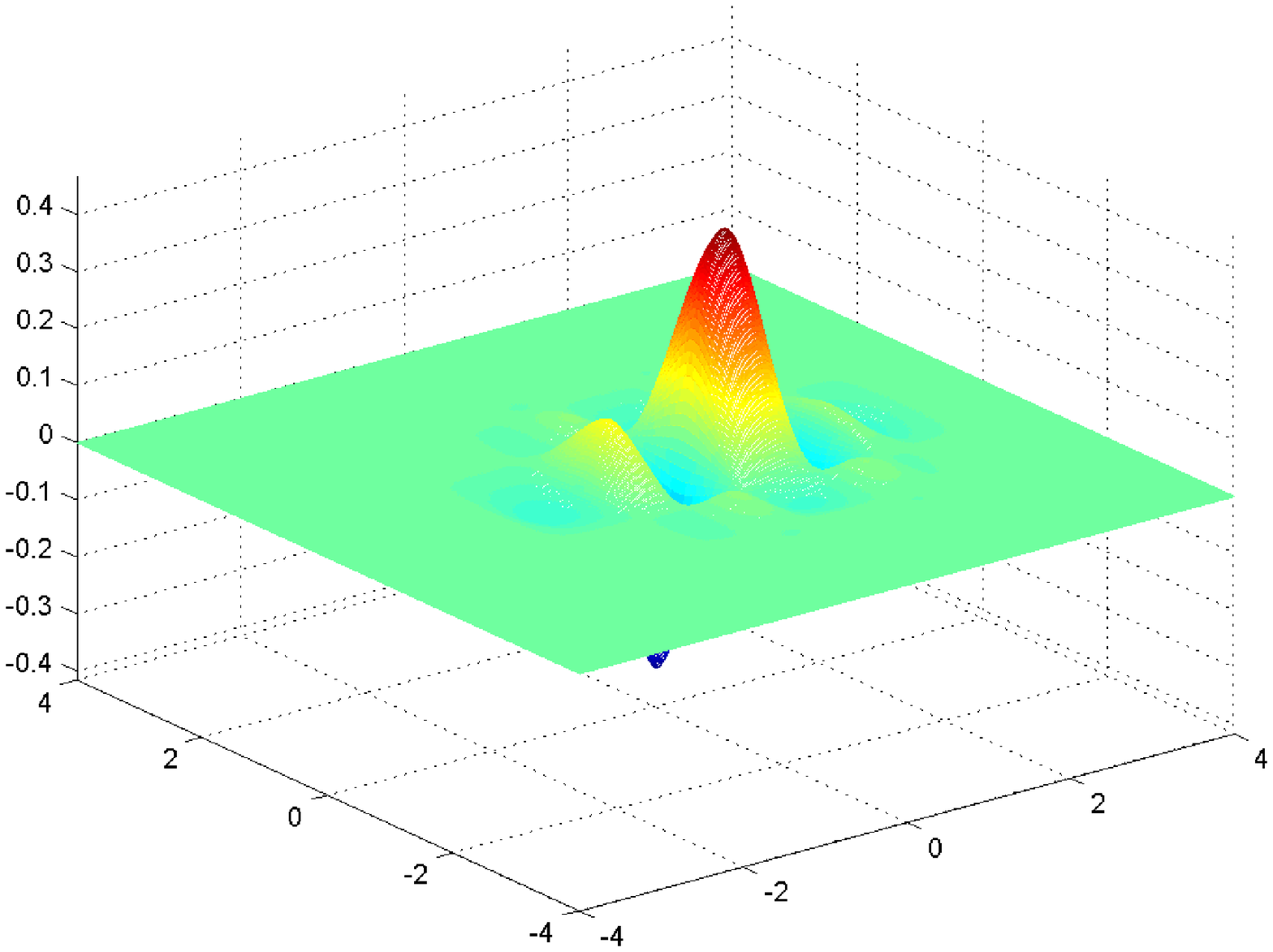}}
\subfigure{
\includegraphics[width=1.5in,height=1.0in]
{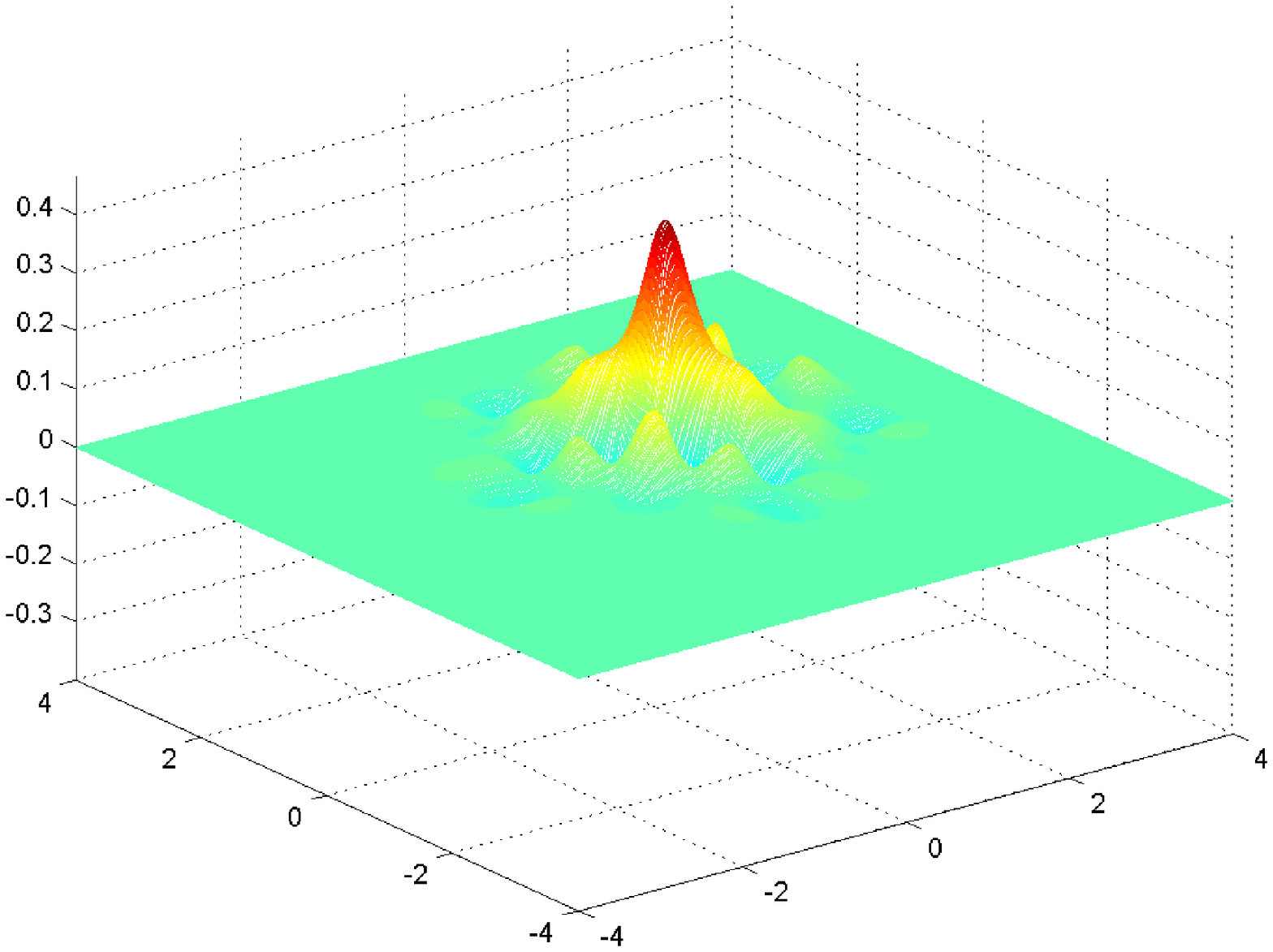}}
\subfigure{
\includegraphics[width=1.5in,height=1.0in]
{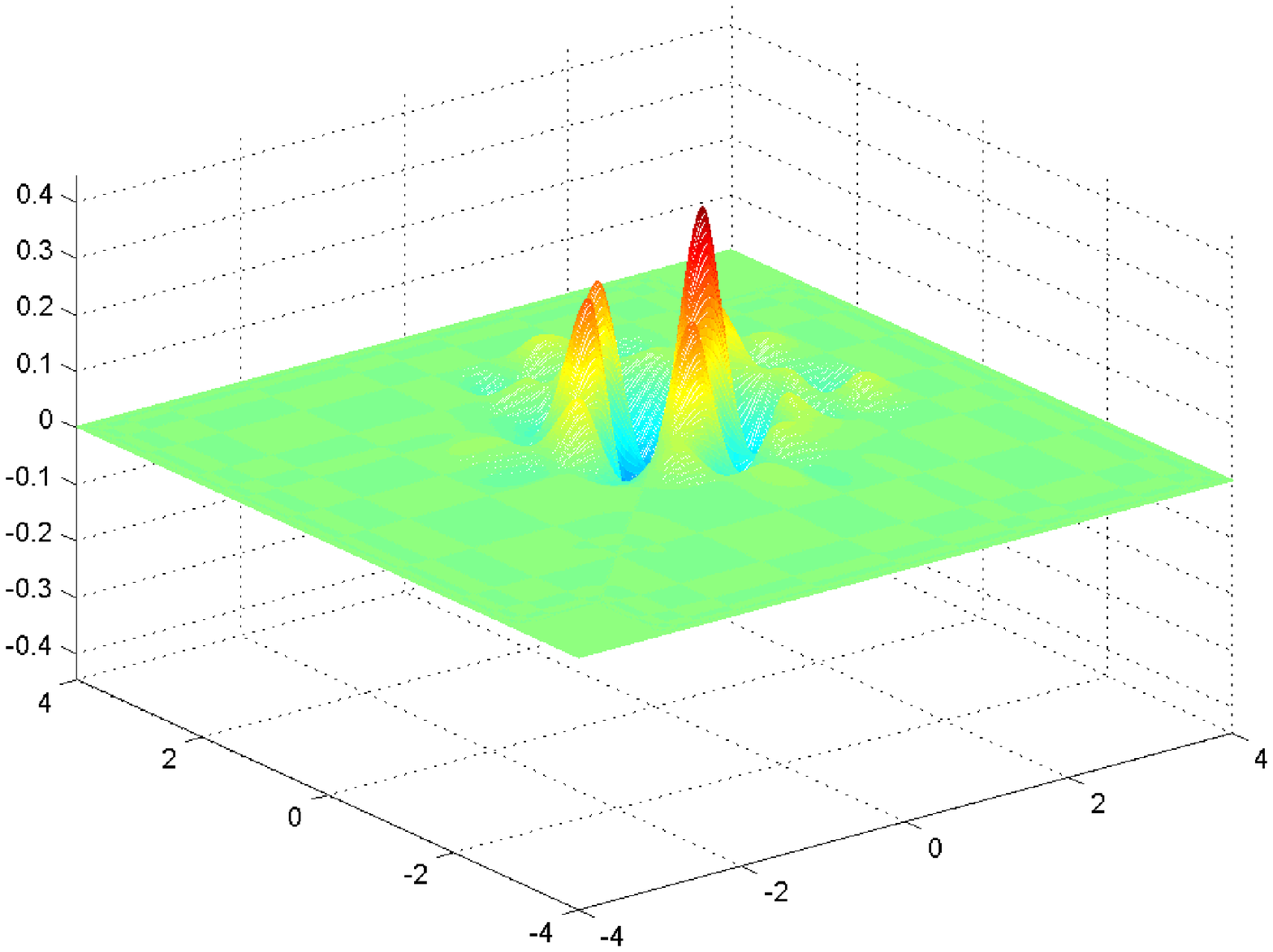}}\\
\subfigure{
\includegraphics[width=0.7in,height=0.7in]
{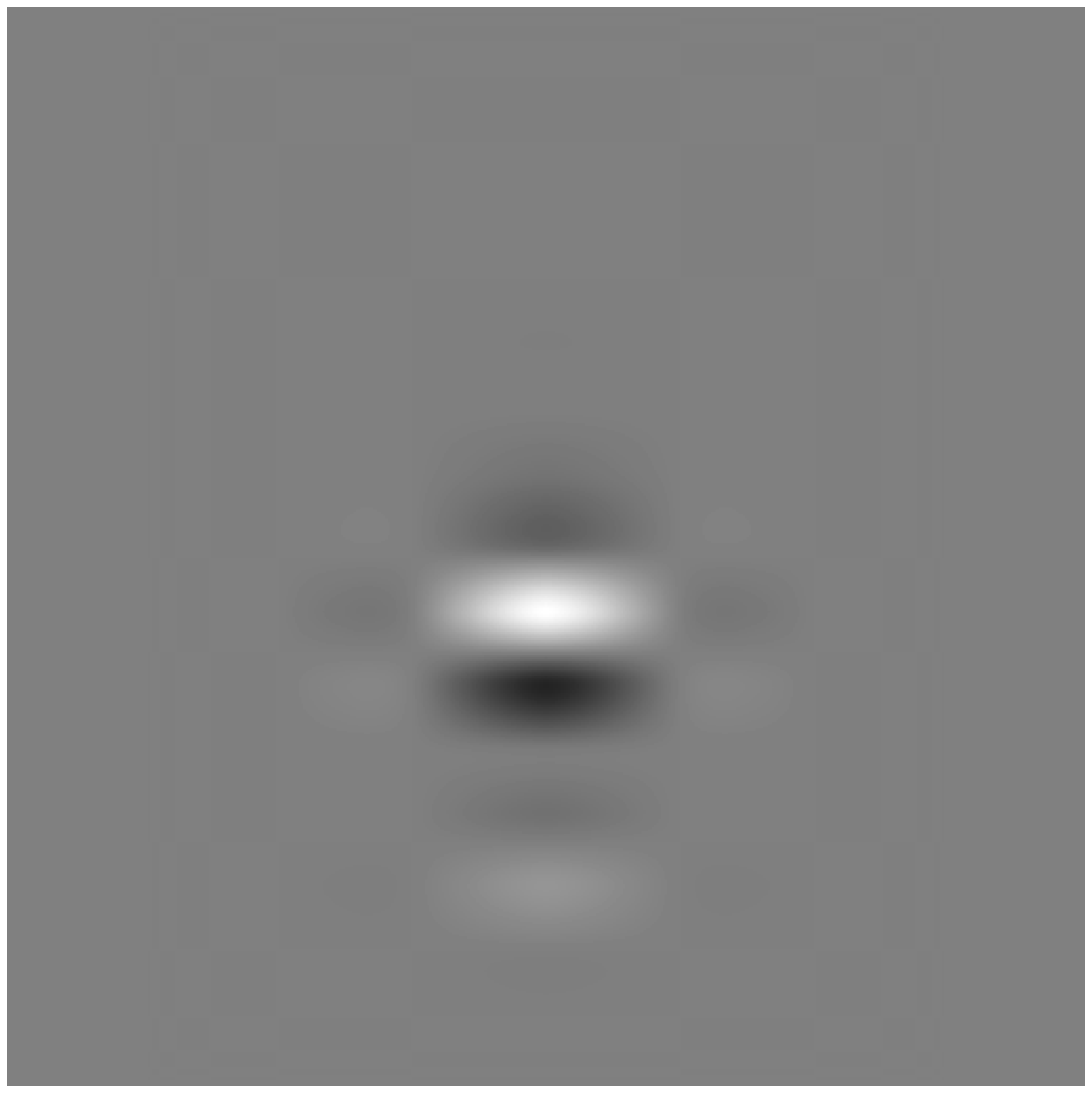}}
\subfigure{
\includegraphics[width=0.7in,height=0.7in]
{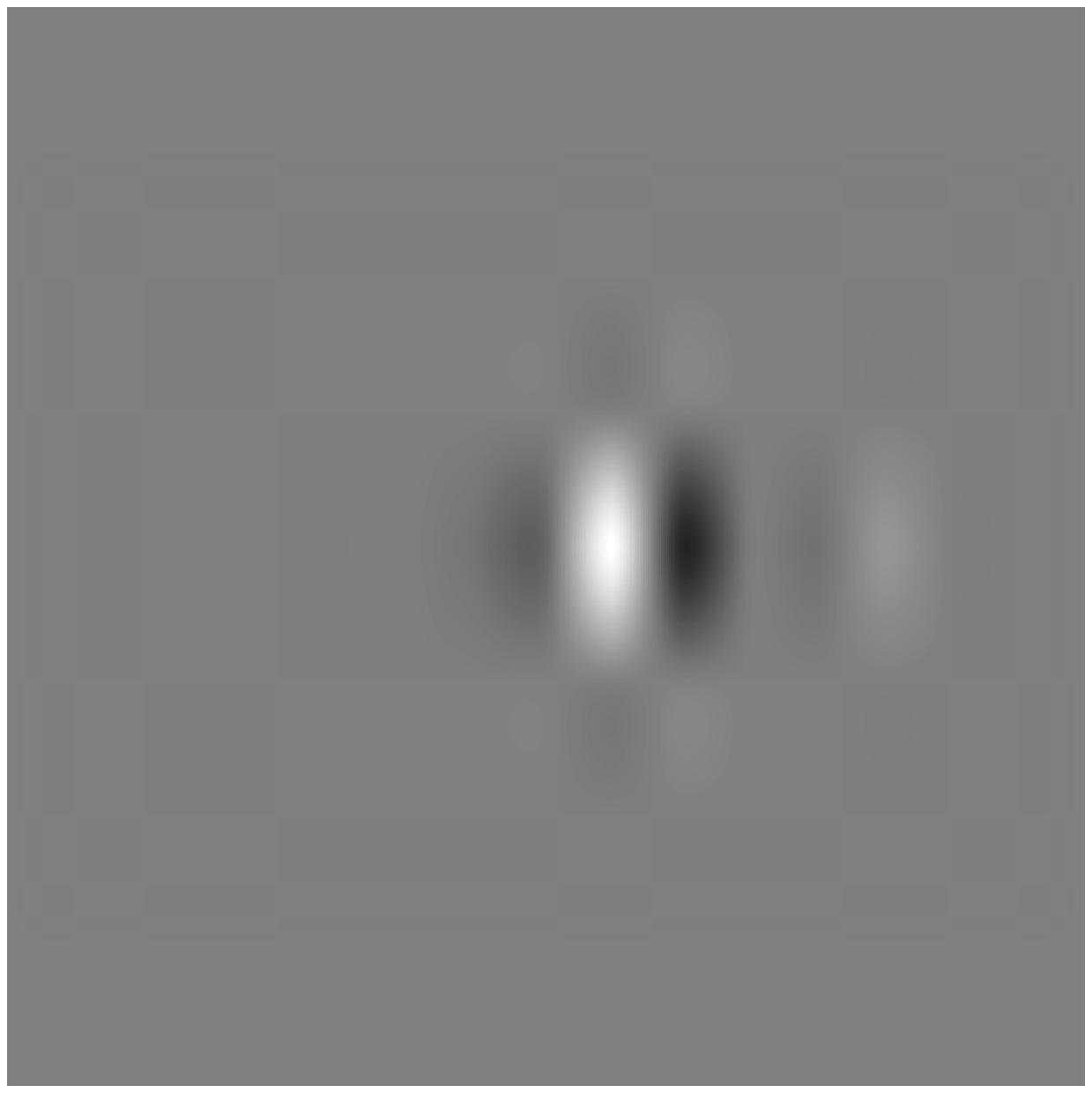}}
\subfigure{
\includegraphics[width=0.7in,height=0.7in]
{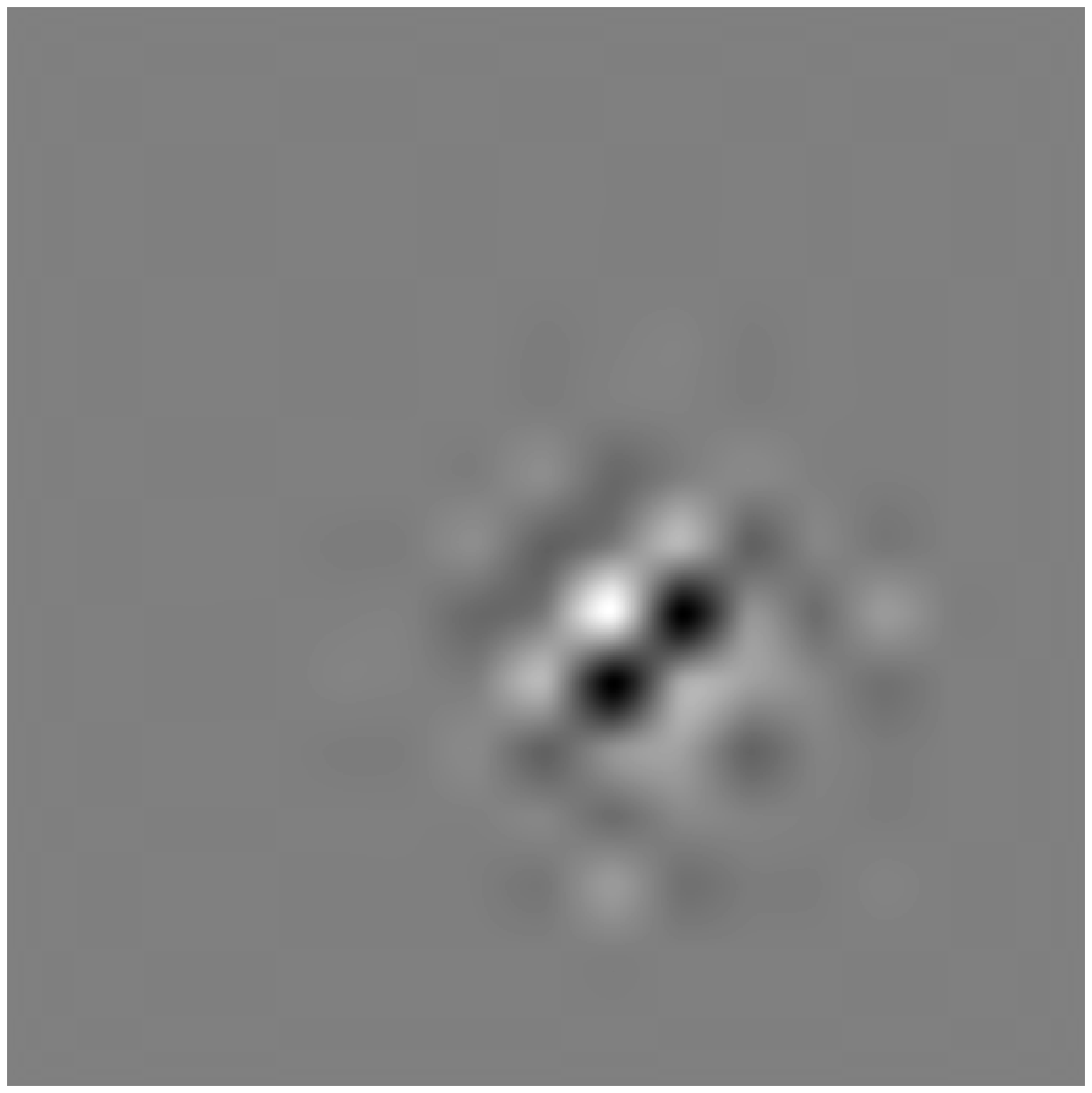}}
\subfigure{
\includegraphics[width=0.7in,height=0.7in]
{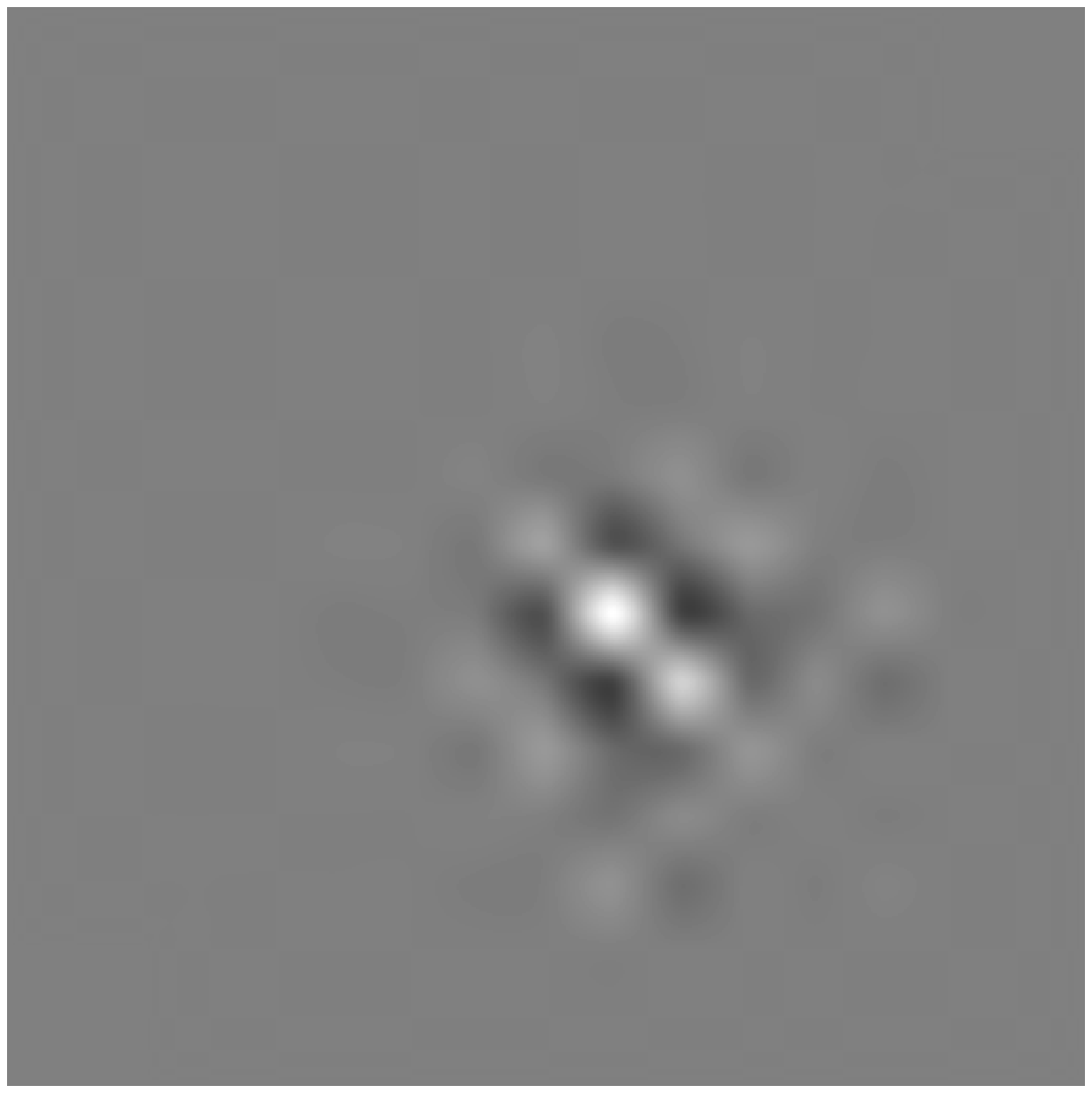}}
\subfigure{
\includegraphics[width=0.7in,height=0.7in]
{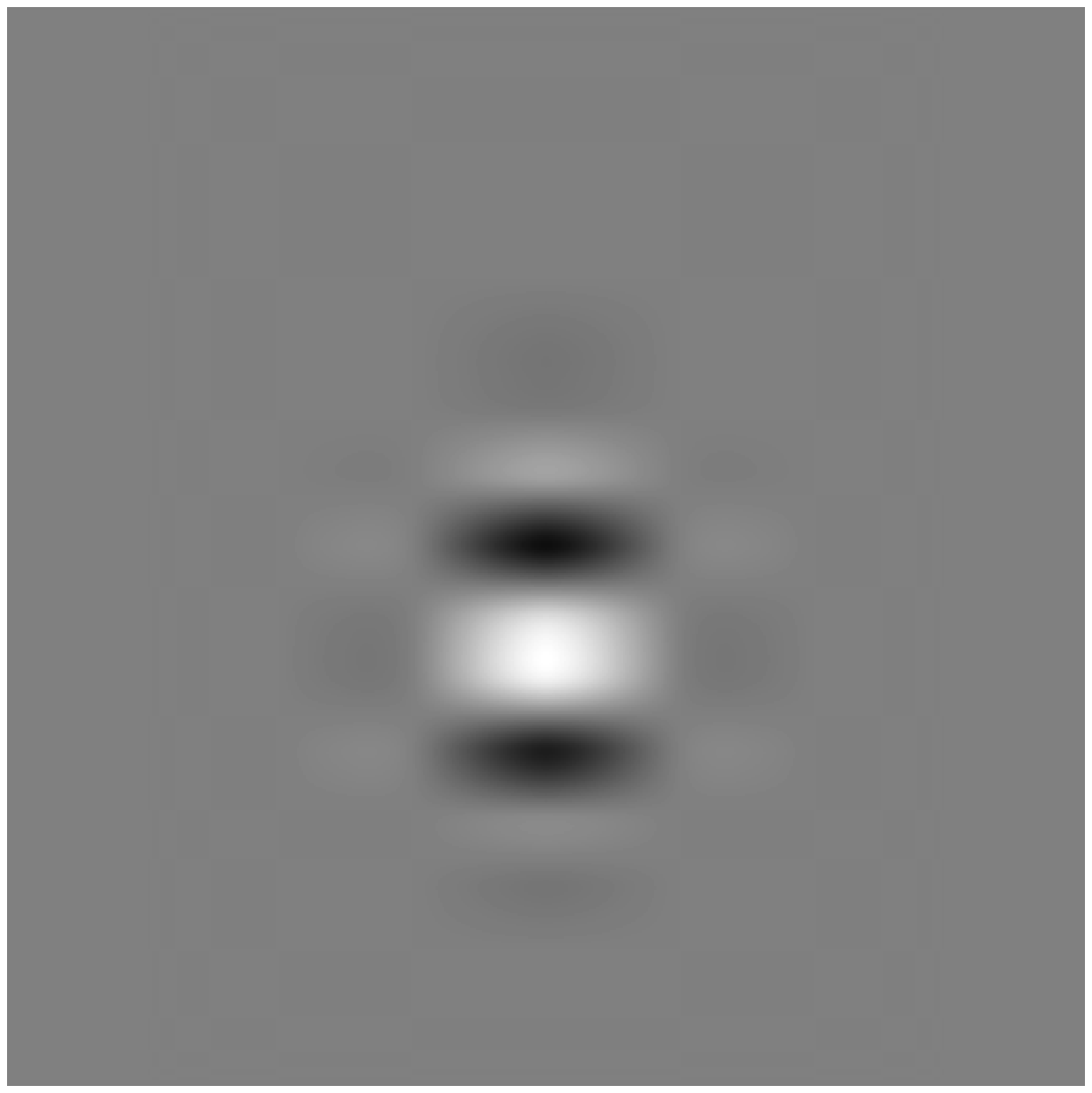}}
\subfigure{
\includegraphics[width=0.7in,height=0.7in]
{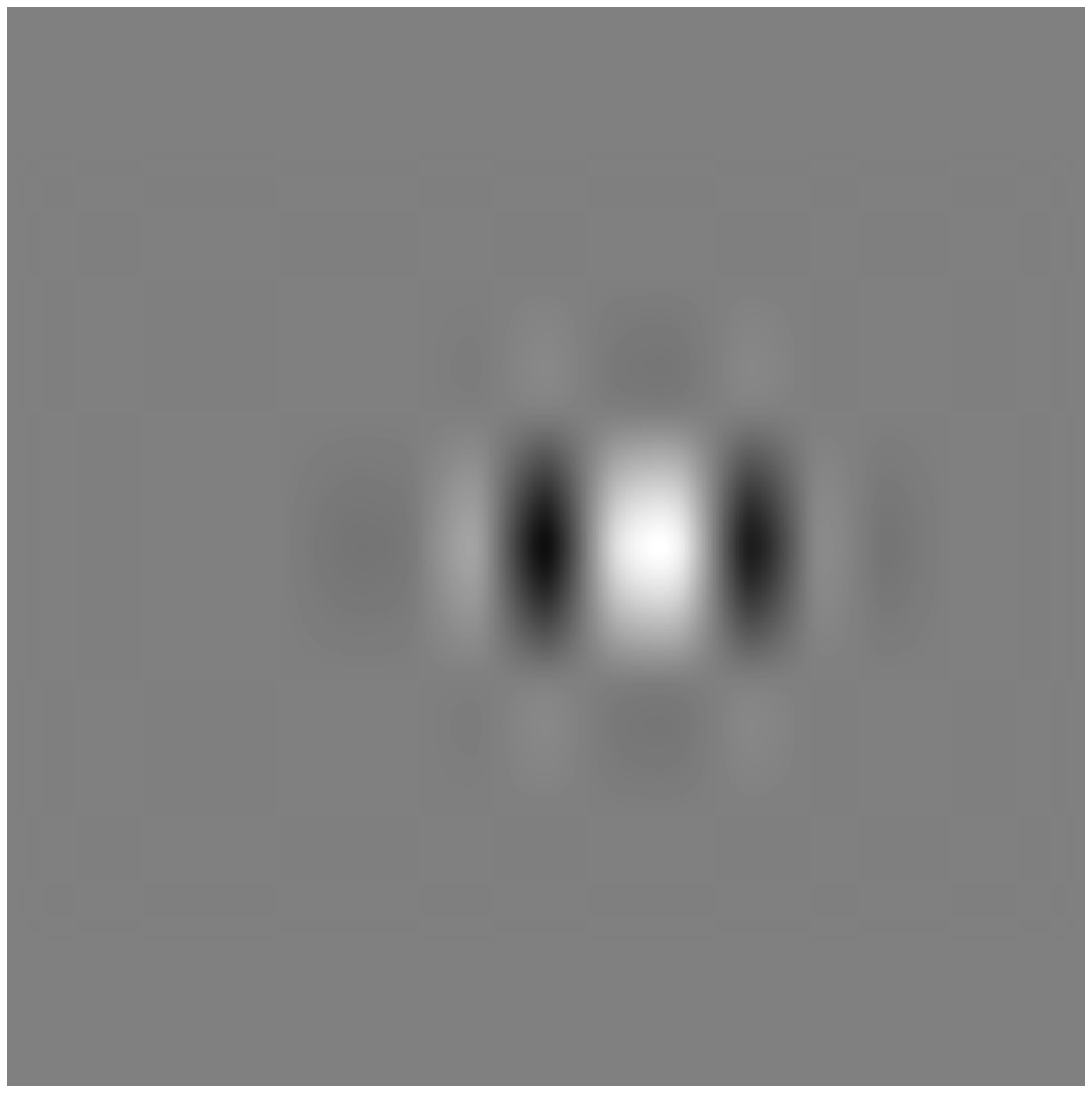}}
\subfigure{
\includegraphics[width=0.7in,height=0.7in]
{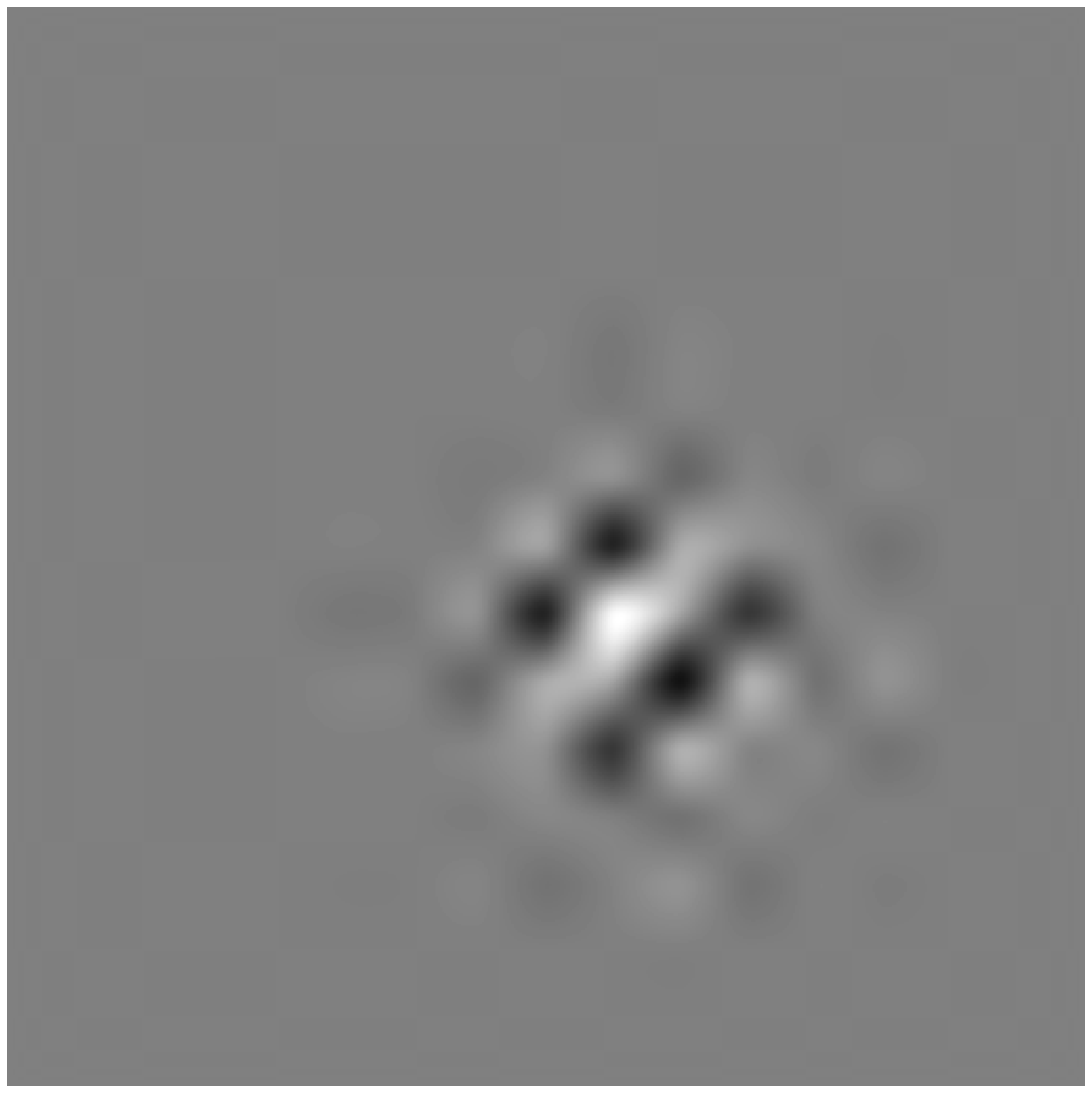}}
\subfigure{
\includegraphics[width=0.7in,height=0.7in]
{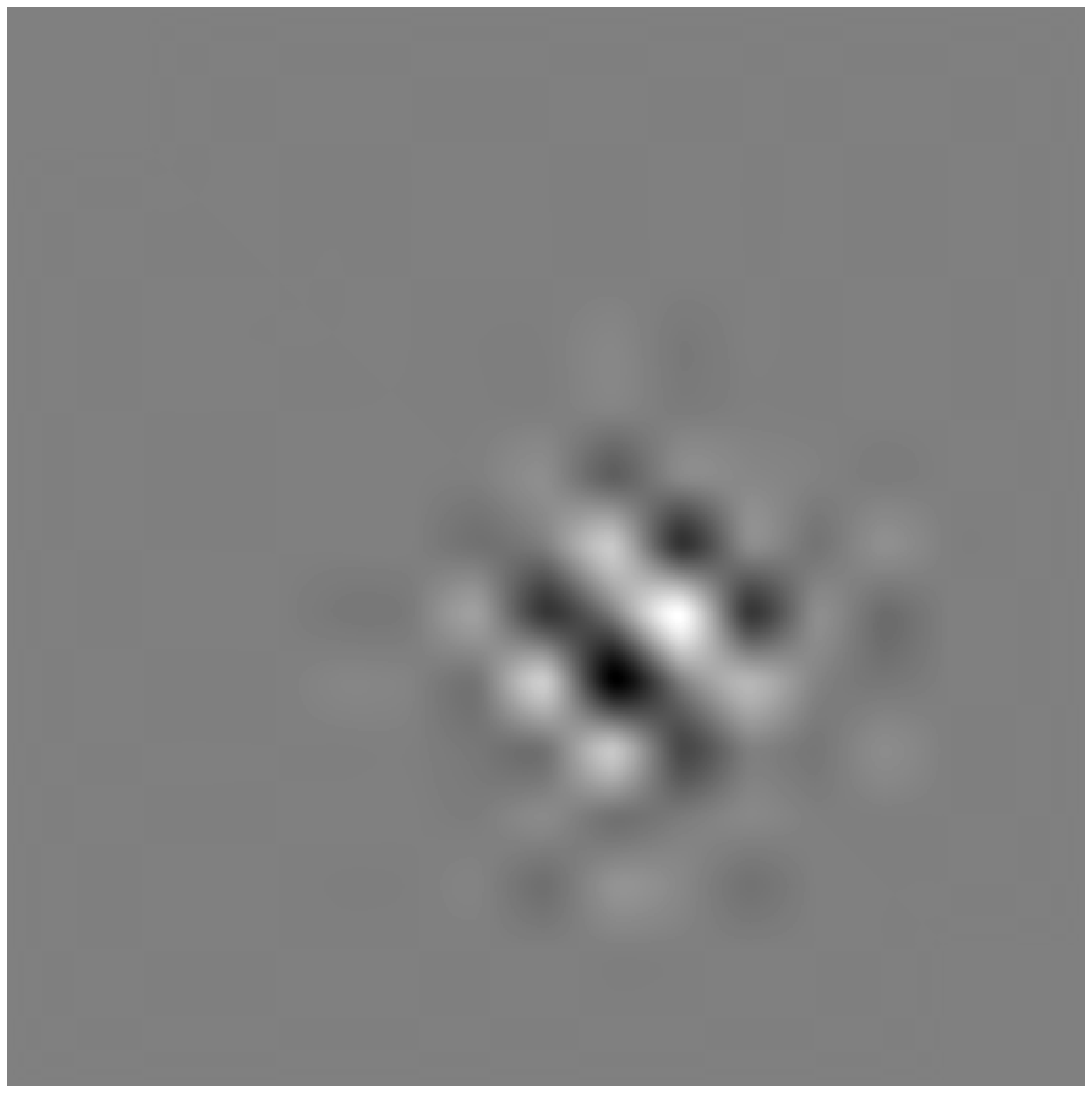}}
\begin{caption}{
The first row is for the real part and the second row is for the imaginary part of the tight framelet generators in Example~\ref{ex3} with $N=2$.
The third row is the greyscale image of the eight generators:
the first four for real part and the last four for imaginary part.
} \label{fig3}
\end{caption}
\end{center}
\end{figure}

\begin{example}\label{ex4} {\rm Let $\pa(z) =-\frac{3}{64}z^{-2}+\frac{5}{64}z^{-1}+\frac{15}{32}
+\frac{15}{32}z+\frac{5}{64}z^2-\frac{3}{64}z^3
=\{-\tfrac{3}{64}, \tfrac{5}{64}, \tfrac{15}{32}, \frac{15}{32}, \tfrac{5}{64}, -\tfrac{3}{64} \}_{[-2,3]}$.
Using Algorithm~\ref{alg:tffb}, we obtain a tight framelet filter bank $\{a; b_1, b_2\}$ with
\begin{align*}
\pb_1(z) =\frac{\sqrt{297879}}{6354752}z^{-2}(z-1)^2(3203z^3+1921z^2-31z-93),\quad
\pb_2(z)=-\frac{\sqrt{496465}}{794344}z^{-2}(z-1)^2(248z^2+z+3).
\end{align*}
Applying Algorithm~\ref{alg:main} with $N =0$, we have a finitely supported complex tight framelet filter bank $\{a; b^p, b^n\}$ with
$b^n=\ol{b^p}$ and
\begin{align*}
\pb^p(z) =&(-0.00427685553137+0.00414104756179i)z^{-2}
+(0.00712809255229-0.00690174593633i)z^{-1}\\
&-(0.0855371106277+0.173923997595i)+(0.256611331884+0.179445394344i)z\\
&-(0.263739424437-0.169782950034i)z^2+(0.0898139661592-0.172543648408i)z^3.
\end{align*}
By calculation we have $d_{\R}=\frac{557}{1024}\pi \approx 1.70885$, $d_A\approx 0.12595$, and $d_B\approx 0.444929$. 
If we take $N=2$, then
\begin{align*}
\pb^p(z)=&(0.000174962462944+0.000667428960698i)z^{-4}-(0.000291604104907+0.00111238160116i)z^{-3}\\
&(0.00604271655936+0.00470763073225i)z^{-2}-(0.0147368599441+0.0256441568388i)z^{-1}\\
&(0.119900001837+0.197463905830i)-(0.282016222613+0.153449185519i)z\\
&(0.207557346012-0.197627972773i)z^2+(-0.0335526030324+0.174187921034i)z^3\\
&(0.0198783637212-0.00521099275091i)z^4+(-0.0229561008971+0.00601780292596i)z^5.
\end{align*}
By calculation, we have $d_B\approx 0.387149$. 
See Figure~\ref{fig4} for the graphs of the eight tight framelet generators in the associated two-dimensional real-valued tight framelet for $L_2(\R^2)$ in \eqref{tp:ctf}.
} \end{example}

\begin{figure}[ht]
\begin{center}
\subfigure{
\includegraphics[width=1.5in,height=1.0in]
{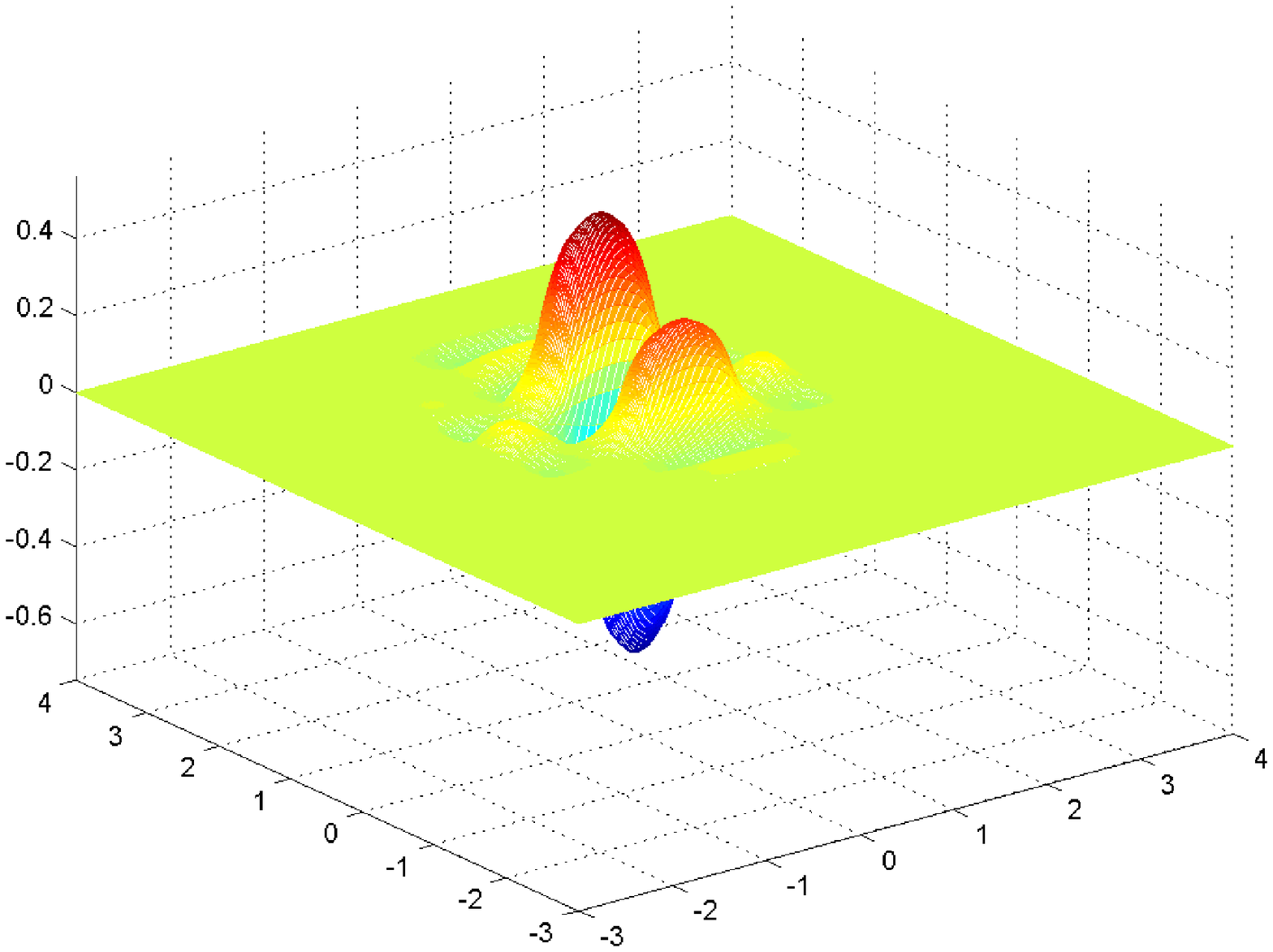}}
\subfigure{
\includegraphics[width=1.5in,height=1.0in]
{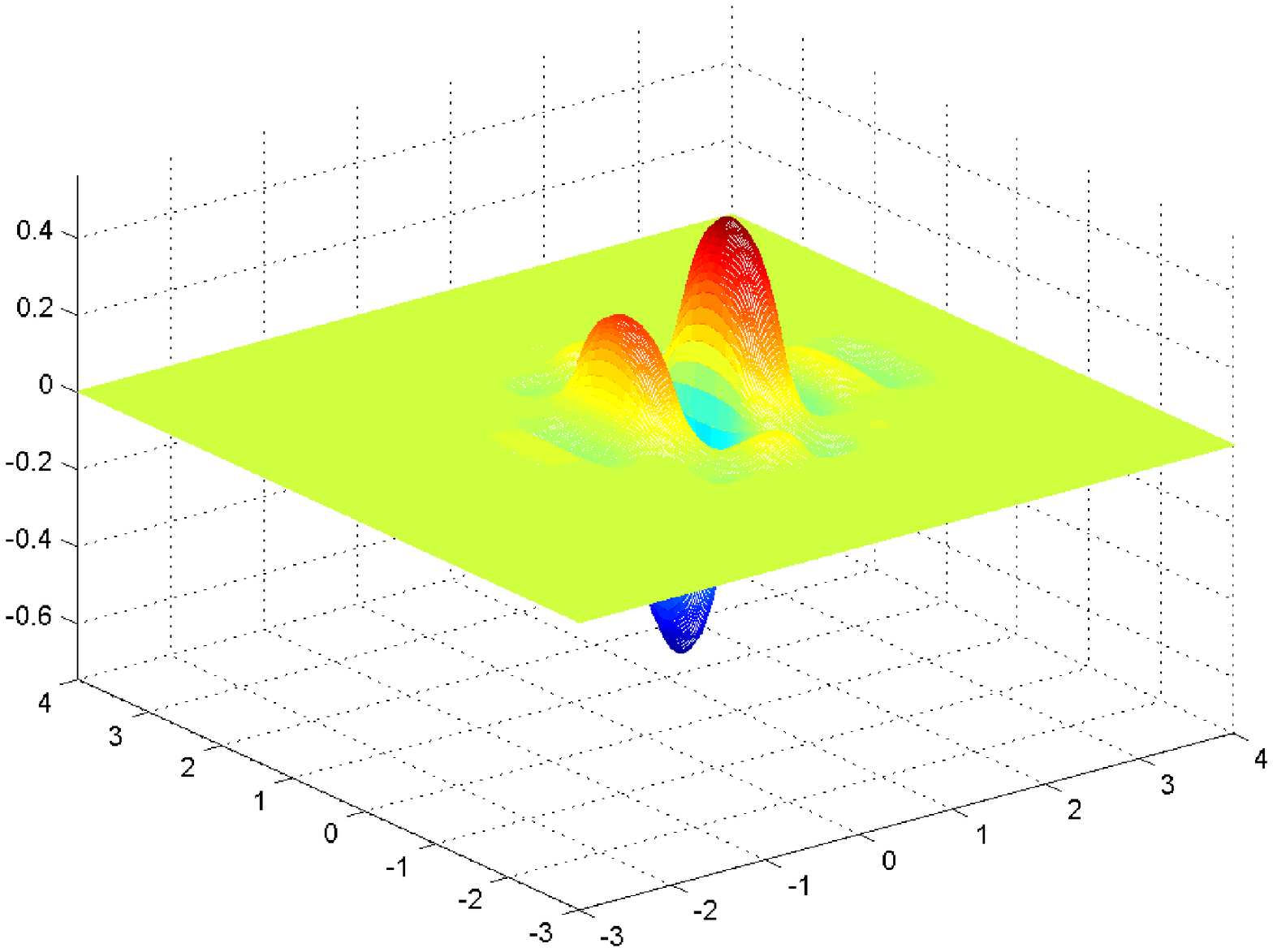}}
\subfigure{
\includegraphics[width=1.5in,height=1.0in]
{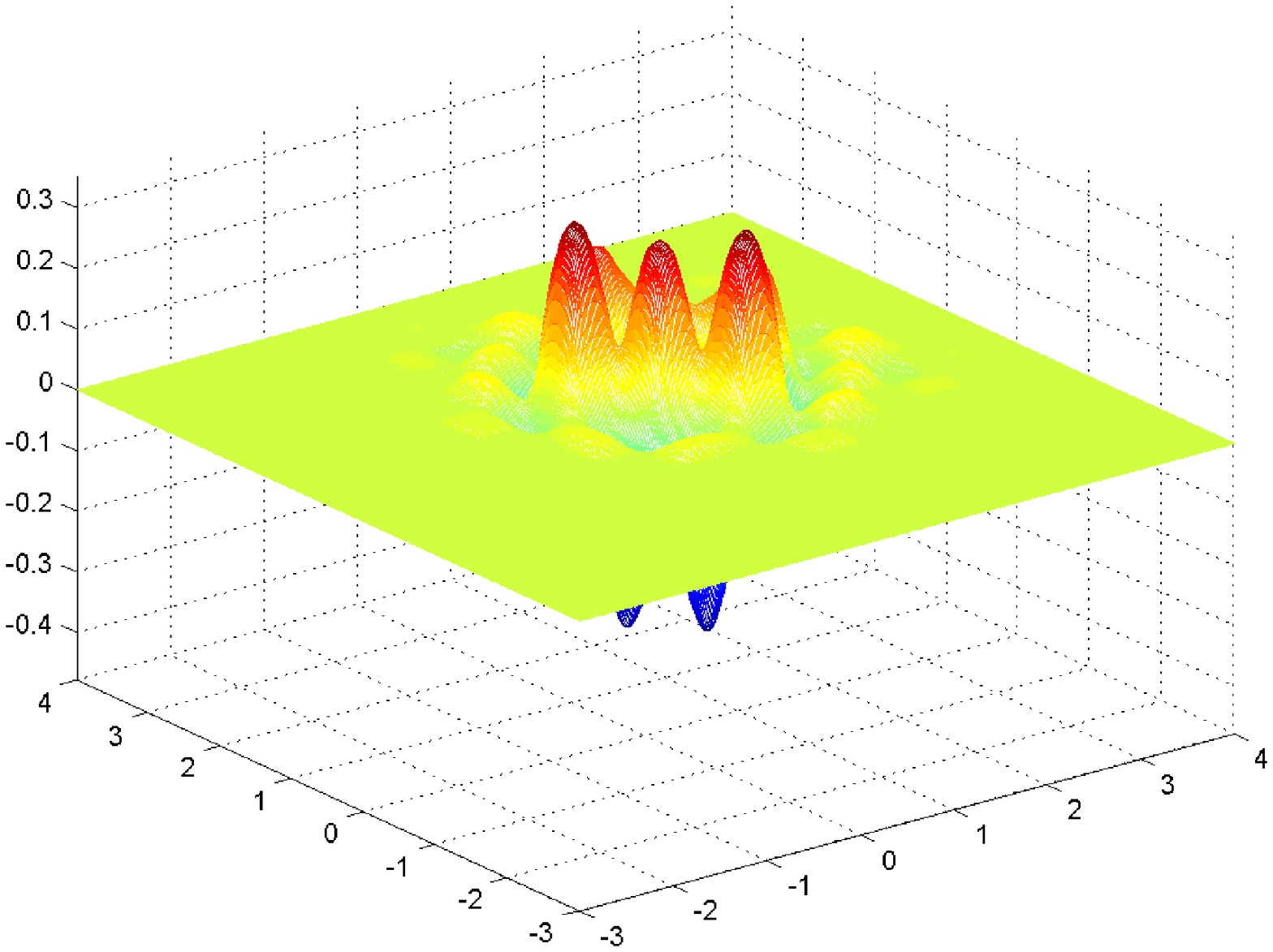}}
\subfigure{
\includegraphics[width=1.5in,height=1.0in]
{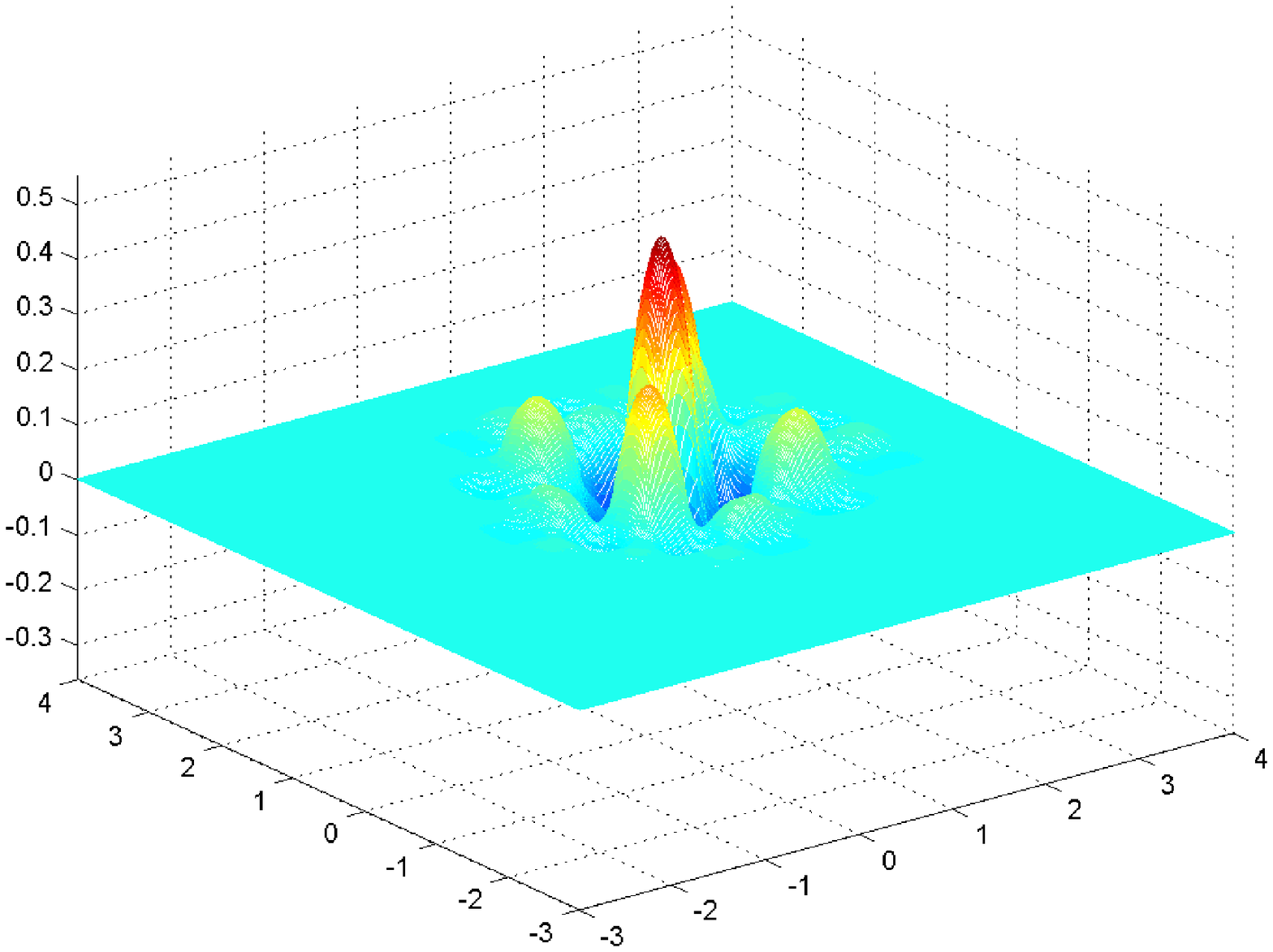}}\\
\subfigure{
\includegraphics[width=1.5in,height=1.0in]
{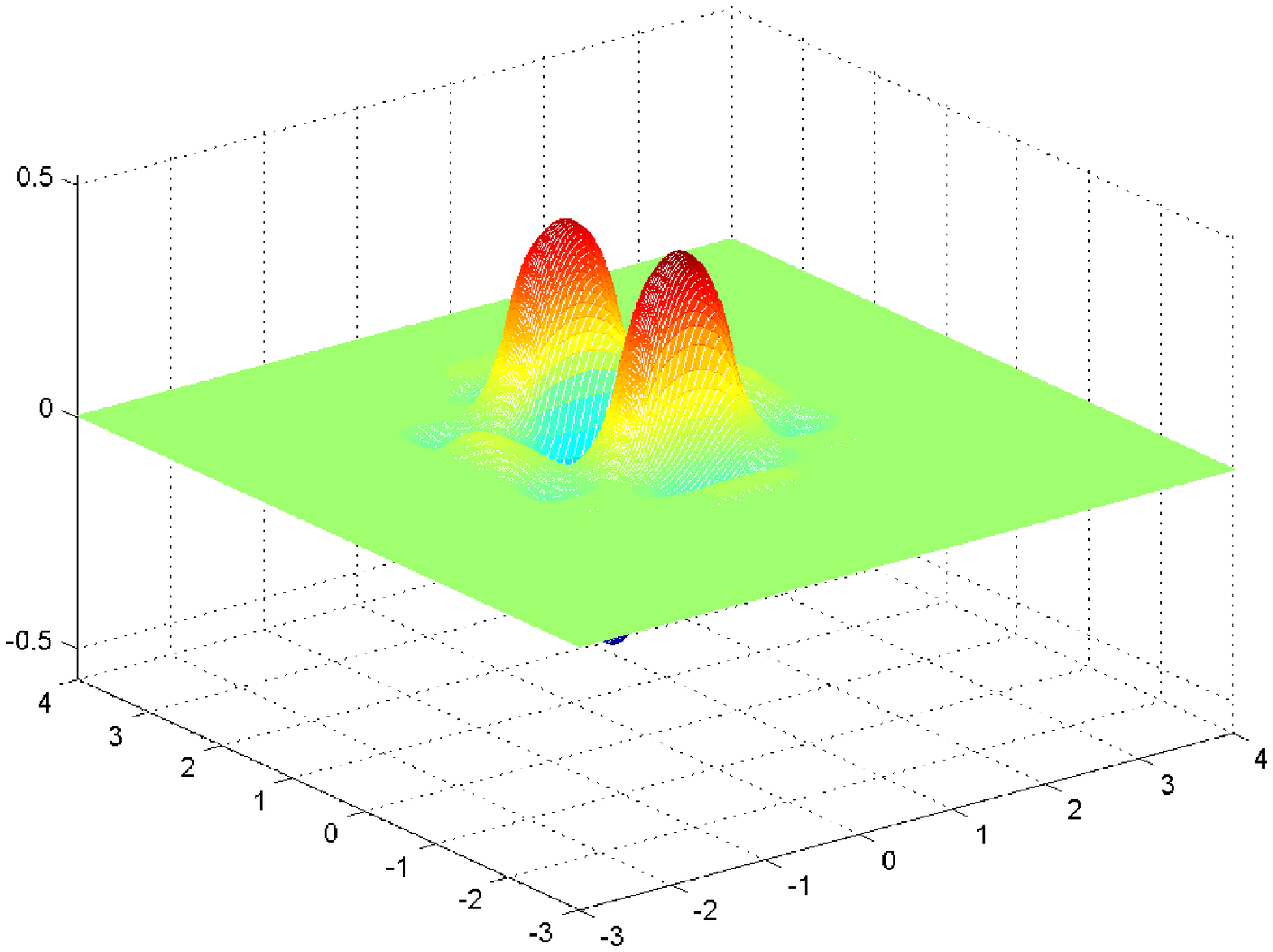}}
\subfigure{
\includegraphics[width=1.5in,height=1.0in]
{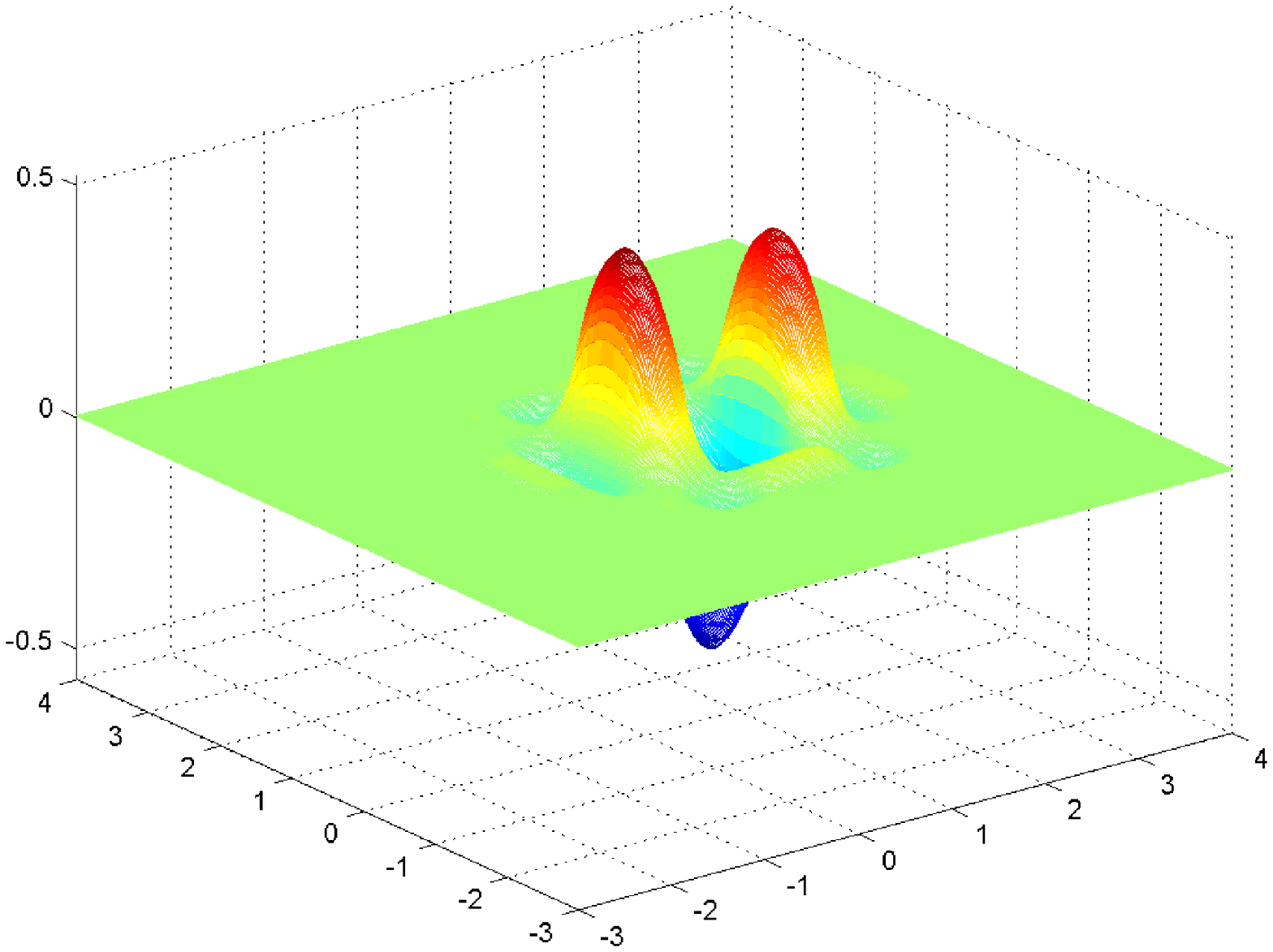}}
\subfigure{
\includegraphics[width=1.5in,height=1.0in]
{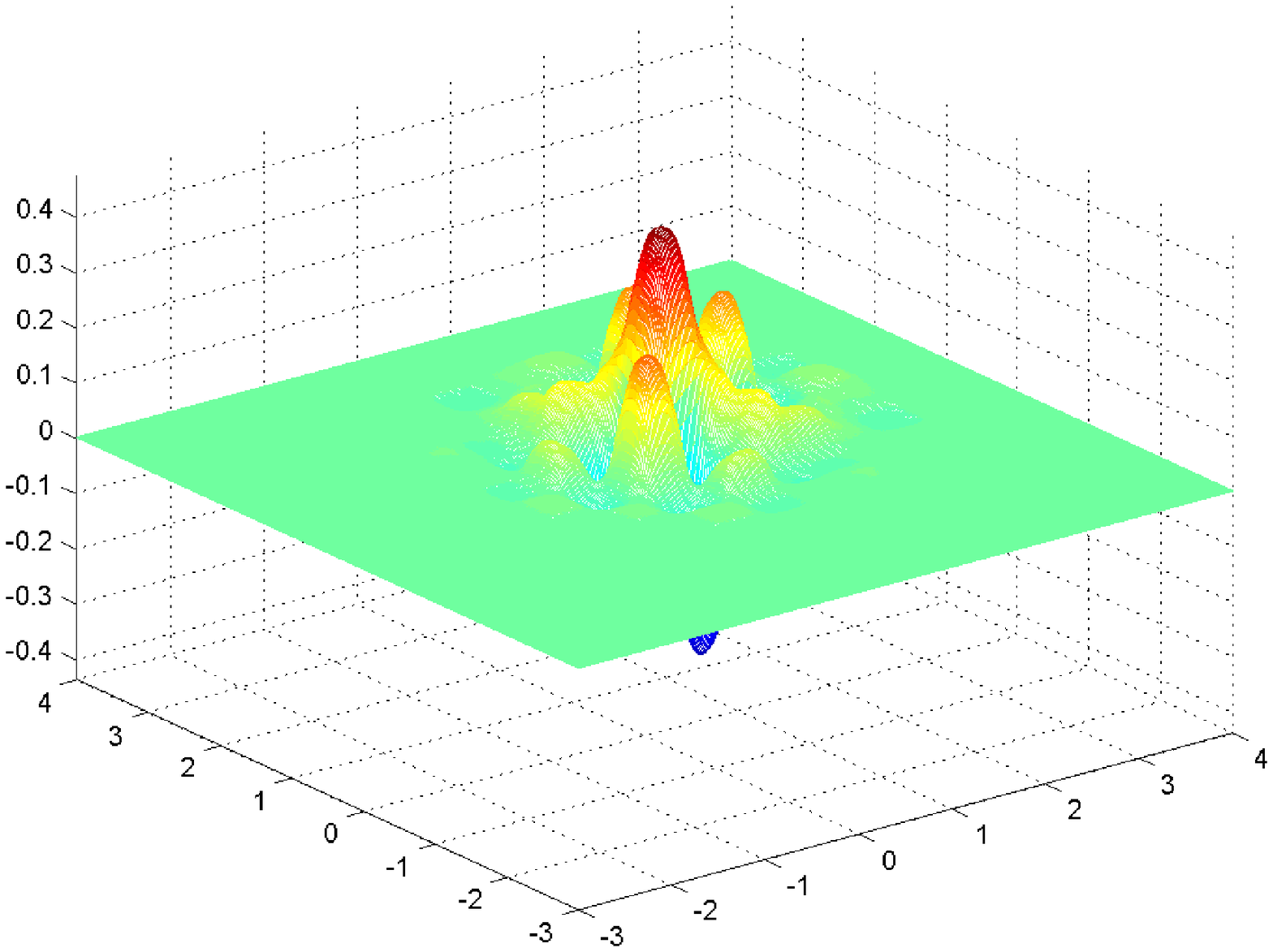}}
\subfigure{
\includegraphics[width=1.5in,height=1.0in]
{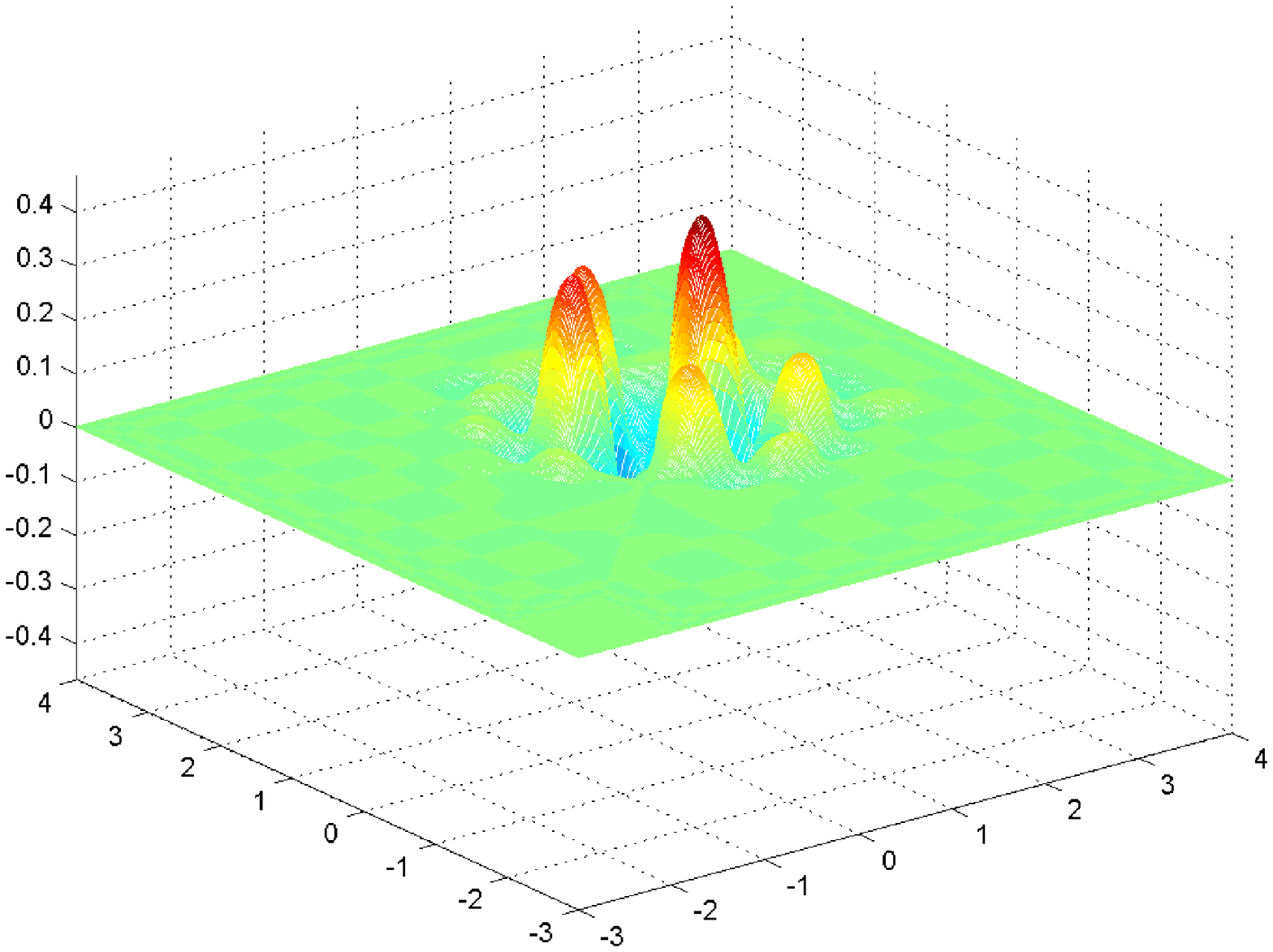}}\\
\subfigure{
\includegraphics[width=0.7in,height=0.7in]
{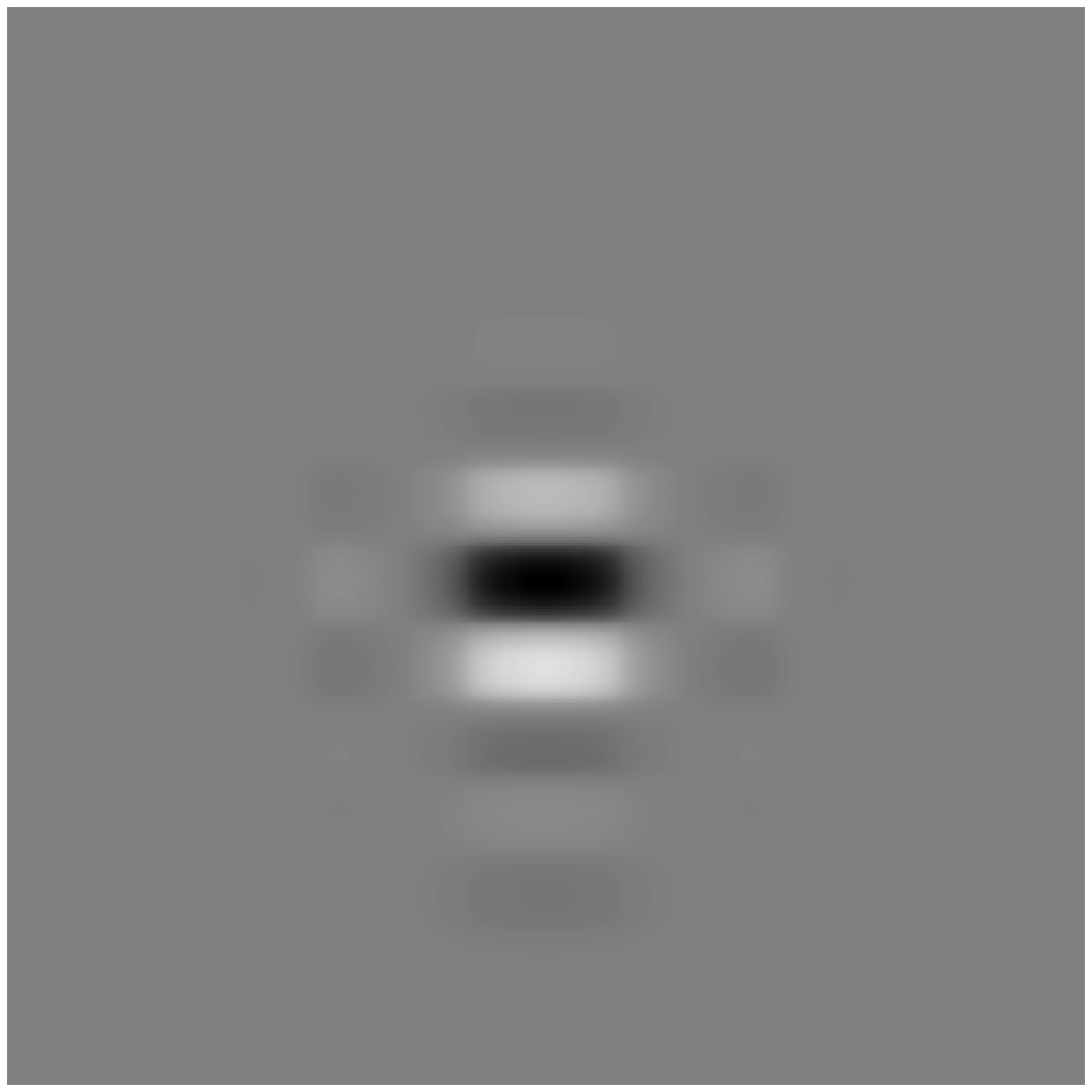}}
\subfigure{
\includegraphics[width=0.7in,height=0.7in]
{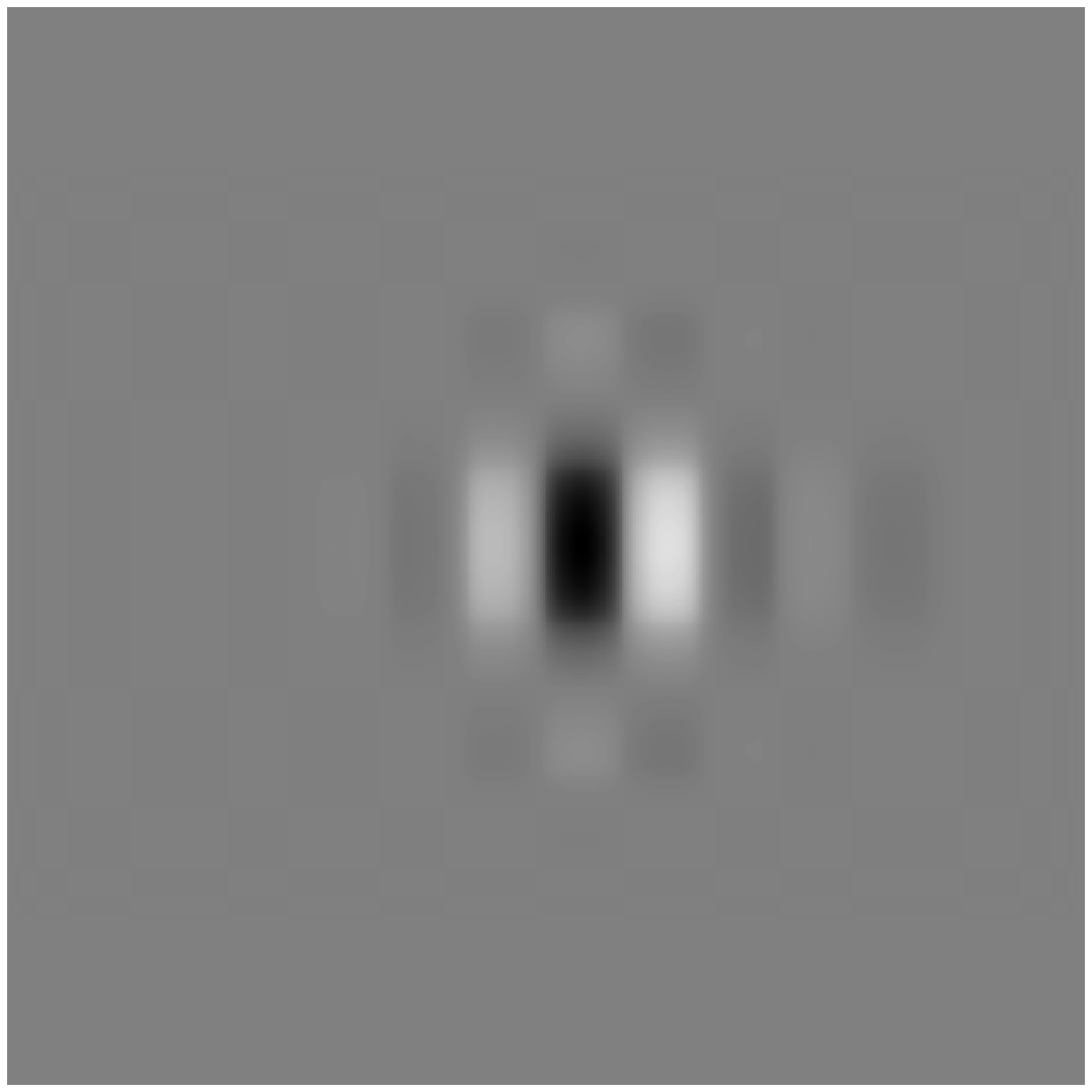}}
\subfigure{
\includegraphics[width=0.7in,height=0.7in]
{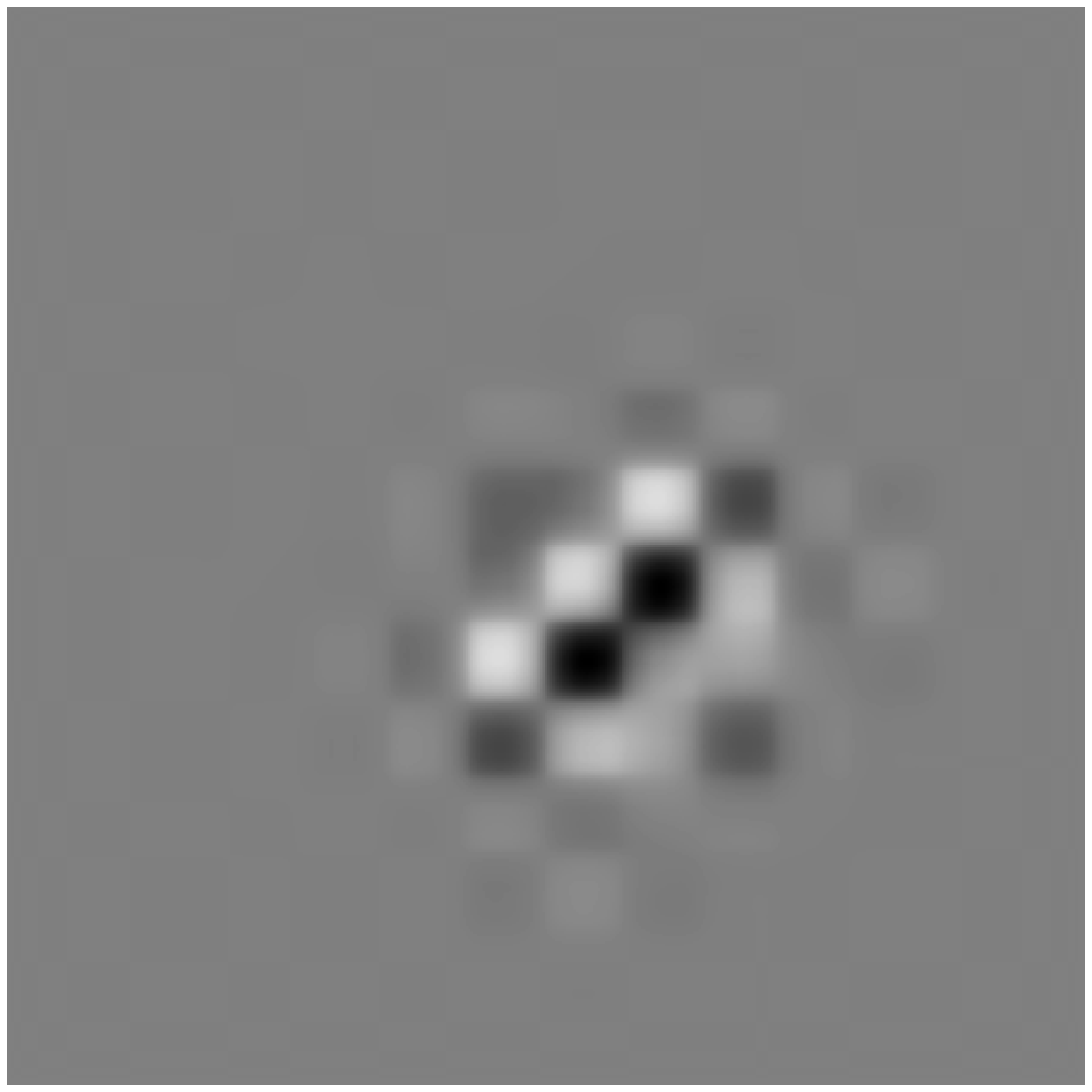}}
\subfigure{
\includegraphics[width=0.7in,height=0.7in]
{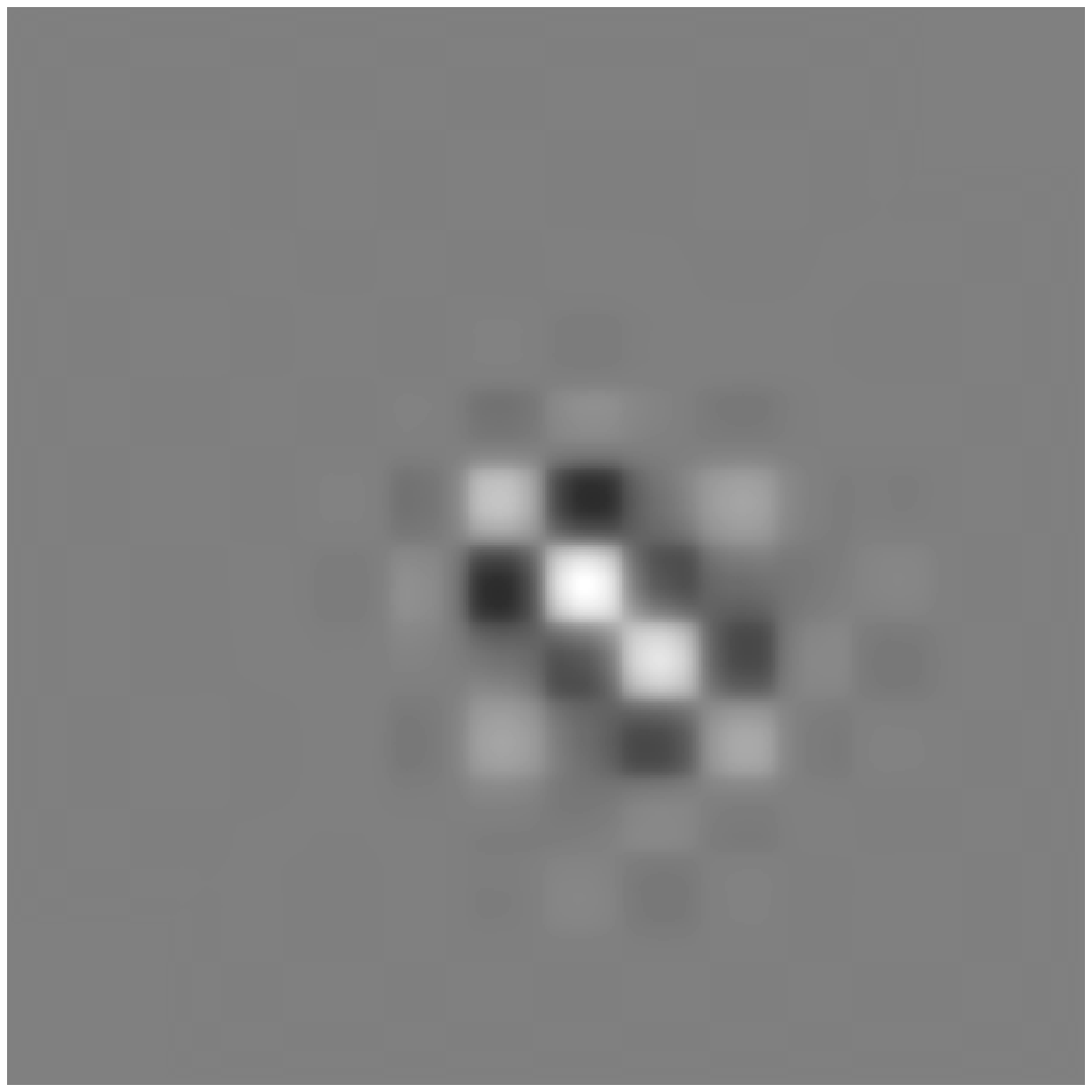}}
\subfigure{
\includegraphics[width=0.7in,height=0.7in]
{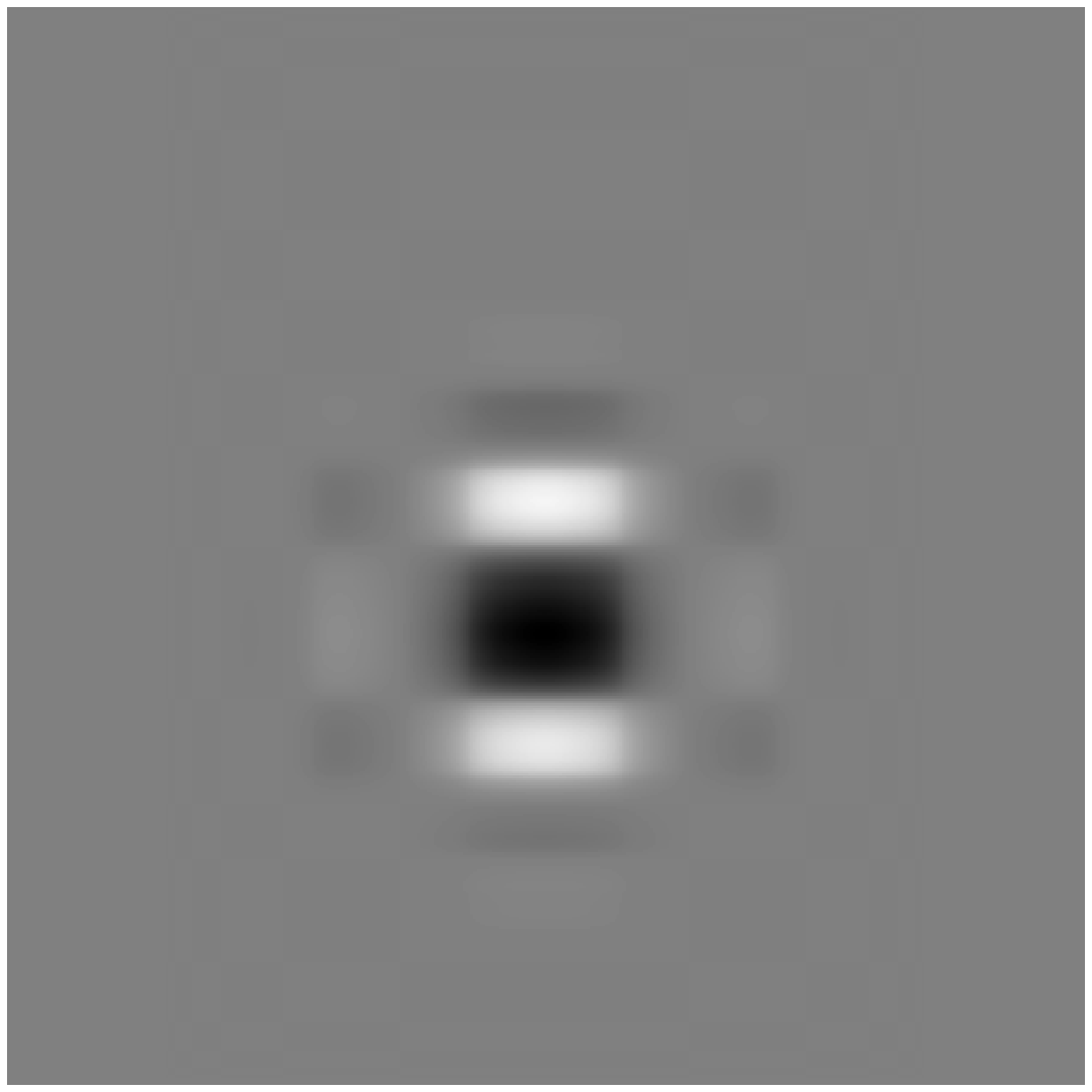}}
\subfigure{
\includegraphics[width=0.7in,height=0.7in]
{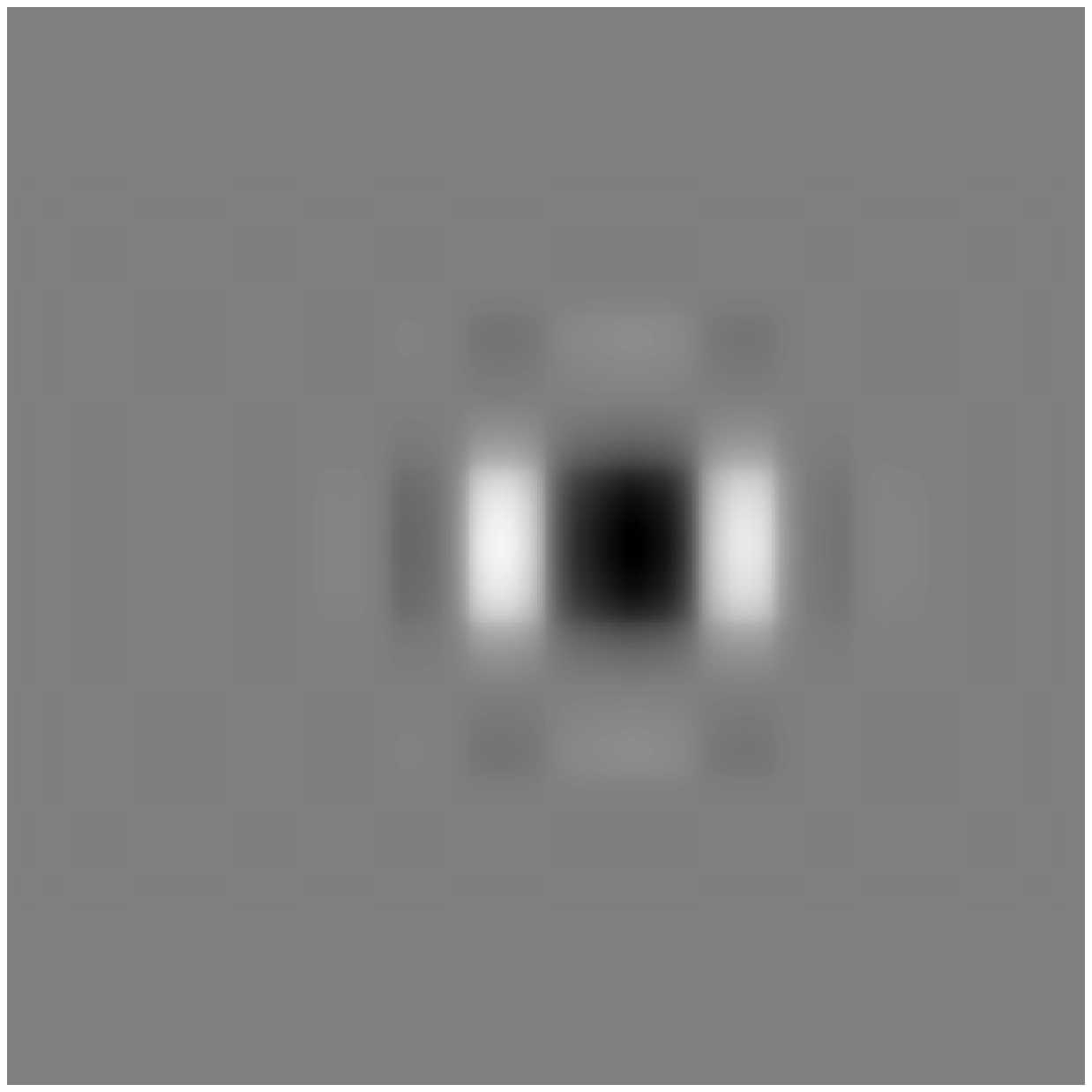}}
\subfigure{
\includegraphics[width=0.7in,height=0.7in]
{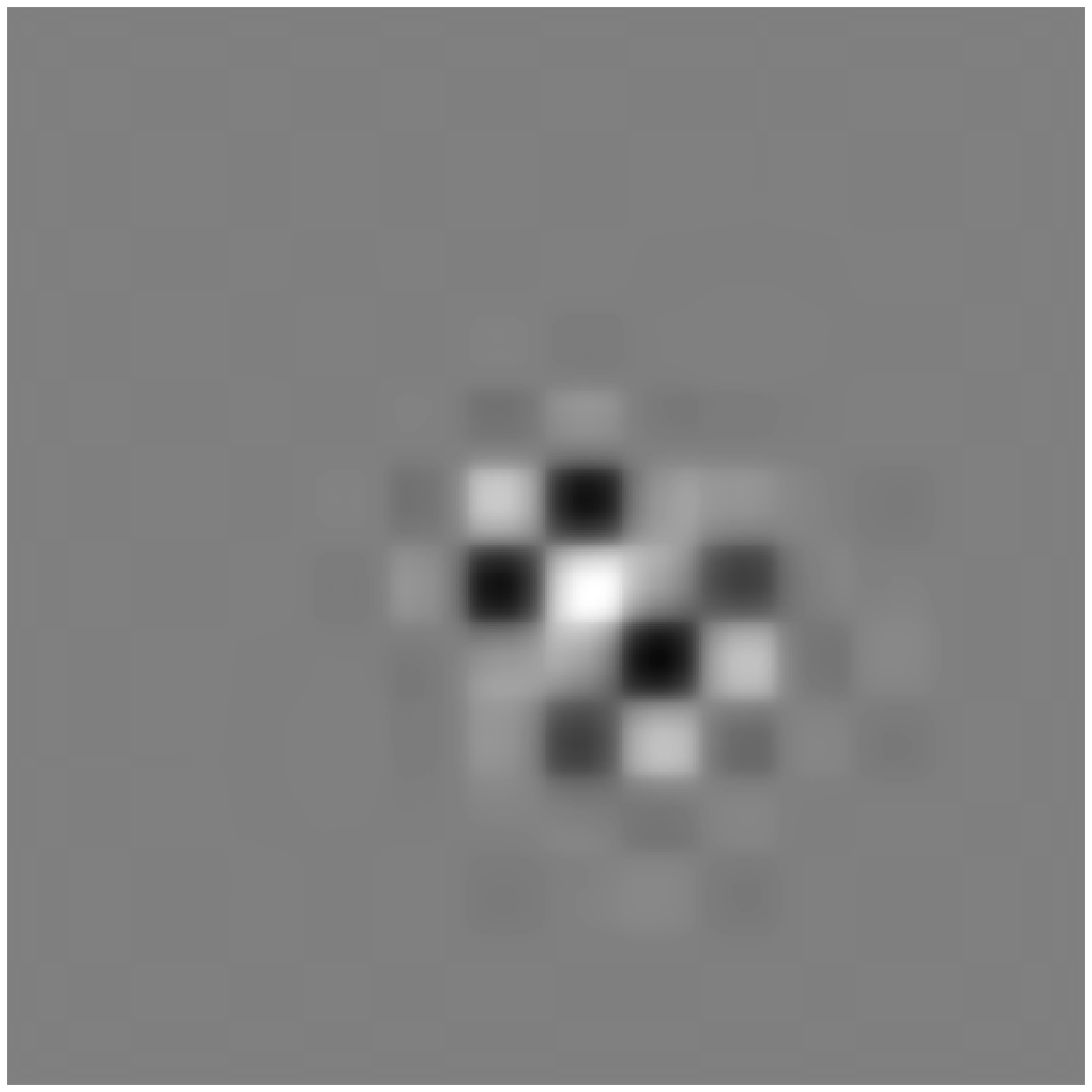}}
\subfigure{
\includegraphics[width=0.7in,height=0.7in]
{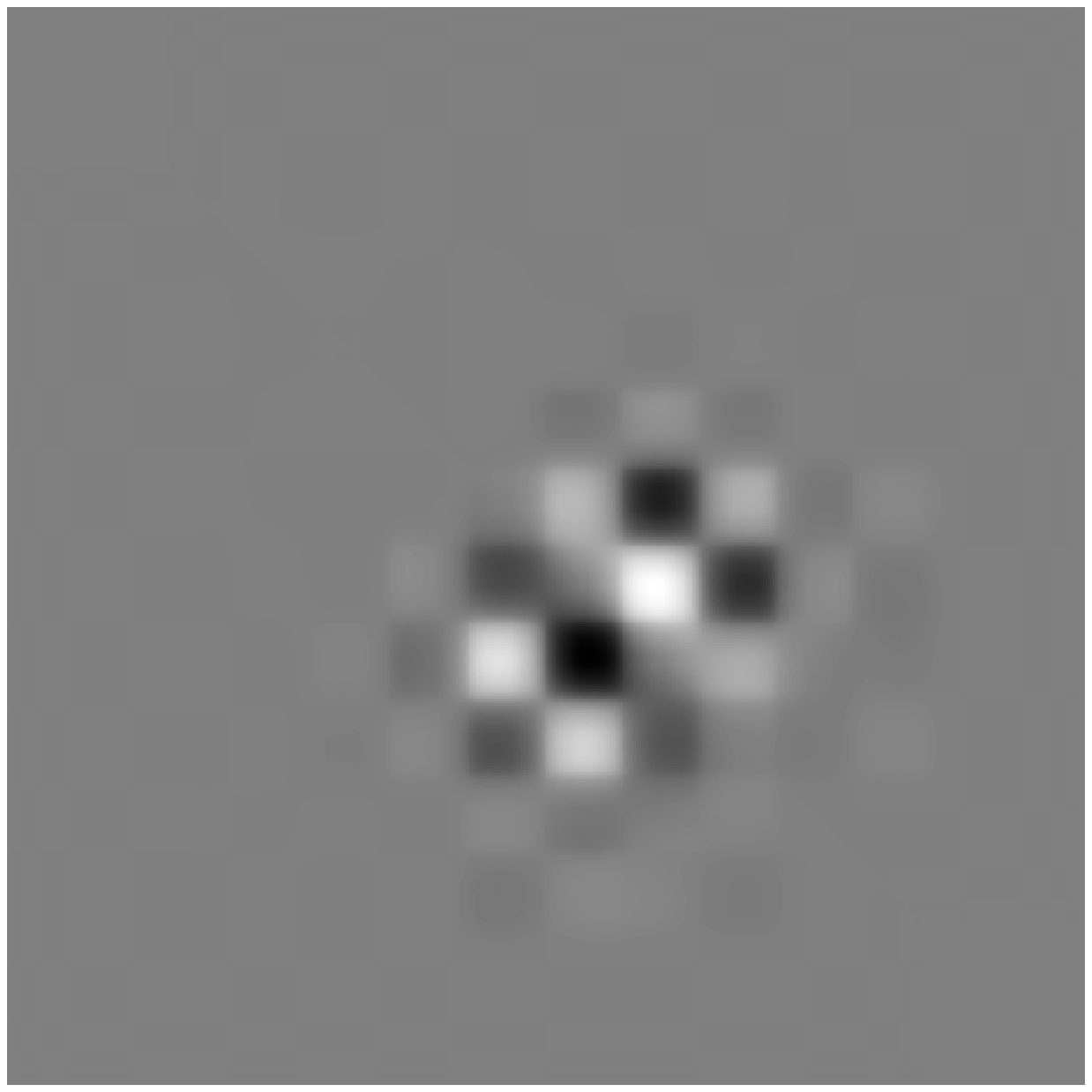}}
\begin{caption}{
The first row is for the real part and the second row is for the imaginary part of the tight framelet generators in Example~\ref{ex4} with $N=2$.
The third row is the greyscale image of the eight generators:
the first four for real part and the last four for imaginary part.
} \label{fig4}
\end{caption}
\end{center}
\end{figure}

\end{document}